%% file: main.tex
\documentclass[12pt]{article}%
\usepackage{amsfonts,array}
\usepackage{multicol,subcaption,makecell}
\usepackage[font=footnotesize,labelfont=bf]{caption}
\usepackage{afterpage}
\usepackage{amsmath,bbm,dsfont,mathrsfs,mathtools,appendix}
\usepackage{amssymb}
\usepackage{graphicx}
\usepackage{fullpage}%
\usepackage{authblk}

\usepackage{framed}
\usepackage{bbm}
\usepackage[table]{xcolor}
\usepackage{tikz}
\definecolor{shadecolor}{gray}{0.9}

\usepackage[colorlinks=true,linkcolor=blue,citecolor=green,plainpages=false,pdfpagelabels]%
{hyperref}%
\hypersetup{
	colorlinks,
	linkcolor={blue!60!green},
	citecolor={green!50!yellow!75!black},
	urlcolor={blue!80!black},
	linktoc=all
}
\providecommand{\U}[1]{\protect\rule{.1in}{.1in}}

\newtheorem{theorem}{Theorem}

\newtheorem{corollary}[theorem]{Corollary}

\newtheorem{definition}[theorem]{Definition}
\newtheorem{example}[theorem]{Example}

\newtheorem{lemma}[theorem]{Lemma}

\newtheorem{question}[theorem]{Question}
\newtheorem{proposition}[theorem]{Proposition}
	
\newtheorem{remark}[theorem]{Remark}

\newenvironment{proof}[1][Proof]{\noindent\textbf{#1.} }{\ \rule{0.5em}{0.5em}}

\newcommand{\dom}{\operatorname{dom}}
\newcommand{\im}{\operatorname{im}}
\newcommand{\Hil}{\mathcal{H}}
\newcommand{\Kil}{\mathcal{K}}
\newcommand{\Gil}{\mathcal{G}}
\newcommand{\M}[1]{\mathbb{M}_{#1}(\mathbb{C})}
\newcommand{\Mr}[1]{\mathbb{M}_{#1}(\mathbb{R})}
\newcommand{\B}[1]{\mathcal{L}({#1})}
\newcommand{\iden}{\mathbbm{1}}
\newcommand{\Tcal}{\mathcal T}
\newcommand{\comment}[1]{}

\newcommand{\cF}{{\mathcal{F}}}

\newcommand{\id}{{\rm{id}}} %identity

\newcommand{\cH}{\mathcal{H}}
\newcommand{\R}{\mathbbm{R}}
\newcommand{\C}{\mathbb{C}}

\newcommand{\N}{\mathbb{N}}
\newcommand{\cL}{\mathcal{L}}
\newcommand{\nn}{\nonumber}

\newcommand{\overbar}[1]{\mkern 1.5mu\overline{\mkern-1.5mu#1\mkern-1.5mu}\mkern 1.5mu}

\newcommand{\scA}{\mathscr{A}}
\newcommand{\scB}{\mathscr{B}}
\newcommand{\scC}{\mathscr{C}}
\newcommand{\scD}{\mathscr{D}}
\newcommand{\scG}{\mathscr{G}}
\newcommand{\scT}{\mathscr{T}}

\newcommand{\cP}{\mathcal{P}}
\newcommand{\cV}{\mathcal{V}}

\newcommand{\cR}{\mathcal{R}}
\newcommand{\cE}{\mathcal{E}}

\newcommand{\1}{\mathbbm{1}}
\newcommand{\cD}{\mathcal{D}}

\def\>{{\rangle}}
\def\<{{\langle}}
\newcommand{\be}{\begin{equation}}
	\newcommand{\ee}{\end{equation}}
\newcommand{\bea}{\begin{eqnarray}}
	\newcommand{\eea}{\end{eqnarray}}

\newcommand{\floor}[1]{\left\lfloor {#1} \right\rfloor}

\newcommand{\ket}[1]{|#1\rangle} %ket
\newcommand{\bra}[1]{\langle#1|} %bra
\newcommand{\kb}[1]{|#1\rangle\!\langle#1|} %ketbra
\newcommand{\Tr}{\mathrm{Tr}}
\newcommand{\norm}[1]{\left\lVert #1 \right\rVert}
\newcommand{\abs}[1]{\left\lvert #1 \right\rvert}

\newcommand{\spa}{{\rm span}}

\newcommand{\supp}{\operatorname{supp}}

\def\cK{\mathcal{K}}

\newcommand{\eqdef}{\coloneqq}

\numberwithin{equation}{section}
\numberwithin{theorem}{section}

\usepackage{color}
\definecolor{colorthree}{rgb}{0.01,0.51,0.93}

\usepackage{setspace}

\allowdisplaybreaks
\title{Transmitting algebras through quantum channels}

\author[1,2]{Robert Salzmann}
\author[3,4]{Satvik Singh}
\affil[1]{Institute for Quantum Information, RWTH Aachen University, Germany}
\affil[2]{Univ Lyon, Inria, ENS Lyon, UCBL, LIP, F-69342, Lyon Cedex 07, France}
\affil[3]{Department of Mathematics, Technical University of Munich, Garching, Germany}
\affil[4]{Munich Center for Quantum Science and Technology (MCQST), Munich, Germany}

\date{}

\begin{document}

\maketitle

\begin{abstract}
We develop a theory of exact transmission of finite-dimensional $C^*$-algebras through quantum channels. These algebras describe hybrid classical-quantum information, allowing the dimension of the quantum system to depend on the classical message. Allowing arbitrary encodings and decodings yields a transmission set of algebra types, ordered by embedding, that unifies zero-error information theory with operator-algebraic error correction. We ask if this set admits a dominating algebra into which every transmittable algebra embeds. Our central finding is that domination can fail even in small dimensions. In its absence, several incomparable maximal algebras can describe different optimal uses of the same channel, forcing the user to select one depending on the operational task and the type of information to be preserved.

We introduce hybrid capacities that reconstruct the only possible dominating algebra type and present a complete finite characterization of domination using minimal forbidden algebra types. We identify dominating algebras for channels whose operator systems are graph-isomorphic to $*$-algebras, including highly divisible channels, and for channels with zero one-shot zero-error quantum capacity. Under tensor products, we prove that joint coding can produce new algebra types, giving an algebraic analogue of superadditivity from Shannon theory. Domination can fail for the joint use of two channels, even when both channels separately admit dominating algebras. Finally, we construct a channel whose $n$-fold tensor powers have doubly exponentially many maximal algebra types, attaining the largest possible growth scaling for finite-dimensional channels. Consequently, new transmittable algebra types appear at arbitrarily large block-lengths for this channel, so its full transmission structure cannot be generated from any finite collection of bounded-block-length codes.
\end{abstract}

\newpage

\tableofcontents

\section{Introduction}
\label{sec:introduction}

A fundamental problem in information theory is to determine how much information can be transmitted through a given \emph{noisy} channel \emph{perfectly} without error \cite{Shannon1956zero, Korner1998zero, Duan2013noncomm}. In the
classical setting, this amounts to asking for the size of the largest set of input symbols that cannot be
confused at the output, i.e. the independence number of the so-called \emph{confusability graph} of the channel \cite{Shannon1956zero}. In the quantum setting, one may ask for the largest collection of
perfectly distinguishable classical messages, or the largest quantum system that
can be transmitted exactly through a given quantum channel. These
questions lead, respectively, to the classical and quantum zero-error capacities of the channel. Similar to the classical setting, these admit equivalent formulations in terms of the independence numbers of the so-called \emph{noncommutative confusability graph} of the quantum channel
\cite{Duan2009zerosuper,Duan2013noncomm}.

These familiar capacities, however, record only particular scalar aspects of a more structured transmission problem. A physical memory need not contain information that is purely classical or purely quantum. It may instead
consist of a classical label together with a quantum system whose dimension
depends on that label. The natural mathematical object describing such hybrid
classical-quantum information is a finite-dimensional $C^*$-algebra \cite{Kuperberg2003hybrid}:
\begin{equation}
    \scA\cong\bigoplus_{k=1}^K \M{d_k}.
    \label{eq:intro-algebra}
\end{equation}

The center of $\scA$ stores the classical label $k$, while the simple summand $\M{d_k}$ stores a $d_k$-dimensional quantum system conditioned on that label. In particular, the classical and quantum zero-error capacities can be retrieved by only considering purely classical algebras $\scA \cong \C^m$ and purely quantum algebras $\scA \cong \M{d}$, respectively.

This motivates the central question that is addressed in this work.

\begin{shaded}
\begin{question}\label{ques:main}
        Which finite-dimensional \(C^*\)-algebras 
    \[     \scA\cong\bigoplus_{k=1}^K \M{d_k} \]    
    can be transmitted
    exactly through a given noisy quantum channel $\Phi:\B{\Hil} \to \B{\Kil}$? Does every quantum channel admit a unique `optimal' transmittable algebra into which every other transmittable algebra can be embedded?
\end{question}    
\end{shaded}

\subsection{Transmittable algebras}

In order to answer this question, we have to be precise about what exactly it means for an algebra to be \emph{transmitted} via a noisy channel $\Phi:\B{\Hil}\to \B{\Kil}$. We first briefly contrast this with the well-studied notion of \emph{correctability} from quantum error-correction. The original Knill--Laflamme
conditions characterize when quantum information encoded in an input subspace $C\subseteq \Hil$ can be corrected \cite{Knill1997err, Knill2000error}. Operator-algebra quantum error
correction extends this framework to subsystem codes and, more generally, to
algebraic codes \cite{kribs2006error, Beny2007opalg-ecc}. In this formulation, a concrete algebraic code
$\scA\subseteq\B{\Hil}$ is said to be correctable for a channel
$\Phi:\B{\Hil}\to\B{\Kil}$ when information stored in that
algebra can be recovered after the noise acts by applying a suitable decoder. Equivalently,
if $\Phi(\cdot)=\sum_i K_i(\cdot)K_i^*$ is the Kraus representation of the channel \cite{watrous_theory_2018} and $P$ is the unit projection of
$\scA$, then the operator-algebraic Knill--Laflamme conditions \cite[Theorem 2]{Beny2007opalg-ecc} require
\begin{equation}
    [X,PK_i^*K_jP]=0
    \qquad
    \text{for all }X\in\scA\text{ and all }i,j.
    \label{eq:intro-algebra-KL}
\end{equation}

However, instead of fixing a concrete code
algebra inside the input space and asking whether it is correctable, Question~\ref{ques:main} considers
the abstract $C^*$-algebra itself as the information object to be transmitted. Thus, we introduce the following central definition.

\begin{definition}
A finite-dimensional $C^*$-algebra $\scA \cong \oplus_k \M{d_k}$ is
said to be \emph{transmittable} via a quantum channel $\Phi: \B{\Hil}\to \B{\Kil}$ if it admits a faithful encoding
representation $\pi_E:\scA\longrightarrow\B{\Hil}$ whose image is correctable for $\Phi$ in the sense of \eqref{eq:intro-algebra-KL}.
\end{definition}

Passing to the abstract isomorphism
type removes information about the particular location of the code and the chosen encoder/decoder. Hence, two differently
embedded but $*$-isomorphic correctable algebras represent the same
transmitted information type. There is an equivalent formulation which makes the connection with quantum
Shannon theory more transparent. Let
$\scA=\oplus_k\M{d_k}$ be represented in the standard way on the Hilbert space $\oplus_k\C^{d_k}$, and let $\cP_{\scA}(X)=\sum_k P_kXP_k$ be the trace-preserving conditional expectation onto $\scA$. Then, we show that $\scA$ is transmittable via $\Phi:\B{\Hil}\to \B{\Kil}$ if and only if there exist encoding and
decoding channels $\cE$ and $\cD$ such that (Lemma~\ref{lemma:EncDecTransmittableAlgebra})
\begin{equation}
    \cD\circ\Phi\circ\cE=\cP_{\scA}.
    \label{eq:intro-operational-transmission}
\end{equation}
Thus, transmittability of algebras is an exact channel-simulation problem. The novelty is that the target noiseless channel is allowed
to be an arbitrary finite-dimensional conditional expectation rather than
only a classical or a fully quantum identity channel.

Another equivalent formulation follows from the Knill--Laflamme conditions
(Lemma~\ref{lemma:pairwise-S-orthogonal}). Let
$S_\Phi:=\operatorname{span}_{i,j}\{K_i^*K_j\}$ be the noncommutative
confusability graph of $\Phi$ \cite{Duan2013noncomm}. Then, $\scA\cong\bigoplus_{k}\M{d_k}$
is transmittable via $\Phi$ if and only if there exist subspaces
$C_1,\ldots,C_K\subseteq\Hil$ with $\dim C_k=d_k$ whose orthogonal
projections $P_k$ satisfy
    \begin{align}
        P_kS_\Phi P_k&=\C P_k \qquad\text{for every }k,   \label{eq:intro-code-orthogonality-1}\\
        P_kS_\Phi P_l&=0 \qquad\quad \,\, \text{for }k\neq l   \label{eq:intro-code-orthogonality-2}.
    \end{align}
The first condition \eqref{eq:intro-code-orthogonality-1} is an equivalent description of the standard Knill-Laflamme conditions for subspace codes $C_k$ \cite{Knill2000error}.
The second condition \eqref{eq:intro-code-orthogonality-2} says that channel outputs from different code spaces have orthogonal
supports. Thus, the receiver can identify the classical block label
without disturbing the encoded quantum information and then apply the
corresponding quantum recovery. Together, these conditions describe
exactly how the classical and quantum parts of a transmittable algebra
are preserved.

\subsection{Ordering}

Once the notion of transmittable algebras is fixed, it is natural to ask if, for a given channel $\Phi$, there exists an `optimal' or `largest' algebra $\scA$ which can be transmitted through $\Phi$. In order to answer this question, we order the set of all transmittable algebras by \emph{embedding}. 

\begin{definition}[Embedding preorder]
    An algebra $\scA \cong \oplus_{k=1}^K \M{d_k}$ \emph{embeds} into $\scB \cong \oplus_{l=1}^L \M{n_l}$, denoted $\scA \leq \scB$, if the $(d_k)_k$ blocks of $\scA$ can be (faithfully) packed inside the $(n_l)_l$ blocks of $\scB$, meaning that there exists a $K\times L$ non-negative integer matrix $\Lambda$ satisfying
\begin{align}
  \forall l&: \quad  \sum_k \Lambda_{kl}d_k\leq n_l \label{eq:Lambda-l} \\ 
  \forall k&: \quad\sum_l \Lambda_{kl} \geq 1. \label{eq:Lambda-k}
\end{align}
\end{definition}

The matrix $\Lambda$ records how many copies of each $\M{d_k}$ occur in each
$\M{n_l}$. Here, \eqref{eq:Lambda-l} is the packing constraint, while
\eqref{eq:Lambda-k} guarantees that no source block is discarded. Hence,
if $\scA\leq\scB$, then $\scB$ is atleast as useful as $\scA$ for storing information. The \emph{transmission set}
\begin{equation}
    \scT(\Phi):=\{\scA:\scA\text{ is transmittable via }\Phi\}
\end{equation}
is a lower set with respect to this preorder, i.e., if
$\scB\in\scT(\Phi)$ and $\scA\leq\scB$, then $\scA\in\scT(\Phi)$. For purely classical and quantum algebras, this preorder reduces to a total order:
\begin{equation}
    \C^m\leq\C^{m'}
    \iff
    m\leq m',
    \qquad
    \M{d}\leq\M{d'}
    \iff
    d\leq d',
    \label{eq:intro-total-orders}
\end{equation}
which is precisely what allows the corresponding zero-error capacities of a channel to be defined as the size of the largest classical and quantum algebras that it can transmit:
\begin{align}
    C_0(\Phi) &:= \log \max \{m : \C^m \text{ is transmittable via } \Phi \}, \\
    Q_0(\Phi) &:= \log \max \{d : \M{d} \text{ is transmittable via } \Phi \}.
\end{align}

However, the embedding preorder on general finite-dimensional $C^*$-algebras is \emph{not} a total order. The simplest example is provided by the algebras $\M{2}$ and $\C^3$. $\M{2}$ cannot embed into $\C^3$, since a noncommutative algebra cannot embed into a commutative one. Conversely, $\C^3$ cannot embed into $\M{2}$, since this would require three nonzero mutually orthogonal projections in a two-dimensional space. The two algebras therefore represent genuinely incomparable resources: $\M{2}$ stores a coherent qubit, whereas $\C^3$ stores three perfectly distinguishable classical messages. 

This leads to two different notions of optimality (see Definition~\ref{def:max-dom}). A transmittable algebra is \emph{maximal} if it cannot be embedded into any strictly larger transmittable algebra. We collect all maximal transmittable algebras for a given channel $\Phi$ in the set 
\begin{equation}
    \max (\scT(\Phi)) := \{ \scA \in \scT(\Phi) : \scA \text{ is maximal in } \scT(\Phi) \text{ with respect to} \leq  \}.
\end{equation}
The elements in $\max (\scT(\Phi))$ represent different optimal use cases of the noisy channel $\Phi$ for zero-error hybrid classical-quantum communication. On the other hand, a transmittable algebra is \emph{dominating}, if every other transmittable algebra can be embedded inside it. This equivalent to saying that, upto $*$-isomorphism, there is only one element in the set $\max (\scT(\Phi))$, which is precisely the dominating algebra. Such an algebra, if it exists, would provide a single universally optimal algebra type associated with the channel, see Figure~\ref{fig:max-vs-dom}.

\medskip

We can thus rephrase the second part of Question~\ref{ques:main} as follows.

\begin{shaded}
\vspace{-10pt}
\begin{question}
Does every noisy channel admit a dominating transmittable algebra?
\end{question}
\vspace{-10pt}
\end{shaded}

One of the central contributions of this work is a \emph{negative}
answer to this question. Consider the channel $\Phi:\M{4}\to\M{3}$
defined by $\Phi(\cdot)=\sum_{i=1}^3E_i (\cdot)E_i^*$, where
\begin{equation}
E_1=\ket{f_0}\bra{e_0}+\ket{f_1}\bra{e_1},
\qquad
E_2=\ket{\phi_+}\bra{e_2},
\qquad
E_3=\ket{\phi_-}\bra{e_3},
\end{equation}
$\ket{\phi_\pm}=(\ket{f_1}\pm\ket{f_2})/\sqrt2$, and
$\{\ket{e_i}\}_{i=0}^3$ and $\{\ket{f_j}\}_{j=0}^2$ are orthonormal
input and output bases. These Kraus operators are \emph{partial
isometries}: each acts isometrically on its own input sector, while
coherences between different sectors are discarded.

A qubit encoded in
$C=\operatorname{span}\{\ket{e_0},\ket{e_1}\}$ passes through
unchanged, up to relabeling the basis. Similarly, the inputs
$\ket{e_0},\ket{e_2},\ket{e_3}$ give the orthogonal outputs
$\ket{f_0},\ket{\phi_+},\ket{\phi_-}$, allowing three classical
messages to be transmitted perfectly. Thus, $\M{2}$ and $\C^3$ are transmittable via $\Phi$. The obstruction to combining these uses to transmit $\M{2}\oplus \C$ is visible at
the output, since both $\ket{\phi_+}$ and $\ket{\phi_-}$ have nonzero
overlap with $\ket{f_1}$, which is part of the qubit's output space.
Every input supported on $C^\perp$ produces a mixture of $\ket{\phi_\pm}$, so none supplies a classical alternative distinguishable
from the entire qubit code. Moreover, the rank-one branches $E_2,E_3$
cannot preserve a qubit; every perfect qubit code must avoid them
and therefore lie in $C$. Thus, the only maximal transmittable algebras
are $\M{2}$ and $\C^3$, with no transmittable algebra dominating both
(see Example~\ref{ex:PedagogicalExampleWithoutDomAlgebra} for details).

For comparison, consider the block-dephasing channel
\begin{equation}
\Psi:\M{3}\to\M{3},
\qquad
\Psi(X)=PXP+P_\perp XP_\perp,
\end{equation}
where $P$ and $P_\perp$ are complementary orthogonal projections of
ranks $2$ and $1$. Here the two sectors remain orthogonal at the
output. The channel preserves an arbitrary qubit in the $P$-sector
and keeps the remaining one-dimensional sector perfectly
distinguishable from it. Consequently, $\M{2}\oplus\C$ is itself
transmittable and contains both $\M{2}$ and $\C^3$. It is the
dominating transmittable algebra of $\Psi$
(Theorem~\ref{thm:alg-max}). Figure~\ref{fig:max-vs-dom} shows the
transmission sets of $\Psi$ in the left panel and $\Phi$ in the right
panel.

\begin{figure}[htbp] 
    \centering
    \tikzset{
        intro hasse/.style={
            x=1cm, y=1cm, font=\small, line width=0.55pt,
            every node/.style={
                inner xsep=4pt, inner ysep=3pt, outer sep=1pt,
                text height=1.8ex, text depth=0.5ex
            }
        },
        intro maximal/.style={
            draw, rounded corners=2pt, fill=black!4
        }
    }
    \begin{subfigure}[b]{0.48\textwidth}
        \centering
        \begin{tikzpicture}[intro hasse]
            % Matching bounds keep the shared algebras aligned.
            \path[use as bounding box] (-2.65,-0.35) rectangle (2.65,3.65);
            \node (c1) at (0,0) {$\C$};
            \node (c2) at (0,1.1) {$\C^2$};
            \node (m2) at (-1.3,2.2) {$\M{2}$};
            \node (c3) at (1.3,2.2) {$\C^3$};
            \node[intro maximal] (top) at (0,3.3) {$\M{2}\oplus\C$};
            \draw (c1) -- (c2);
            \draw (c2) -- (m2) -- (top) -- (c3) -- (c2);
        \end{tikzpicture}
        \caption{A dominating algebra exists.}
        \label{fig:intro-hasse-dominating}
    \end{subfigure}
    \hfill
    \begin{subfigure}[b]{0.48\textwidth}
        \centering
        \begin{tikzpicture}[intro hasse]
            \path[use as bounding box] (-2.65,-0.35) rectangle (2.65,3.65);
            \node (c1) at (0,0) {$\C$};
            \node (c2) at (0,1.1) {$\C^2$};
            \node[intro maximal] (m2) at (-1.3,2.2) {$\M{2}$};
            \node[intro maximal] (c3) at (1.3,2.2) {$\C^3$};
            \draw (c1) -- (c2);
            \draw (m2) -- (c2) -- (c3);
        \end{tikzpicture}
        \caption{No dominating algebra exists.}
        \label{fig:intro-hasse-no-dominating}
    \end{subfigure}
    \caption{Hasse diagrams of two possible transmission sets, with nonzero algebras
    identified up to $*$-isomorphism. An upward path from $\scA$ to $\scB$
    means $\scA\leq\scB$; edges show only cover relations. Boxes mark maximal
    elements. In (a), the boxed algebra dominates every transmittable algebra.
    In (b), the two boxed algebras have no common upper bound in the
    transmission set.}
\label{fig:max-vs-dom}
\end{figure}
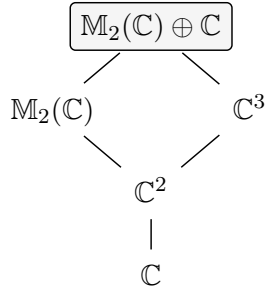
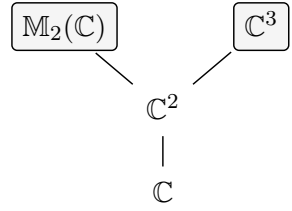

\comment{
This question also connects with the simultaneous classical and quantum
transmission problem studied by Devetak and Shor \cite{DevetakShor2005}.
They characterize the pairs of classical and quantum rates achievable over
many channel uses with asymptotically vanishing error. In our notation, transmitting
$\M{d}^{\oplus c}\cong\C^c\otimes\M{d}$ means sending one of $c$ classical
messages together with an arbitrary $d$-dimensional quantum state. Requiring
exact recovery in a single channel use gives the one-shot zero-error version
of this task, with (classical, quantum) rate pair $(\log c,\log d)$; the hybrid capacities
$c_d(\Phi)$ introduced in Section~\ref{sec:hybrid-capacities} describe the corresponding trade-off. The distinction is that our framework retains more information than just a pair of rates, since general
algebras $\oplus_k\M{d_k}$ also allow the quantum dimension to depend on the
classical message. This makes the “non-existence of dominating algebras” result sharper than the ordinary existence of a rate trade-off in the sense of \cite{DevetakShor2005}.
}

These observations motivate the structural questions studied in this paper. A dominating algebra, when it exists, implies that there is a single best way to use the channel for zero-error communication. Hence, it is natural to ask which properties of the noise guarantee its existence. Moreover, since communication protocols combine channel
uses, it is equally natural to ask whether existence of dominating algebras remains
stable under tensor products, or whether joint encoding introduces new,
incomparable information types. For memoryless uses of a fixed channel $\Phi$, the question
concerns the sequence $\{ \scT(\Phi^{\otimes n})\}_{n\in \N}$: can all its maximal algebras
be generated, under tensor products and embeddings, from a finite collection
of algebra types available at bounded block-lengths, or must new types keep
appearing infinitely often? The number of maximal algebras at each block-length measures the
diversity of optimal hybrid codes and allows us to test whether their
structure admits such a finite description.

In this work, we answer some of the above questions and develop the basic theory of the transmission set $\scT(\Phi)$ for general quantum channels $\Phi$, focusing in particular
when it can, and cannot, be represented by a single dominating algebra. In the next subsection, we summarize our main findings.

\subsection{Main results}

\begin{itemize}

    \item \emph{Nonexistence of dominating algebras.} As already stated,
    a quantum channel need not have a dominating transmittable algebra. In particular, we construct
    explicit channels $\Phi$ in small dimensions for which
    \begin{equation}
        \max(\scT(\Phi))=\{\M{2},\C^3\},
    \end{equation}
    so that preserving a qubit and transmitting three classical messages are
    incomparable optimal uses of the same channel (see Examples~\ref{ex:PedagogicalExampleWithoutDomAlgebra} and \ref{ex:NoDomSecondExample}).

\item \emph{Hybrid capacities and the structure of dominating algebras.}
We introduce the hybrid classical-quantum capacity $c_d(\Phi)$, which is the maximum number of classical messages that can be sent alongside a $d$-dimensional quantum system via $\Phi$,
\begin{align}
c_d(\Phi):= \max\left\{c\in\N \, \Big|\ \M{d}^{\oplus c}\  \text{transmittable via $\Phi$} \ \right\}.
 \end{align} We use these capacities to construct the \emph{capacity algebra} (see Section~\ref{sec:hybrid-capacities}):
\begin{equation}
\scA_{\operatorname{cap}}(\Phi)
=
\bigoplus_{d=1}^{d_{\max}(\Phi)}
\M{d}^{\oplus N_d(\Phi)}.
\end{equation}
For every $d$, exactly $c_d(\Phi)$ copies of $\M{d}$ embeds into $\scA_{\operatorname{cap}}(\Phi)$, and this property uniquely determines the algebra (Theorem~\ref{thm:CapacitiesToAlgebra}). Furthermore, if a dominating algebra $\scA_{\dom}(\Phi)$ exists, the recursively determined multiplicities $N_d(\Phi)$ are non-negative and $\scA_{\dom}(\Phi)\cong\scA_{\text{cap}}(\Phi)$
    (Corollary~\ref{cor:CapacitiesToAlgebra}). However, the existence and transmission of $\scA_{\rm cap}(\Phi)$ do not suffice to guarantee domination (Remark~\ref{rem:capacity-without-domination}). We supply the missing extra conditions, which serve as a complete finite
    criterion for domination: the candidate $\scA_{\rm cap}(\Phi)$ must not only be well-defined and transmittable, but also none of its
    minimal forbidden algebra types must be transmittable (Theorem~\ref{thm:dominating-forbidden-types} and
    Corollary~\ref{cor:complete-capacity-criterion}).
    When $d_{\max}(\Phi)\leq3$, transmittability
    of the capacity algebra already guarantees domination (Corollary~\ref{cor:dominating-three}).

    \item \emph{Channels with a dominating algebra.}
    If the noncommutative graph of a channel $S_\Phi$ is a closed under matrix multiplication (upto graph isomorphism), then its commutant is the dominating
    transmittable algebra (Theorem~\ref{thm:alg-max} and Corollary~\ref{corollary:alg-max}):
    \begin{equation}
        \scA_{\dom}(\Phi)\cong S_\Phi'.
    \end{equation}
    This implies that every channel $\Psi:\B{\Hil}\to \B{\Hil}$, when self-iterated sufficiently many times: $\Psi^l:= \Psi \circ \Psi \circ \ldots \circ \Psi$ for $l\geq (\dim \Hil)^2$, admit a dominating transmittable algebra isomorphic to its peripheral algebra: $\scA_{\dom}(\Psi^l)\cong \mathscr{X}^*(\Psi)$ (Theorem~\ref{theorem:dom-highly-divisible}). 
    
    Similarly, every channel with
    $Q_0(\Phi)=0$ admit a purely classical dominating algebra $\scA_{\dom}(\Phi)\cong \C^{2^{C_0(\Phi)}}$ (Theorem~\ref{theorem:dom-classical}).
    
    Further
    examples show both that these sufficient conditions are not necessary, see Appendix~\ref{appen:SnotAlgebra} and Figure~\ref{fig:channel-classes} .
    
    \item \emph{Tensor products and new joint information types.}
    Product coding implies
    \begin{equation}
        \scT(\Phi)\star\scT(\Psi)
        \subseteq\scT(\Phi\otimes\Psi),
    \end{equation}
    where $\scT(\Phi)\star\scT(\Psi)$ consists of all algebras $\scC$ for
    which there exist $\scA\in\scT(\Phi)$ and $\scB\in\scT(\Psi)$ with
    $\scC\leq\scA\otimes\scB$. Thus, $\star$ collects the information types
    obtainable by coding separately for the two channels. We show that the inclusion
    can be strict: joint coding can transmit algebra types that no such
    product algebra can support. This is an algebra-valued analogue of the
    familiar notion of 
    superadditivity from quantum Shannon theory 
    \cite{Hastings2009super, Yard2008super, Chen2010zerosuper,Duan2009zerosuper}. In particular, we exhibit channels $\Phi$ and $\Psi$ which each
    possess a dominating transmittable algebra while $\Phi\otimes\Psi$ does
    not. Thus, the existence of a universally optimal algebra is not stable
    under tensor products.

    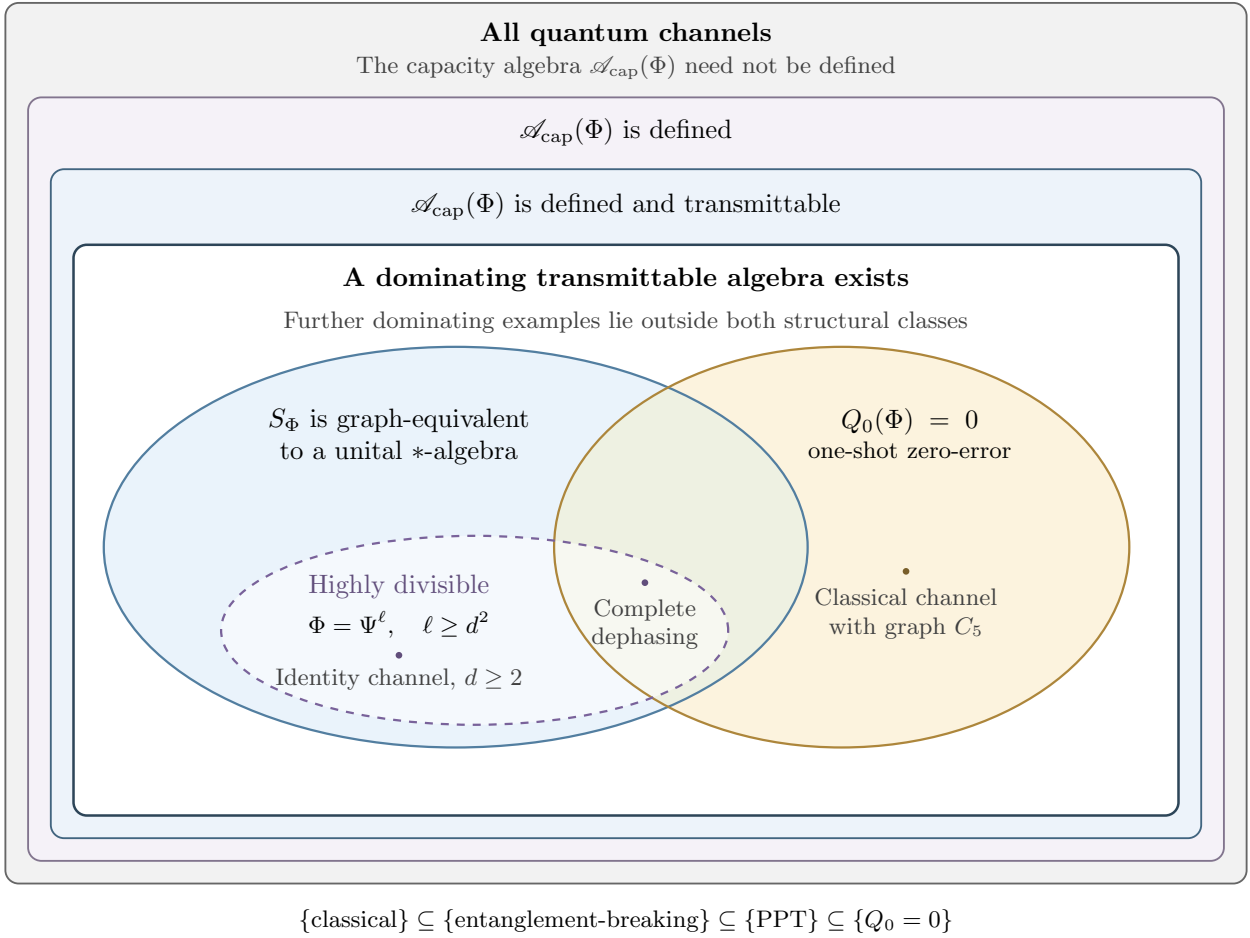
\begin{figure}[!ht]
  \centering
   \input{channel-classes-tikz.tex}
   \caption{Classes of channels admitting a capacity algebra and a
   dominating transmittable algebra. The two structural sufficient
   classes from Section~\ref{sec:algebraic-operator-systems} and \ref{sec:zero-quantum-capacity} overlap, and neither contains the other; their union is
   strictly smaller than the class admitting a dominating algebra.
   The bottom line records further sufficient classes without depicting
   their intersections with the algebraic class.}
   \label{fig:channel-classes}
 \end{figure}

    \item \emph{Many maximal algebras at large block-length.}
    We define the channel \emph{width}
    \begin{equation}
        w_n(\Phi):=\abs{\max(\scT(\Phi^{\otimes n}))}
    \end{equation}
    which counts the number of maximal transmittable algebras at block-length $n$ for memoryless use of $\Phi$. We show that finite generation of the tensor-power transmission sets $\scT(\Phi^{\otimes n})$ would
    force $w_n(\Phi)$ to grow at most polynomially in $n$
    (Lemma~\ref{lemma:finite-generation-polynomial}). 
    We then construct a fixed channel $\Psi$ with (Proposition~\ref{prop:DoubleExp})
    \begin{equation}
        w_n(\Psi)
        \gtrsim 2^{2^n}.
    \end{equation}
    Hence, the number of inequivalent
    optimal hybrid transmission modes can grow double exponentially with the
    number of channel uses, and new product-indecomposable algebra types must
    appear at arbitrarily large block-lengths. A universal count of
    algebra types gives $w_n(\Phi)\leq2^{d^n}-1$ for input dimension
    $d$ (Proposition~\ref{prop:width-upper-bound}), so double-exponential
    growth is optimal in scale.
\end{itemize}

\subsection{Outline}

The remainder of the paper is organized as follows. We first collect the
necessary background on finite-dimensional $C^*$-algebras, noncommutative
confusability graphs, conditional expectations, and operator-algebraic quantum
error correction in Section~\ref{sec:prelim}. Section~\ref{sec:TransmittableAlgebras} develops the basic theory of the transmission set $\scT(\Phi)$. We then study dominating
algebras, their reconstruction from hybrid capacities, and several classes of
channels for which they exist in Section~\ref{sec:dom}. Section~\ref{sec:tensor-products} studies transmittable algebras under
tensor products, showing that domination can fail even when it holds for each factor and that the number of maximal algebra types can grow doubly exponentially with the block-length. Finally, we conclude with a discussion in Section~\ref{sec:discussion}. Some technical proofs and the required
order-theoretic terminology are collected in the appendices.

\section{Preliminaries} \label{sec:prelim}

We denote quantum systems and their associated complex Hilbert spaces\footnote{All Hilbert spaces are assumed to be finite-dimensional throughout this work.} by capital letters $A,B,C$, with corresponding dimensions $d_A, d_B$ and $d_C$, respectively. The space of linear operators acting on a Hilbert space $\Hil$ is denoted by $\B{\Hil}$ and the convex set of quantum states or density operators (i.e.~positive semi-definite operators in $\B{\Hil}$ with unit trace) is denoted by $\cD (\Hil)$. For a unit vector $\ket{\psi}\in \Hil$, we denote $\ket{\psi}\bra{\psi}\in \cD (\Hil)$ by $\psi$. 

A quantum channel $\Phi:\B{A}\to \B{B}$ is a linear, completely positive, and trace-preserving map. The adjoint $\Phi^*: \B{B}\to \B{A}$ of a quantum channel $\Phi : \B{A}\to \B{B}$ is a linear, unital, and completely positive map defined via the following equation:
\begin{equation}
    \forall X \in \B{A}, \, \forall Y \in \B{B} : \quad \Tr(Y \Phi(X))= \Tr (\Phi^*(Y) X).
\end{equation}
Two channels $\Phi:\B{A}\to \B{B}$ and $\Phi_c : \B{A}\to \B{E}$ are said to be \emph{complementary} to each other if there exists an isometry $V:A\to B\otimes E$ such that 
\begin{equation}
    \Phi(\cdot) = \Tr_E \left( V(\cdot)V^{*} \right) \quad\text{and}\quad \Phi_c (\cdot) =  \Tr_B \left( V(\cdot) V^{*} \right).
\end{equation}

\subsection{$C^*$-algebras}
Let $\scA$ and $\scB$ be $C^*$-algebras\footnote{All $C^*$-algebras are assumed to be finite-dimensional throughout this work.}. A map $\phi:\scA\to\scB$ is called a \emph{$*$-homomorphism} if it is linear, multiplicative, i.e. $\phi(a_1a_2)=\phi(a_1)\phi(a_2),$ and $*$-preserving, i.e. $\phi(a^*) = \phi(a)^*.$
We say $\phi:\scA\to\scB$ is a \emph{$*$-isomorphism} if it is a bijective $*$-homomorphism and we write $\scA \cong \scB$ if such a $*$-isomormphism  exists. 
Every finite-dimensional $C^*$-algebra $\scA$ is $*$-isomorphic to a direct sum of matrix algebras
\begin{align}\label{eq:A-shape}
    \scA \cong \bigoplus_{k=1}^K \M{\lambda_k},
\end{align}
where $K, \lambda_1,\cdots, \lambda_K\in\N$ are positive integers \cite[Theorem 11.2]{Takesaki1979algebra}. In particular, every such $\scA$ contains a unique unit element $1_{\scA}\in\scA$. The \emph{shape} $\lambda=(\lambda_1,\ldots ,\lambda_K)$ of $\scA$, which we will always assume to be ordered in a non-increasing fashion $\lambda_1\geq \lambda_2\geq \ldots \geq \lambda_K$, is a complete invariant for the algebraic structure of $\scA$, i.e., if $\scB$ is another finite-dimensional $C^*$-algebra with shape $\mu=(\mu_1, \ldots \mu_L)$, then $\scA \cong \scB$ if and only if $\mu=\lambda$. 

A $*$-homomorphism $\pi: \scA\to \cL(\cH)$ for some Hilbert space $\cH$ is called \emph{representation of $\scA$.} An injective representation is called \emph{faithful.}

For a finite dimensional $C^*$-subalgebra $\scA\subseteq \cL(\cH)$, there exists an orthogonal direct sum decomposition of the underlying Hilbert space $\Hil=C \oplus C^{\perp}$ with $C=\oplus_k A_k \otimes B_k$ such that
\begin{align}
\label{eq:AlgebraUnitaryL(H)}
    \scA = (\oplus_k \B{A_k} \otimes 1_{B_k}) \oplus 0_{C^{\perp}},
\end{align}
where $\scA \cong \bigoplus_{k=1}^K \M{\lambda_k}$ and $d_{A_k}=\lambda_k$ for each $k$. The unit element $1_{\scA}=\oplus_k 1_{A_k} \otimes 1_{B_k} \oplus 0$ is exactly the orthogonal projection onto $C\subseteq \Hil$. 

\begin{definition}[Embedding preorder]\label{def:embedding-preorder}
For two $C^*$-algebras $\scA,\scB$, we say that $\scA$ embedds into $\scB$, denoted $\scA\leq \scB$, if there exists an injective $*$-homomorphism $\phi:\scA\to \scB$. We further denote by $\mathscr{A}< \mathscr{B}$ the fact that $\mathscr{A}\le \mathscr{B}$ and $\mathscr{A}\not\cong \mathscr{B}.$   
\end{definition}
Definition~\ref{def:embedding-preorder} defines a \emph{preorder} on the class of finite dimensional $C^*$-algebras, meaning that the binary relation $\leq$ is reflexive and transitive. Moreover, 
\begin{equation}\label{eq:pre-partial-order}
    \scA \leq \scB, \scB \leq \scA \implies \scA\cong \scB.
\end{equation}
Hence, the quotient space by $*$-isomorphism becomes partially ordered (see Appendix~\ref{appen:order-theory}). 

For finite-dimensional $C^*$-algebras
$\scA=\bigoplus_{k=1}^K\M{d_k}$ and
$\scB=\bigoplus_{l=1}^L\M{n_l}$, every injective
$*$-homomorphism $\phi:\scA\to\scB$ has the form
\begin{equation}
    \phi(X_1,\ldots,X_K)
    =\bigoplus_{l=1}^L U_l
      \left[\left(\bigoplus_{k=1}^K
        X_k\otimes1_{\Lambda_{kl}}\right)\oplus0_{r_l}\right]U_l^*,
    \qquad r_l=n_l-\sum_k\Lambda_{kl}d_k,
    \label{eq:embedding-block-form}
\end{equation}
where $U_l$ is unitary on $\C^{n_l}$ and $\Lambda$ is a $K\times L$
matrix of nonnegative integers satisfying
\begin{equation}
    \sum_k\Lambda_{kl}d_k\leq n_l\quad\text{for every }l,
    \qquad
    \sum_l\Lambda_{kl}\geq1\quad\text{for every }k.
\end{equation}
Terms with $\Lambda_{kl}=0$ are omitted, as is a zero block with
$r_l=0$. The matrix $\Lambda$ records the multiplicities of the source
blocks in each target block; the second condition expresses
faithfulness. The homomorphism is unital precisely when $r_l=0$ for
every $l$. In particular, the embedding preorder permits
$\M{2}\leq\M{3}$ via $X\mapsto X\oplus0$. Equivalently, we can reformulate embedding in terms of a packing problem:
\begin{equation}
    \bigoplus_{k=1}^K\M{d_k}\leq\bigoplus_{l=1}^L\M{n_l}
    \iff
    \exists f:[K]\to[L]\ \text{such that}\
    \sum_{k:f(k)=l} d_k\leq n_l \quad\forall l.
    \label{eq:algebra-packing-assignment}
\end{equation}
Determining if such a packing exists for given algebras $\scA, \scB$ is NP-hard \cite{Korte2018}.

\begin{lemma}[Faithful representation dimension]
\label{lemma:finite-algebra-types}
For $\scA\cong\bigoplus_{k=1}^K\M{d_k}$ and $D\in\N$,
\begin{equation*}
    \scA\leq\B{\C^D}
    \quad\Longleftrightarrow\quad
    \sum_{k=1}^K d_k\leq D.
\end{equation*}
Consequently, there are only finitely many isomorphism classes of
$C^*$-algebras that admit a faithful representation on a fixed
finite-dimensional Hilbert space.
\end{lemma}
\begin{proof}
In a faithful representation of $\scA \cong \oplus_k \M{d_k}$, every simple block occurs at least once,
so its Hilbert-space dimension is at least $\sum_k d_k$. Conversely,
one copy of each block, with an unused zero block if necessary, gives
a faithful representation in every dimension $D\geq\sum_kd_k$.
The possible shapes are integer partitions of integers between $1$
and $D$, of which there are finitely many.
\end{proof}

\subsection{Confusability graphs and operator systems}

For a positive integer $n\in \N$, we collect positive integers until $n$ in $[n]:= \{1,2,\ldots ,n \}$.

Consider a classical noisy channel given by a (column) stochastic matrix $N\in \mathbb{M}_{d_B\times d_A}(\R)$, with elements $N(j|i)$ expressing the probability of observing the symbol $j\in [d_B]$ given that the symbol $i\in [d_A]$ was sent. In order to send messages $m$ via $N$ perfectly, they should be encoded in symbols $i_m\in[d_A]$ in such a way that the corresponding output probability distributions on $[d_B]$ have disjoint supports for $m\neq m'$. The {\em{confusability graph}} $G_N$ of $N$ captures this property; it has vertex set $V=[d_A]$ and edges between $i, i'\in V$ that are \emph{confusable} (denoted $i\sim_N i'$), i.e.~for which $\exists  j\in [d_B]$ such that $N(j|i)N(j|i') >0$. Then, $G_N$ is an undirected graph with a loop at each vertex. The maximum number of messages that can be sent via $N$ with zero-error is simply the size of the largest independent set of $G_N$, i.e. the independence number $\alpha(G_N)$ \cite{Shannon1956zero}.

Similarly, in order to send classical messages $m$ through a quantum channel $\Phi:\B{\Hil_A}\to \B{\Hil_B}$, they should be encoded in states $\psi_m\in \cD (\Hil_A)$ such that the corresponding outputs have orthogonal supports:
\begin{align*}
\forall m\neq m': \quad  0 &=  \Tr (\Phi(\psi_m) \Phi(\psi_{m'}) ) \\
&= \Tr \left[ \sum_i K_i \ket{\psi_m}\bra{\psi_m} K^{*}_i \sum_j K_j \ket{\psi_{m'}}\bra{\psi_{m'}} K^{*}_j \right] \\
        &= \sum_{i,j} |\bra{\psi_m} K^{*}_i K_j \ket{\psi_{m'}} |^2.
\end{align*}

In other words, 
\begin{equation}
    \forall m\neq m': \qquad \ket{\psi_m}\bra{\psi_{m'}} \perp \operatorname{span}\{K^{*}_i K_j \}.
\end{equation}

This motivates the notion of the \emph{noncommutative} confusability graph.

\begin{definition}\label{def: op-sys} \cite{Duan2013noncomm}
    Let $\Phi: \B{\Hil}\to \B{\Kil}$ have a Kraus representation $\Phi(X)=\sum_{i=1}^n K_i XK_i^{*}$. The operator system (or the \emph{noncommutative (confusability) graph}) of $\Phi$ is defined as
\begin{equation*}
    S_{\Phi} := {\rm{span}} \{K^{*}_i K_j : \, 1\leq i,j \leq n\} \subseteq \B{\Hil}. 
\end{equation*}
\end{definition}

It is easy to check that the above definition is independent of the chosen Kraus representation of $\Phi$. Moreover, $\sum_{i=1}^n K_i^{*} K_i=1_A \in S_{\Phi}$ (since $\Phi$ is trace-preserving) and $X\in S_{\Phi} \implies X^{*} \in S_{\Phi}$. Such $*-$closed subspaces $S\subseteq \B{\Hil}$ containing the identity are called \emph{operator systems} \cite{paulsen-book}. Moreover, any such operator system $S$ arises as the noncommutative graph of some channel $\Phi$ \cite{Duan2009zerosuper}. One can check that if $\Phi_c:\B{\Hil}\to \B{\Hil_E}$ is complementary to $\Phi$, then the operator system is obtained as the image of the environment algebra under $(\Phi_c)^*$ \cite{Duan2013noncomm}:
\begin{equation}
    S_{\Phi} = (\Phi_c)^*(\B{\Hil_E}) := \{(\Phi_c)^* (X) : X\in\B{\Hil_E}\}.
\end{equation}

\begin{remark}\label{remark:c-q-embed}
For a column stochastic matrix $N\in \mathbb{M}_{d_B\times d_A}(\R)$, we define the corresponding quantum channel $\Phi_N : \B{\Hil_A}\to \B{\Hil_B}$ as 
\begin{equation}\label{eq:PhiN}
    \Phi_N(X) := \sum_{j\in [d_B]} \sum_{i\in [d_A]} N(j|i) \langle i|X|i\rangle \ket{j}\bra{j},
\end{equation}
where $\{\ket{i}\}_{i\in [d_A]}\subseteq \Hil_A$, $\{\ket{j} \}_{j\in [d_B]}\subseteq \Hil_B$ denote the standard bases. One should verify
\begin{equation*}
    S_{\Phi_N} = \operatorname{span} \{ \ket{i}\bra{i'} \, : \, i=i' \,\,\text{or} \,\, i\sim_N i' \} \subseteq \B{\Hil_A}.
\end{equation*}
\end{remark}

Exactly as in the classical case, we can introduce noncommutative graph parameters that describe the zero-error communication capacities of the corresponding noisy channel.

\begin{definition}[Independence numbers]\cite{Duan2013noncomm} \label{def:op-parameters}
    For an operator system $S\subseteq \B{\Hil}$,
\begin{itemize}
    \item the maximum $\mathscr{M}\in \N$ such that there exist pure states $\{\psi_m \}_{m\in [\mathscr{M}]} \subseteq \cD(\Hil)$ such that 
\begin{equation}
    \forall m\neq m': \quad |\psi_m\rangle \langle\psi_{m'} | \perp S
\end{equation}
is called the \emph{independence number} of $S$ (denoted as $\alpha(S)$) . 
\item the maximum number $\mathscr{M}\in \N$ such that there exist Hilbert spaces $\Hil_{A}, \Hil_R$, a state $\rho\in \cD(\Hil_{A})$, and isometries $\{V_m: \Hil_{A} \to \Hil\otimes \Hil_R\}_{m\in [\mathscr{M}]}$ such that 
\begin{equation}
\forall m\neq m': \quad V_m \rho V_{m'}^* \perp S\otimes \B{\Hil_R},
\end{equation}
is called the \emph{entanglement-assisted independence number} of $S$ (denoted ${\alpha}_{\rm ea}(S)$).  
\item the maximum $d\in \N$ such that there exists a subspace $C\subseteq \Hil$ with $\dim C=d$ satisfying $P_C S P_C = \mathbb{C} P_C$, is called the \emph{quantum independence number} of $S$ (denoted as $\alpha_q(S)$). 
\end{itemize}
\end{definition}

\begin{definition}[One-shot zero-error capacities]
\label{def:zero-error-capacities}
For a channel $\Phi:\B{\Hil}\to\B{\Kil}$, 
\begin{itemize}
    \item its \emph{one-shot zero-error classical capacity} is defined as
\begin{align}
\label{eq:ZeroErrorClassicalCapacity}
   \nn C_0(\Phi) :=\log \max \bigg\{ \mathscr{M} : \exists \rho_1, \ldots ,\rho_{\mathscr{M}} \in \cD (\Hil) \text{ such that } \\ \forall \, m\neq m':  \Phi(\rho_m) \perp \Phi(\rho_{m'}) \bigg\}.
\end{align} 
 \item its \emph{one-shot zero-error entanglement-assisted classical
    capacity} is defined as
    \begin{align}
        C_0^{\rm ea}(\Phi)
        :=\log\max\bigg\{\mathscr{M}\in\N:
        \exists \,\text{Hilbert spaces }\Hil_0, \Hil_R, \, \rho\in\cD(\Hil_0\otimes\Hil_R), \,\,\, \nonumber \\ \text{ channels } \cE_1, \ldots \cE_{\mathscr{M}} : \B{\Hil_0} \to \B{\Hil}  \,\text{such that }\,\, \nonumber \\
       \forall\,m\neq m': (\Phi \circ \cE_{m}\otimes\id_R)(\rho)
        \perp(\Phi\circ \cE_{m'}\otimes\id_R)(\rho)\bigg\}.
    \end{align}
    \item its \emph{one-shot zero-error quantum capacity}  is defined as
    \begin{align}
        \nn Q_0(\Phi) := \log \max \bigg\{ d : \exists \,\text{channels } \cE:\B{\C^d}\to \B{\Hil}, \cD :\B{\Kil}\to \B{\C^d} \\
        \text{ such that } \cD \circ \Phi \circ \cE = \id_{\B{\C^d}}\bigg\}.
        \label{eq:ZeroErrorQuantumCapacity}
    \end{align}
\end{itemize}
\end{definition}

\begin{theorem}\cite{Duan2013noncomm}\label{thm:DSW}
    For a quantum channel $\Phi: \B{\Hil}\to \B{\Kil}$,
\begin{align}
    C_0 (\Phi) &= \log \alpha (S_{\Phi}), \\
    C_0^{\rm ea} (\Phi) &= \log \alpha_{\rm ea}(S_{\Phi}) \\
    Q_0 (\Phi) &= \log \alpha_q (S_{\Phi}).
\end{align}
\end{theorem}

In general, computing the independence numbers of operator systems -- or, equivalently, computing the one-shot zero-error capacities of quantum channels -- is difficult \cite{Shor2008complexity, Costa2010oneshot-NPhard}. Moreover, the independence numbers can be highly non-multiplicative under tensor products \cite{Chen2010zerosuper}. However, it was recently shown \cite{Singh2025thesis, Singh2026markovian} that if an operator system $S\subseteq \B{\Hil}$ also has the structure of an algebra (i.e., it is closed under matrix multiplication), then its independence numbers can be explicitly computed and are multiplicative.

\begin{lemma}\label{lemma:alpha-algebra} \cite[Lemma 18]{Singh2026markovian}
    Let $S= \bigoplus_{k} \left( 1_{A_k}\otimes \B{B_k} \right) \subseteq \B{\Hil}$ be a $*-$algebra, where the block structure is with respect to the underlying decomposition $\Hil = \oplus_k A_k \otimes B_k$. Then,  
    \begin{align} 
     \sum_k d_{A_k} &= \alpha(S) \\
     \sum_k d^2_{A_k} &= \alpha_{\rm ea}(S) \\
     \max_k d_{A_k} &= \alpha_q (S)
    \end{align}
    Furthermore, for any other $*-$algebra $T\subseteq\B{\Kil}$,
    \begin{align}
    \alpha(S\otimes T) &= \alpha(S)\alpha(T), \\
    \alpha_{\rm ea}(S\otimes T) &= \alpha_{\rm ea}(S)\alpha_{\rm ea}(T), \\
    \alpha_q(S\otimes T) &= \alpha_q(S) \alpha_q(T).
    \end{align}
\end{lemma}

In the language of operator systems, the notions of pre- and post-processing by quantum channels are captured by homomorphisms and inclusions.

\begin{definition} \cite{Stahlke2016zero} \label{def:op-homo}
    Let $S\subseteq \B{\Hil}$ and $T\subseteq\B{\Kil}$ be operator systems. We say that $S$ is (graph) \emph{homomorphic} to $T$, denoted $S\longrightarrow T$, if there exists an isometry $V:\Hil\to \Kil\otimes \Hil_E$ such that 
    \begin{equation*}
         V^{*} (T \otimes \B{\Hil_E} ) V \subseteq S.
    \end{equation*}
    We call two operator systems $S$ and $T$ (graph) \emph{isomorphic} if $S\longrightarrow T$ and $T\longrightarrow S$.
\end{definition}

Let us note some basic properties of graph homomorphisms and inclusions below. 

\begin{lemma} \cite{Stahlke2016zero} \label{lemma:op-homo-2}
    Let $Q,R, S$ and $T$ be operator systems.
    \begin{itemize}
        \item If $R\longrightarrow S$ and $S\longrightarrow T$, then, $R\longrightarrow T$. 
        \item If $Q\longrightarrow R$ and $S\longrightarrow T$, then $Q \otimes S \longrightarrow R \otimes T$.
    \end{itemize}
\end{lemma}

\begin{lemma}\label{lemma:op-homo}
    Let $\Phi:\B{\Hil}\to \B{\Kil}$, $\Psi:\B{\Kil}\to \B{\Gil}$ be quantum channels. Then,
    \begin{equation*}
        S_{\Phi} \subseteq S_{\Psi\circ \Phi} \longrightarrow S_{\Psi}.
    \end{equation*}
\end{lemma}
\begin{proof}
    Choose Kraus representations $\Phi(\cdot)=\sum_{i=1}^n K_i (\cdot)K_i^{*}$, $\Psi(\cdot)=\sum_{j=1}^m F_j (\cdot)F_j^{*}$. Then, 
    \begin{equation*}
        S_{\Psi\circ \Phi} = \operatorname{span}\{K_i^{*}F_j^{*} F_q K_p : 1\leq i,p\leq n, 1\leq j,q\leq m \}.
    \end{equation*}
    Clearly, for all $1\leq i,p\leq n$, we have $K_i^{*}K_p = \sum_j K_i^{*}F_j^{*}F_j K_p \in S_{\Psi\circ \Phi}$, since $\sum_j F^{*}_j F_j = \iden_{\Kil}$. Hence, $S_{\Phi} = \operatorname{span}\{K^{*}_i K_p : 1\leq i,p \leq n \}\subseteq S_{\Psi\circ \Phi}$. 

    To prove the second claim, let $V:\Hil\to \Kil\otimes \Hil_E$ be a Stinespring isometry for $\Phi$. Then, it is easy to check that (see for e.g. \cite{Duan2013noncomm})
    \begin{equation*}
         V^{*} (S_{\Psi}\otimes \B{\Hil_E}) V =S_{\Psi\circ \Phi}.
    \end{equation*}
\end{proof}

\begin{lemma}\label{lemma:op-bottleneck}
   Let $S\subseteq \B{\Hil}$ and $T\subseteq\B{\Kil}$ be operator systems such that $S\longrightarrow T$. Then,
    \begin{equation*}
        \alpha(S)\leq \alpha(T), \quad \alpha_{\rm ea}(S)\leq \alpha_{\rm ea}(T), \quad \alpha_q (S)\leq \alpha_q (T).
    \end{equation*} 
    Consequently, if $S$ and $T$ are (graph) isomorphic, the inequalities, then
    \begin{equation*}
        \alpha(S)= \alpha(T), \quad \alpha_{\rm ea}(S)= \alpha_{\rm ea}(T), \quad \alpha_q (S)= \alpha_q (T).
    \end{equation*} 
\end{lemma}

\subsection{Conditional expectations} \label{sec:cond-expec}

Let $\scA\subseteq\B{\Hil}$ be a unital $*$-subalgebra, where
unital means that $1_{\Hil}\in\scA$. Choose an orthogonal decomposition
$\Hil=\bigoplus_k A_k\otimes B_k$ in which
$\scA=\bigoplus_k\B{A_k}\otimes 1_{B_k}$.
A positive linear map $\cP_{\scA}^*:\B{\Hil}\to\B{\Hil}$ with
image $\operatorname{im}(\cP_{\scA}^*)=\scA$ that fixes every element of $\scA$ is called a
\emph{conditional expectation} onto $\scA$. Such a map is automatically
unital and completely positive, and satisfies the bimodule property
\begin{equation*}
    \cP_{\scA}^*(AXB)=A\cP_{\scA}^*(X)B,
    \qquad A,B\in\scA,\ X\in\B{\Hil}.
\end{equation*}
The finite-dimensional form of a conditional expectation and its adjoint
is as follows:
\begin{align}
    \cP_{\scA}^*(X)
    &=\bigoplus_k\Tr_{B_k}\!\left[
        (V_k^*XV_k)(1_{A_k}\otimes\delta_k)
      \right]\otimes 1_{B_k},
      \label{eq:cond-expec}\\
    \cP_{\scA}(X)
    &=\bigoplus_k\Tr_{B_k}(V_k^*XV_k)\otimes\delta_k,
      \nonumber
\end{align}
where $V_k:A_k\otimes B_k\to\Hil$ are the canonical inclusion isometries, $\delta_k\in\cD(B_k)$ are arbitrary states, and
$\Tr_{B_k}$ denotes the partial trace over $B_k$ \cite[Proposition 1.5]{Wolf2012Qtour}. The adjoint $\cP_{\scA}$ is an idempotent quantum channel. Choosing $\delta_k=1_{B_k}/d_{B_k}$ for every $k$ gives the unique
trace-preserving conditional expectation onto $\scA$. For this choice,
$\cP_{\scA}=\cP_{\scA}^*$ is the orthogonal projection onto $\scA$
for the Hilbert--Schmidt inner product on $\B{\Hil}$.

Although different choices of the states $\delta_k$ give different
adjoint channels, these channels can simulate each other by
post-processing. Indeed, if $\cP_{\scA}$ and $\cP'_{\scA}$ are the
adjoints of two conditional expectations onto the same algebra, then
\begin{equation}\label{eq:conditional-post-processing}
    \cP'_{\scA}\circ\cP_{\scA}=\cP'_{\scA},
    \qquad
    \cP_{\scA}\circ\cP'_{\scA}=\cP_{\scA}.
\end{equation}

\subsection{Asymptotics of quantum channels}
Let $\Phi:\B{\Hil}\to \B{\Hil}$ be a quantum channel. Then, $\Phi$ admits a Jordan decomposition 
\begin{equation*}
    \Phi = \sum_{i} \lambda_i \mathcal{P}_i + \mathcal{N}_i \quad \text{with} \quad \mathcal{N}_i \mathcal{P}_i = \mathcal{P}_i \mathcal{N}_i = \mathcal{N}_i \,\,\, \text{and} \,\,\, \mathcal{P}_i \mathcal{P}_j = \delta_{ij}\mathcal{P}_i,
\end{equation*}
where the sum runs over the distinct eigenvalues $\lambda_i$ of $\Phi$, $\mathcal{P}_i$ are projectors whose rank equals the algebraic multiplicity of $\lambda_i$, and $\mathcal{N}_i$ denote the corresponding nilpotent operators \cite{Wolf2012Qtour}. All the eigenvalues $\lambda_i$ of $\Phi$ satisfy $\abs{\lambda_i}\leq 1$ and they are either real or come in complex conjugate pairs. Since $\Phi$ always admits a fixed point, $\lambda=1$ is always an eigenvalue of $\Phi$. Moreover, all $\lambda_i$ with $\abs{\lambda_i}=1$ have equal algebraic and geometric multiplicities, so that $\mathcal{N}_i=0$ for all such eigenvalues. We denote the $l$-fold iterated channel by
\begin{equation*}
    \Phi^l := \underbrace{\Phi \circ \Phi \circ \ldots \circ \Phi}_{l\,  \text{times}}.
\end{equation*}
The \emph{peripheral space} of $\Phi$ is defined as
\begin{equation*}
    \mathscr X(\Phi)
    =\operatorname{span}\left\{X\in\B{\Hil}:
        \Phi(X)=e^{i\theta}X\text{ for some }\theta\in\mathbb R
      \right\}.
\end{equation*}
The asymptotic part of $\Phi$ and the projector onto the peripheral space $\mathscr{X} (\Phi)$, are respectively defined as follows: 
   \begin{equation}\label{eq:phiinf-proj}
\Phi_{\infty}:= \sum_{i:\, |\lambda_i|=1}\lambda_i \mathcal{P}_i \quad \text{and} \quad  \mathcal{P}_{\Phi} = \sum_{i: \, |\lambda_i|=1} \mathcal{P}_i .
\end{equation}
Both maps arise as limit points of $\{\Phi^l \}_{l\in \N}$, and hence are quantum channels themselves \cite{Wolf2012Qtour}. In particular, $\cP_{\Phi}$ is idempotent. As $l\to \infty$, since the non-peripheral eigenmodes decay to zero, the following holds true in any norm on the finite-dimensional space of linear maps:
\begin{equation*}
    \lim_{l\to\infty}\norm{\Phi^l-\Phi_\infty^l}=0.
\end{equation*}

We next describe the algebra associated with the peripheral space. Define $\Hil_0^{\perp}:= \im (\cP_{\Phi}(1_{\Hil}))$ and the restricted channel $\overbar{\cP}_{\Phi}:\B{\Hil_0^{\perp}}\to \B{\Hil_0^{\perp}}$ as 
\begin{equation}\label{eq:Pbar}
    \overbar{\cP}_{\Phi}(\cdot) = V^{*} \cP_{\Phi}( V(\cdot)V^{*}) V,
\end{equation}
where $V:\Hil_0^{\perp}\to \Hil$ is the canonical inclusion isometry. It is then easy to check that $\overbar{\cP}_{\Phi}$ is also an idempotent channel. Moreover, it has a full-rank invariant state by construction, so that $\mathscr{X}^*(\Phi):=\im (\overbar{\cP}_{\Phi}^*) \subseteq \B{\Hil_0^{\perp}}$ is a unital $*-$subalgebra \cite{Lindblad1999fixed, Wolf2012Qtour}.

\begin{definition}\label{def:peripheral-alg}
    For a channel $\Phi:\B{\Hil}\to \B{\Hil}$, its \emph{peripheral algebra} is defined as the image $\mathscr{X}^*(\Phi):= \im (\overbar{\cP}_{\Phi}^*)\subseteq \B{\Hil^{\perp}_0}$, where $\overbar{\cP}_{\Phi}$ is defined as the restriction \eqref{eq:Pbar}.
\end{definition}

Since $\mathscr{X}^*(\Phi) \subseteq \B{\Hil_0^{\perp}}$ is a unital $*$-subalgebra, there exists an orthogonal decomposition
$\Hil_0^\perp=\bigoplus_k A_k\otimes B_k$ such that $\mathscr{X}^*(\Phi)=\oplus_k \B{A_k}\otimes 1_{B_k}$ \cite{Takesaki1979algebra, Arveson1976algebra}. The conditional-expectation formulas from Section~\ref{sec:cond-expec} give
\begin{align}
    \overbar{\cP}_\Phi(Y)
    &=\bigoplus_{k}\Tr_{B_k}(V_k^*YV_k)\otimes\delta_k,
      \label{eq:Pbar-cond}\\
    \overbar{\cP}_\Phi^*(Y)
    &=\bigoplus_{k}\Tr_{B_k}\!\left[
        (V_k^*YV_k)(1_{A_k}\otimes\delta_k)
      \right]\otimes1_{B_k},
      \label{eq:Pbar-cond2}
\end{align}
where $V_k:A_k\otimes B_k\to\Hil_0^\perp$ are the canonical isometries and $\delta_k\in\cD(B_k)$ are full-rank states. 

The following identities are easy to verify from the above discussion:
\begin{align}
    \mathcal{P}_{\Phi} &= \mathcal{V}\circ R_V \circ \mathcal{P}_{\Phi}, \label{eq:PPbar-relations} \\
    R_V \circ \mathcal{P}_{\Phi} &= \overbar{\mathcal{P}}_{\Phi}\circ R_V\circ \mathcal{P}_{\Phi}, \label{eq:PPbar-relations2} \\ 
    \overbar{\mathcal{P}}_{\Phi} &=R_V\circ \mathcal{P}_{\Phi}\circ \mathcal{V},\label{eq:PPbar-relations3}
\end{align}
where $\mathcal{V}:\B{\Hil_0^{\perp}}\to \B{\Hil}$ is the isometric channel $\mathcal{V}(X)=VXV^{*}$ and $R_V:\B{\Hil}\to \B{\Hil_0^{\perp}}$ is the restriction channel $R_V (Y) = V^{*} Y V + \Tr [(1 - VV^{*})Y ]\sigma $ for some state $\sigma\in \cD(\Hil_0^{\perp})$. Moreover, it is easy to check that
\begin{equation}\label{eq:SPbar=X*'}
    S_{\overbar{\cP}_\Phi}
    =\bigoplus_k1_{A_k}\otimes\B{B_k}
    =\bigl(\mathscr X^*(\Phi)\bigr)',
\end{equation}
where the commutant is taken in $\B{\Hil_0^\perp}$.
These identities connect the peripheral algebra to the operator
systems used to study exact transmission.

For detailed proofs of the facts in this section, see \cite[Chapter 6]{Wolf2012Qtour}, \cite[Chapter 2.7]{Singh2025thesis}.

\subsection{Correctable algebras}
 
\begin{definition} \label{def:correctable} A unital $*$-subalgebra $\scA \subseteq \B{\Hil}$ is called \emph{correctable} for a channel $\Phi:\B{\Hil}\to \B{\Kil}$ if any of the following equivalent conditions hold:
\begin{itemize}
    \item $\exists$ channel $\cD : \B{\Kil}\to \B{\Hil}$ such that $(\cD\circ \Phi)^*(X)=X$ for all $X\in \scA$.
    \item $\exists$ channel $\cD : \B{\Kil}\to \B{\Hil}$ such that $(\cD\circ \Phi)^*\circ \cP_{\scA}^*=\cP_{\scA}^*$.
    \item $\exists$ channel $\cD : \B{\Kil}\to \B{\Hil}$ such that $\cP_{\scA} \circ\cD\circ \Phi=\cP_{\scA}$.
    \item $\exists$ channel $\cD : \B{\Kil}\to \B{\Hil}$ such that $\cD\circ \Phi=\cP_{\scA}$.
\end{itemize}
Here $\cP_{\scA}^*$ is any conditional expectation onto $\scA$,
and $\cP_{\scA}$ is its adjoint channel.

\end{definition}

The first two conditions are equivalent because
$\operatorname{im}(\cP_{\scA}^*)=\scA$, and the second and third
are adjoints of one another. The third implies the fourth after replacing
the decoder by $\cP_{\scA}\circ\cD$; the reverse implication follows
from $\cP_{\scA}^2=\cP_{\scA}$. Thus the decoders in the four
conditions need not be the same. The post-processing identities \eqref{eq:conditional-post-processing} show that correctability is independent of the chosen conditional
expectation. We may therefore use the unique trace-preserving choice
$\cP_{\scA}=\cP_{\scA}^*$.
Operationally, correctability means that a single recovery channel
restores the expectation value of every observable in $\scA$, for every
input state.

The operator algebra Knill--Laflamme conditions give the following
characterization.

\begin{theorem}[\cite{Beny2007opalg-ecc, Beny2007opalg-ecc2}] \label{theorem:KL-alg} 
A unital $*$-subalgebra $\scA\subseteq\B{\Hil}$ is correctable for a channel $\Phi:\B{\Hil}\to \B{\Kil}$, given by its Kraus representation
$\Phi(\cdot)=\sum_i K_i(\cdot)K_i^*$, if and only if
\begin{equation}\label{eq:KL-1}
    [X,K_i^*K_j]=0
    \qquad\text{for all }X\in\scA\text{ and all }i,j.
\end{equation}
\end{theorem}

More generally, let $\scA\subseteq\B{\Hil}$ be a $*$-subalgebra
with unit projection $P$, which may differ from $1_{\Hil}$. Let
$C=P\Hil$, and let $V:C\to\Hil$ be the canonical inclusion isometry,
so that $VV^*=P$. We extend Definition~\ref{def:correctable} by calling
$\scA$ correctable for $\Phi$ when the unital algebra
$V^*\scA V\subseteq\B{C}$ is correctable for the restricted channel
\begin{equation*}
    \Phi|_C:\B{C}\to\B{\Kil},
    \qquad \Phi|_C(Y)=\Phi(VYV^*).
\end{equation*}
Applying Theorem~\ref{theorem:KL-alg} to the Kraus operators $K_iV$
of $\Phi|_C$ gives the following form of the operator algebra Knill--Laflamme conditions. 

\begin{theorem}[\cite{Beny2007opalg-ecc, Beny2007opalg-ecc2}] \label{theorem:KL-alg2}
A $*$-subalgebra $\scA\subseteq\B{\Hil}$ with unit projection $P$ is correctable for a channel $\Phi: \B{\Hil}\to \B{\Kil}$, given in its Kraus form
$\Phi(\cdot)=\sum_i K_i(\cdot)K_i^*$, if and only if
\begin{equation}\label{eq:KL-2}
    [X,PK_i^*K_jP]=0
    \qquad\text{for all }X\in\scA\text{ and all }i,j.
\end{equation}
In particular, suppose that $\Hil=C\oplus C^\perp$, where
$C=\bigoplus_k A_k\otimes B_k$, and
\begin{equation}
    \scA=
    \left(\bigoplus_k1_{A_k}\otimes\B{B_k}\right)\oplus0_{C^\perp}.
    \label{eq:BlockDecompTheorem}
\end{equation}
Then, $\scA$ is correctable for $\Phi$ if and only if
\begin{equation*}
    PS_\Phi P\subseteq
    \left(\bigoplus_k\B{A_k}\otimes1_{B_k}\right)\oplus0_{C^\perp}.
\end{equation*}
\end{theorem}

Finally, the action of a channel on a correctable algebra can be
described by a representation. The following formulation is a simple extension of the proof presented in
\cite{Choi2009opalg-ecc-rep}.

\begin{theorem}\label{thm:alg-correct-rep}
A $*$-subalgebra $\scA\subseteq\B{\Hil}$ with unit projection $P$
is correctable for a channel $\Phi:\B{\Hil}\to\B{\Kil}$ if and only
if there exists a faithful representation (i.e. an injective
$*$-homomorphism) $\pi:\scA\to\B{\Kil}$ such that
\begin{equation*}
    \Phi(X)=\Phi(P)\pi(X)=\pi(X)\Phi(P)
    \qquad\text{for all }X\in\scA.
\end{equation*}
The representation can be chosen so that $\pi(P)=Q$, where $Q$ is
the orthogonal projection onto $\operatorname{im}\Phi(P)$.
\end{theorem}

\begin{proof}
Suppose first that $\scA$ is correctable, and write
$\Phi(\cdot)=\sum_{i=1}^r K_i(\cdot)K_i^*$. Let $C=P\Hil$ and
$V:C\to\Hil$ be the inclusion isometry. Define the operator
$W:C\otimes\C^r\to\Kil$ by
\begin{equation*}
    W\left(\sum_{i=1}^r\psi_i\otimes\ket{i}\right)
    =\sum_{i=1}^rK_iV\psi_i.
\end{equation*}
For $X\in\scA$, set $T_X=(V^*XV)\otimes1_{\C^r}$.
The Knill--Laflamme conditions \eqref{eq:KL-2} imply that
\begin{equation*}
    G:=W^*W
    =\sum_{i,j}V^*K_i^*K_jV\otimes\ket{i}\bra{j}
\end{equation*}
commutes with every $T_X$. Write the polar decomposition as
$W=UG^{1/2}$, with initial projection $R=U^*U$ and final projection
$UU^*=Q$. The latter identity follows from $WW^*=\Phi(P)$.
Since $R$ also commutes with every $T_X$ (being the support projection of $G$), the map
\begin{equation*}
    \pi(X)=UT_XU^*
\end{equation*}
is a $*$-homomorphism: using $U=UR$, we have
$\pi(X)\pi(Y)=UT_XRT_YU^*=UT_{XY}U^*=\pi(XY)$.
It satisfies $\pi(P)=Q$ and $\pi(X)W=WT_X$, or equivalently,
\begin{equation*}
    \pi(X)K_iP=K_iX
    \qquad\text{for all }i\text{ and }X\in\scA.
\end{equation*}
Consequently,
\begin{equation*}
    \pi(X)\Phi(P)=\sum_i\pi(X)K_iPK_i^*
    =\sum_iK_iXK_i^*=\Phi(X).
\end{equation*}
Applying this identity to $X^*$ and taking adjoints gives
$\Phi(P)\pi(X)=\Phi(X)$. If $\pi(X)=0$, then
$\Phi(X^*X)=\Phi(P)\pi(X^*X)=0$, so trace preservation implies
$\Tr(X^*X)=0$ and hence $X=0$. Thus $\pi$ is faithful.

Conversely, suppose that such a representation exists. For any
projection $e\in\scA$, the identity
$\Phi(e)=\Phi(P)\pi(e)=\pi(e)\Phi(P)$ shows that
$\Phi(e)$ is supported on the range of $\pi(e)$. Since
$\Phi(e)=\sum_i(K_ie)(K_ie)^*$, positivity implies
\begin{equation*}
    (1_{\Kil}-\pi(e))K_ie=0.
\end{equation*}
Applying the same argument to $P-e$, and using
$\pi(e)\pi(P-e)=0$, gives $\pi(e)K_i(P-e)=0$.
It follows that $\pi(e)K_iP=K_ie$. Projections linearly span the
finite-dimensional algebra $\scA$, so
$\pi(X)K_iP=K_iX$ for every $X\in\scA$. Taking adjoints of the
identity for $X^*$ now yields
\begin{equation*}
    XPK_i^*K_jP
    =PK_i^*\pi(X)K_jP
    =PK_i^*K_jPX.
\end{equation*}
Theorem~\ref{theorem:KL-alg2} therefore implies that $\scA$ is
correctable.
\end{proof}

\section{Transmittable algebras}
\label{sec:TransmittableAlgebras}
In this section, we introduce the central notion of a \emph{transmittable} algebra for a noisy channel. 

\subsection{Basic definitions and results}

\begin{definition}\label{def:alg-transmit}
    Let $\Phi:\B{\Hil}\to \B{\Kil}$ be a quantum channel. We say that a finite-dimensional $C^*$-algebra $\scA$ is \emph{transmittable} via $\Phi$ if there exists an encoding faithful representation $\pi_E:\scA\to \B{\Hil}$ such that the encoded algebra $\pi_E(\scA) \subseteq \B{\Hil}$ is correctable for $\Phi$ in the sense of Definition~\ref{def:correctable}.  The transmission set $\scT (\Phi)$ of $\Phi$ is defined as
\begin{equation}
    \scT (\Phi)
    :=
    \left\{
        \scA \ \middle|\ 
        \scA \text{ is a finite-dimensional C}^*\text{-algebra
        transmittable via } \Phi
    \right\},
\end{equation}
with the preorder $\leq$ induced by the embedding preorder (see Definition~\ref{def:embedding-preorder}).
\end{definition}

We identify abstract algebras up to $*$-isomorphism when discussing
the elements and cardinalities of $\scT(\Phi)$, so that the preorder $\leq$ on $\scT(\Phi)$ becomes a partial order (see \eqref{eq:pre-partial-order}). Every encoding
representation of a transmittable algebra $\scA\in \scT(\Phi)$ acts on the input space $\Hil$, so Lemma~\ref{lemma:finite-algebra-types}
implies that $\scT(\Phi)$ contains only finitely many algebra types. 

\begin{remark}\label{remark:lower-set}
The transmission set is a \emph{lower set}, i.e. if $\scB\in\scT(\Phi)$ and
$\scA\leq\scB$, then $\scA\in\scT(\Phi)$. Indeed, let
$\phi:\scA\to\scB$ be an embedding (Definition~\ref{def:embedding-preorder}) and
$\pi_E:\scB\to\B{\Hil}$ a faithful encoding representation for
$\scB$. Their composition $\pi_E':=\pi_E\circ\phi$ serves as the faithful encoding representation for $\scA$. Indeed, writing $P=\pi_E(1_{\scB})$ and $Q=\pi_E'(1_{\scA})$, we have
$Q\leq P$, and since $\scB\in \scT(\Phi)$, the Knill--Laflamme conditions (Theorem~\ref{theorem:KL-alg2}) imply
\begin{equation*}
    [X,QS_\Phi Q]=Q[X,PS_\Phi P]Q=0
    \qquad\text{for all }X\in\pi_E'(\scA).
\end{equation*}
Thus $\pi_E'(\scA)$ is correctable for $\Phi$, meaning $\scA\in \scT(\Phi)$.
\end{remark}

In the next lemma, we note that Definition~\ref{def:alg-transmit} is equivalent to saying that the noisy channel $\Phi$ can simulate a conditional expectation onto a finite-dimensional $*$-subalgebra $\scA = \oplus_k \M{d_k} \subseteq \B{\oplus_k \C^{d_k}}$ via suitable encoding and decoding channels, thus matching the standard capacity definitions from quantum Shannon theory \cite{Wilde2016}.

\begin{lemma} 
\label{lemma:EncDecTransmittableAlgebra}
    Let $\Phi:\B{\Hil}\to\B{\Kil}$ be a channel, and
    $\scA=\bigoplus_k\M{d_k}\subseteq\B{L}$ be represented in the
    standard way on $L=\bigoplus_k\C^{d_k}$. Then, $\scA$ is
    transmittable via $\Phi$ if and only if there exist encoding and
    decoding channels $\cE:\B{L}\to\B{\Hil}$,
    $\cD:\B{\Kil}\to\B{L}$ such that
    \begin{equation}
        \cD \circ \Phi \circ \cE = \cP_{\scA},
    \end{equation}
    where $\cP_{\scA}(X)=\sum_kP_kXP_k$ is the unique
    trace-preserving conditional expectation onto $\scA$, and $P_k$
    is the orthogonal projection of $L$ onto its $k$th summand.
\end{lemma}
\begin{proof}
Suppose first that $\scA$ is transmittable via $\Phi$. Choose a
faithful encoding representation $\pi_E$ with correctable image, and
write $\Hil=C\oplus C^\perp$, where
$C=\bigoplus_k A_k\otimes B_k$ and
\begin{equation*}
    \pi_E(\scA)
    =\left(\bigoplus_k1_{A_k}\otimes\B{B_k}\right)\oplus0_{C^\perp},
    \qquad \dim B_k=d_k.
\end{equation*}
Let $V_C:C\to\Hil$ be the inclusion isometry and
$\cV_C(Y)=V_CYV_C^*$. Fix unitary identifications
$B_k\cong\C^{d_k}$, so that $L=\bigoplus_kB_k$, and set
$\scB:=V_C^*\pi_E(\scA)V_C\subseteq\B{C}$.

For $X\in\B{L}$ and $Y\in\B{C}$, denote their diagonal blocks by
$X_{kk}\in\B{B_k}$ and $Y_{kk}\in\B{A_k\otimes B_k}$, respectively.
Define channels on the full matrix algebras by
\begin{align*}
    \cE':\B{L}\to\B{C},
    &\qquad \cE'(X)
       =\bigoplus_k\frac{1_{A_k}}{d_{A_k}}\otimes X_{kk},\\
    \cR:\B{C}\to\B{L},
    &\qquad \cR(Y)=\bigoplus_k\Tr_{A_k}(Y_{kk}).
\end{align*}
Both maps are quantum channels that discard off-diagonal blocks. Moreover,
$\cE'(X)\in\scB$ for every $X$, and
$\cR\circ\cE'=\cP_{\scA}$.
By correctability, there is a channel
$\cD':\B{\Kil}\to\B{C}$ such that
$\cD'\circ\Phi\circ\cV_C=\cP_{\scB}$, where $\cP_{\scB}$ is the
trace-preserving conditional expectation onto $\scB$
(Definition~\ref{def:correctable}). Hence, the channels
$\cE:=\cV_C\circ \cE'$ and $\cD:=\cR\circ\cD'$ satisfy
\begin{equation*}
    \cD\circ\Phi\circ\cE
    =\cR\circ\cP_{\scB}\circ{\cE'}
    =\cR\circ{\cE'}
    =\cP_{\scA}.
\end{equation*}

Conversely, suppose that channels $\cE:\B{L}\to\B{\Hil}$ and
$\cD:\B{\Kil}\to\B{L}$ satisfy
$\cD\circ\Phi\circ\cE=\cP_{\scA}$. Let
$C=\operatorname{im}\cE(1_L)$, let $P_C$ be the orthogonal
projection onto $C$, and let $V_C:C\to\Hil$ be the inclusion
isometry. Every Kraus operator of $\cE$ has range contained in $C$,
so all outputs of $\cE$ are supported on $C$. Consequently,
\begin{equation*}
    \cE^*(X)=\cE^*(P_CXP_C)\quad(X\in\B{\Hil}),
    \qquad \cE^*(P_C)=1_L.
\end{equation*}
The restriction of $\cE^*$ to the corner $P_C\B{\Hil}P_C$ is
faithful on positive operators. Indeed, if $Z=P_CZP_C\geq0$ and
$\cE^*(Z)=0$, then
\begin{equation*}
    0=\Tr\bigl(\cE^*(Z)\bigr)
     =\Tr\bigl(Z\cE(1_L)\bigr),
\end{equation*}
which implies $Z=0$ because $\cE(1_L)$ is strictly positive on $C$.

Define the completely positive map
\begin{equation*}
    \Theta:\B{L}\to P_C\B{\Hil}P_C,
    \qquad
    \Theta(X)=P_C(\Phi^*\circ\cD^*)(X)P_C.
\end{equation*}
It is unital as a map into this corner, whose unit is $P_C$.
Since $\cP_{\scA}^*=\cP_{\scA}$, taking adjoints of the assumed
channel identity gives $\cE^*\circ\Theta(X)=X$ for all
$X\in\scA$. Two applications of the Schwarz inequality yield
\begin{align*}
    X^*X
    &=\cE^*\bigl(\Theta(X^*X)\bigr)\\
    &\geq\cE^*\bigl(\Theta(X)^*\Theta(X)\bigr)\\
    &\geq\cE^*\bigl(\Theta(X)^*\bigr)\,
          \cE^*\bigl(\Theta(X)\bigr)
     =X^*X.
\end{align*}
Thus the positive operator
$Z_X:=\Theta(X^*X)-\Theta(X)^*\Theta(X)$ belongs to the corner and
satisfies $\cE^*(Z_X)=0$. Faithfulness implies $Z_X=0$.
Applying the same argument to $X^*$ also gives
$\Theta(XX^*)=\Theta(X)\Theta(X)^*$. Hence every $X\in\scA$
belongs to the multiplicative domain of $\Theta$
\cite{Choi1974schwarz}, and $\pi_E:=\Theta|_{\scA}$ is a
$*$-homomorphism. It is injective because
$\cE^*\circ\pi_E=\id_{\scA}$.

Finally, write $\cV_C(Y)=V_CYV_C^*$ and define
$\widetilde{\cD}:=\cV_C^*\circ\cE\circ\cD$.
This map is completely positive and trace-preserving: compression
by $\cV_C^*$ preserves the trace of every output of $\cE$, since
those outputs are supported on $C$. For $X\in\scA$ and
$Y=V_C^*\pi_E(X)V_C$, we have
\begin{align*}
    (\widetilde{\cD}\circ\Phi\circ\cV_C)^*(Y)
    &=V_C^*(\Phi^*\circ\cD^*\circ\cE^*)\bigl(\Theta(X)\bigr)V_C\\
    &=V_C^*(\Phi^*\circ\cD^*)(X)V_C\\
    &=V_C^*\pi_E(X)V_C=Y.
\end{align*}
Thus $\pi_E(\scA)$ is correctable for $\Phi$, proving
transmittability in the sense of Definition~\ref{def:alg-transmit}.
\end{proof}

Next, we give another equivalent description of transmittable algebras, this time using the operator algebra Knill-Laflamme error-correction conditions (Theorem~\ref{theorem:KL-alg}, \ref{theorem:KL-alg2}). We will make use of the following notion of orthogonality of subspaces. 

\begin{definition}
    Let $S\subseteq \B{\Hil}$ be an operator system. We say that two subspaces $C_1, C_2 \subseteq  \Hil$ are \emph{$S$-orthogonal} if
$P_1SP_2=0$, where $P_i$ is the projection onto $C_i$.
\end{definition}

Since $1_{\Hil}\in S$, $S$-orthogonality implies ordinary
orthogonality. For the operator system $S_{\Phi}$ of a channel $\Phi$,
it also has the operational interpretation
\begin{equation*}
    P_1S_\Phi P_2=0
    \quad\Longleftrightarrow\quad
    \Phi(P_1)\perp\Phi(P_2).
\end{equation*}
Indeed, if $\{K_i\}_i$ is a Kraus family for $\Phi$, then
$P_1K_i^*K_jP_2=0$ for all $i,j$ precisely when the output spaces
$\operatorname{span}_i \{K_iC_1\}$ and
$\operatorname{span}_j \{K_jC_2 \}$ are orthogonal. Thus, the condition
expresses perfect distinguishability of the entire output spaces
associated with the two codes.

\begin{lemma}\label{lemma:pairwise-S-orthogonal} 
Let $\Phi: \B{\Hil}\to \B{\Kil}$ be a quantum channel. A $*$-algebra
$\scA \cong \bigoplus_{k=1}^K M_{d_k}$ is transmittable via $\Phi$ if and only if there exist quantum codes $C_1, \ldots C_K \subseteq \Hil$ for $\Phi$ with
$\dim C_k=d_k$ that are pairwise $S_{\Phi}$-orthogonal, i.e., 
\begin{align}
 \forall k: \quad   P_k S_{\Phi} P_k &= \C P_k \\
  \forall k\neq j : \quad  P_k S_{\Phi} P_j &= 0,
\end{align}
where $P_k$ are the orthogonal projections onto $C_k$.
\end{lemma}

\begin{proof}
The proof is essentially a reformulation of the Knill-Laflamme conditions.

Suppose first that $\scA$ is transmittable via $\Phi$: there exists a faithful representation $\pi:\scA \to \B{\Hil}$ such that
$\pi(\scA) =(\oplus_k 1_{A_k} \otimes \B{B_k}) \oplus 0_{C^{\perp}} \subseteq \B{\Hil}$ is correctable for $\Phi$, where $\Hil = C\oplus C^{\perp}$ with the code space $C = \oplus_k A_k \otimes B_k$. Then, Theorem~\ref{theorem:KL-alg2} shows
\begin{equation}\label{eq:PSP-3}
P_{C}S_{\Phi}P_{C}\subseteq  (\oplus_k \B{A_k} \otimes 1_{B_k}) \oplus 0_{C^{\perp}}.
\end{equation}
Choose unit vectors $\ket{u_k}\in\mathcal A_k$ and set
$C_k:= \ket{u_k}\otimes B_k \subseteq C \subseteq \Hil$. The inclusion in \eqref{eq:PSP-3} shows
$P_kS_{\Phi}P_k=\mathbb C P_k$, so each $C_k$ is a quantum code for $\Phi$, and the block
diagonal form of $P_C S_{\Phi} P_C$ in \eqref{eq:PSP-3} shows $P_kS_{\Phi} P_j=0$ for $k\neq j$. 

Conversely, suppose $C_1, \ldots , C_K \subseteq \Hil$ are pairwise
$S_{\Phi}$-orthogonal quantum codes for $\Phi$ with $\dim C_k= d_k$. Let
$C:=\oplus_k C_k \subseteq \Hil$. Then, $P_C = \sum_k P_k$ and 
\begin{equation*}
P_{C}S_{\Phi}P_{C}
\subseteq
\bigoplus_{k=1}^K \mathbb C P_k.
\end{equation*}
Thus, $\scA \cong \oplus_k \M{d_k}$ is transmittable via $\Phi$ via the representation $\pi:\scA \to \B{\Hil}$ that unitarily identifies $\scA$ with $\pi(\scA)= \oplus_k \B{C_k} \subseteq \B{\Hil}$.
\end{proof}

In particular, the encoder in
Lemma~\ref{lemma:EncDecTransmittableAlgebra} can always be chosen
isometric. To see this, choose an isometry
$V:L=\oplus_k\C^{d_k}\to\Hil$ mapping each logical summand
unitarily onto the corresponding code $C_k$ from
Lemma~\ref{lemma:pairwise-S-orthogonal}. Then
\begin{equation}\label{eq:V:L-H}
    V^*S_\Phi V\subseteq\bigoplus_k\C\,1_{\C^{d_k}}.
\end{equation}
The standard algebra $\scA=\bigoplus_k\M{d_k}\subseteq\B{L}$ is
therefore correctable for $\Phi\circ\cV$, where
$\cV(X)=VXV^*$, so Definition~\ref{def:correctable} supplies a
decoding channel $\cD:\B{\Kil}\to\B{L}$ such that
\begin{equation}
    \cD\circ\Phi\circ\cV=\cP_{\scA}.
    \label{eq:isometric-algebra-transmission}
\end{equation}
Conversely, \eqref{eq:isometric-algebra-transmission} implies
transmittability by Lemma~\ref{lemma:EncDecTransmittableAlgebra}.
Thus one copy of each simple block suffices when determining
transmittability; additional multiplicity spaces are unnecessary.

To close this section, we note that (graph) isomorphic operator systems (Definition~\ref{def:op-homo}) have identical transmittable algebras.

\begin{lemma}\label{lemma:alg-DPI}
    Let $\Phi$ and $\Psi$ be quantum channels such that $S_{\Phi} \longrightarrow S_{\Psi}$. Then, 
    \begin{equation}
        \scA \in \scT(\Phi) \implies \scA \in \scT(\Psi).
    \end{equation}
\end{lemma}

\begin{proof}
    Let $\Phi:\B{\Hil}\to \B{\Kil}$ and $\Psi:\B{\Hil'}\to \B{\Kil'}$. By Definition~\ref{def:op-homo}, there exists a channel $\mathcal{N}:\B{\Hil}\to \B{\Hil'}$ with Stinespring isometry $V:\Hil\to \Hil'\otimes E$ such that 
    \begin{equation*}
      S_{\Psi\circ \mathcal{N}} = V^{*}(S_{\Psi} \otimes \B{E})V \subseteq S_{\Phi}.
    \end{equation*}
    Now, consider an algebra $\scA\in \scT(\Phi)$ with an encoding representation $\pi_E:\scA\to \B{\Hil}$. Let $P=\pi_E(1_{\scA})$, so the Knill-Laflamme conditions (Theorem~\ref{theorem:KL-alg2}) give
\begin{equation}
    [X,PsP]=0,
    \qquad X\in\pi_E(\scA),\ s\in S_\Phi.
\end{equation}
The inclusion $S_{\Psi\circ\mathcal N}\subseteq S_\Phi$
therefore implies the same relations for every
$s\in S_{\Psi\circ\mathcal N}$. Hence the same represented
algebra is correctable for $\Psi\circ\mathcal N$, and
$\scA\in\scT(\Psi\circ\mathcal N)$.
By Lemma~\ref{lemma:EncDecTransmittableAlgebra}, the channel
$\mathcal N$ can be included in the encoder, which gives
$\scA\in\scT(\Psi)$.
\end{proof}

\begin{corollary}
    Let $\Phi$ and $\Psi$ be quantum channels such that $S_{\Phi} \longrightarrow S_{\Psi}$ and $S_{\Psi} \longrightarrow S_{\Phi}$, i.e. the two operator systems are graph-isomorphic. Then, $\scT(\Phi) = \scT(\Psi)$.
\end{corollary} 

These results also give an exact graph-theoretic formulation of
transmittability. Let $\scA=\bigoplus_k\M{d_k}\subseteq\B{L}$ be
in its standard representation with $L=\oplus_k\C^{d_k}$. The conditional expectation
$\cP_{\scA}(X)=\sum_kP_kXP_k$ has confusability operator system
\begin{equation*}
    S_{\cP_{\scA}}
    =\operatorname{span}\{P_k\}_k
    =\bigoplus_k\C\,1_{\C^{d_k}} \subseteq \B{L}.
\end{equation*}
Then,
\begin{equation}
    \scA\in\scT(\Phi)
    \quad\Longleftrightarrow\quad
    S_{\cP_{\scA}}\longrightarrow S_\Phi.
    \label{eq:transmission-graph-preorder}
\end{equation}
For the forward implication, the isometry $V$ from \eqref{eq:V:L-H} satisfies
$V^*S_\Phi V\subseteq S_{\cP_{\scA}}$, which is a graph
homomorphism with a one-dimensional ancillary space (Definition~\ref{def:op-homo}). Conversely,
$\scA$ is trivially transmittable via $\cP_{\scA}$, so
Lemma~\ref{lemma:alg-DPI} gives the reverse implication.

\subsection{Converse bounds for transmittable algebras}\label{subsec:converse}

The Knill-Laflamme condition, Theorems~\ref{theorem:KL-alg} and~\ref{theorem:KL-alg2}, as well as its reformulation in terms of transmittability of algebras in Lemma~\ref{lemma:pairwise-S-orthogonal}, provide equivalent characterisations for an algebra to be transmittable via a quantum channel. However, all these criteria involve an optimisation over code spaces in order to find a suitable encoding of the algebra of interest, which can be cumbersome in practice.

In this section, we therefore provide a collection of converse bounds which can be used to show that certain algebras cannot be transmitted via a given channel. In particular, Lemmas~\ref{lemma:KrausRankConverse},~\ref{lemma:spectral-converse}, and~\ref{lem:KernelConverse} establish code-space-agnostic necessary conditions for transmittability. Furthermore, Lemma~\ref{lem:transmittable-capacity-converse} upper bounds the block dimensions of transmittable algebras by certain zero-error capacities of the quantum channel.

\begin{lemma}[Kraus-geometric converse bounds]
\label{lemma:KrausRankConverse}
Let $\Phi :\cL(\cH) \to \cL(\cK)$ be a quantum channel with Kraus representation $\Phi(\cdot) = \sum_{i=1}^n E_i (\cdot) E^*_i$.
If $\scA\cong \bigoplus^m_{k=1} \M{d_k}$ with $d_1\ge d_2\ge \cdots\ge d_m$ be transmittable via $\Phi.$ Then there exists not necessarily distinct numbers $i_1,\cdots, i_m\in[n]$ such that for all $A,J\subseteq [m]$ with $A\cap J = \emptyset$ we have
\begin{align}
\label{eq:ConverseOrthogonality}
  \sum_{k\in J}  d_k &\le \dim V_J - \left[\operatorname{dim}\left(\Pi_{V_A}V_J\right) - \dim(V_A) + \sum_{k\in A} d_k\right]_+,
\end{align}
where 
for some set $B\subseteq [m]$ we defined the space $V_B:= \operatorname{span}_{k\in B}\left(\operatorname{im}(E_{i_k})\right),$ with the $V_\emptyset = \{0\},$ and denoted the corresponding orthogonal projection by
 $\Pi_{V_B}.$

In particular we have
\begin{align}
\label{eq:MaxdkConverse} \max_{k\in[m]} d_k &\leq  \max_{i\in[n]}\operatorname{rank}(E_i) \\
   \label{eq:Sumd_kConverse}
    \sum_{\substack{k\in[m]\\d_k\ge r}} d_k &\le \operatorname{dim}\!\Big(\operatorname{span}_{i\in[n]:\operatorname{rank}(E_i)\ge r}\left(\operatorname{im}(E_i)\right)\Big) \le \sum_{\substack{i\in[n]\\ \operatorname{rank}(E_i)\ge r}} \operatorname{rank}(E_i), \qquad \forall r\ge 1.
\end{align}
Note that for $r=1$, the space in \eqref{eq:Sumd_kConverse} can be written as $\operatorname{span}_i\left(\operatorname{im}(E_i)\right)= \operatorname{supp}\left(\Phi(\1)\right).$

\end{lemma}
\begin{remark}
\label{remark:KrausSequentialConverse}
 In particular, Lemma~\ref{lemma:KrausRankConverse} gives for transmittable algebras 
$\scA = \bigoplus_{k=1}^m \M{d_k}$ that there exist not necessarily distinct $i_1,\cdots, i_m\in[n]$ such that for all $k\in[m]$
\begin{align}
\label{eq:ConverseOrthogonality2}
d_k \le \operatorname{rank}\left(E_{i_k}\right) - \left[ \operatorname{rank}\left(\Pi_{V_k} E_{i_k}\right) - \dim\left(V_k\right)+ \sum_{l=1}^{k-1} d_l \right]_+ \le \operatorname{rank}\left(E_{i_k}\right), 
\end{align}
where $V_k := \operatorname{span}_{l<k}\left(\operatorname{im}(E_{i_l})\right),$ with the convention $V_1 = \{0\}.$ This immediately follows from \eqref{eq:ConverseOrthogonality} with the choice $J=\{k\}$ and $A= \{1,\cdots, k-1\}.$

To illustrate this, assume for simplicity that $\scA=\M{d_1}\oplus \M{d_2}$ with $d_1\ge d_2$. If the same Kraus operator, say $E_1,$ is selected for both blocks, i.e. if $i_1=i_2 =1,$ then \eqref{eq:ConverseOrthogonality2} for $k=2$ implies 
$\operatorname{rank}(E_1)\ge d_1+d_2.$
If, on the other hand, two different Kraus operators are selected, say $E_1$ for the first block and $E_2$ for the second, i.e. $i_1=1$ and $i_2=2,$ then $
d_1\le \operatorname{rank}(E_1)$
and
\begin{align*}
d_2
\le
\operatorname{rank}(E_2)
-
\left[
\operatorname{rank}(E_1^*E_2)
-
\operatorname{rank}(E_1)
+
d_1
\right]_+ .
\end{align*}
Here, $\operatorname{rank}\left(E_1^*E_2\right) = \operatorname{rank}(\Pi_{\operatorname{im}(E_1)} E_2)$ measures the overlap between the outputs of $E_1$ and $E_2$, while $\operatorname{rank}(E_1)-d_1\ge 0$ is the remaining slack in the output of $E_1$ after accommodating the block $\M{d_1}$. In the saturated case $\operatorname{rank}(E_1)=d_1$, this reduces to
\begin{align}
\label{eq:KrausSaturatedTwoBlock}
d_2
\le
\operatorname{rank}(E_2)
-
\operatorname{rank}(E_1^*E_2).
\end{align}
\end{remark}
\begin{proof}[Proof of Lemma~\ref{lemma:KrausRankConverse}]
Since $\scA$ is transmittable via $\Phi$, Lemma~\ref{lemma:pairwise-S-orthogonal} implies that there exist pairwise $S_\Phi$-orthogonal quantum codes $C_1,\ldots,C_m$ with $\dim C_k=d_k$. Let $P_k$ denote the projection onto $C_k$. The Knill--Laflamme conditions imply that for every $i,k$, 
\begin{equation*}
P_kE_i^*E_iP_k=\lambda_{i,k}P_k, \quad \lambda_{i,k}\geq 0
\end{equation*}
Since $\sum_i E_i^*E_i=1_{\Hil}$, we have $\sum_i \lambda_{i,k}=1$ for each $k$. Hence, for every $k\in[m]$,$\exists i_k\in[n]$ with $\lambda_{i_k,k}>0$. 
This implies that \begin{align}
\label{eq:SomeInequality}
\operatorname{rank}(E_{i_k})\ge  \operatorname{dim}(E_{i_k}C_k) = \operatorname{rank}(E_{i_k}P_k) = d_k.
\end{align}
Furthermore, for $k\neq l$, we note that 
\begin{align*}
P_kE_i^*E_jP_l=0
\end{align*}
for all $i,j\in[n]$, from which we see that the subspaces $E_iC_k$ and $E_jC_l$ are orthogonal. Hence, since the sets $A,J\subseteq[m]$ are disjoint, we have that the spaces 
\begin{align*}
    U_A := \bigoplus_{k\in A} E_{i_k}C_k\quad\text{and}\quad U_J := \bigoplus_{k\in J} E_{i_k}C_k
\end{align*}
are also orthogonal, i.e. $U_A\perp U_J$, and further  by \eqref{eq:SomeInequality} we know that $\dim(U_A) = \sum_{k\in A} d_k$ and $\dim(U_J) = \sum_{k\in J} d_k.$ Defining $V_A:= \operatorname{span}_{k\in A}\left(\operatorname{im}(E_{i_k})\right),$ and $V_J:= \operatorname{span}_{k\in J}\left(\operatorname{im}(E_{i_k})\right),$ 
we clearly have $U_A\subseteq V_A$ and $U_J\subseteq V_J.$ From this and above orthogonality we see $U_J \subseteq V_J\cap U^\perp_A$. Therefore, 
\begin{align}
\label{eq:TooManyInequalities}
    \sum_{k\in J} d_k = \dim(U_J) \le \dim\left(V_J\cap U^\perp_A\right) = \dim(V_J) - \dim(\Pi_{U_A}V_J).
\end{align}
To lower bound $\dim(\Pi_{U_A}V_J)$, we note that since $U_A\subseteq V_A$ we have $\Pi_{U_A}V_J = \Pi_{U_A}\Pi_{V_A}V_J$ and
\begin{align*}
    \dim(\Pi_{V_A}V_J) &\le \dim(\Pi_{U_A}V_J) + \dim(\Pi_{U^\perp_A}\Pi_{V_A}V_J) \\& \le \dim(\Pi_{U_A}V_J) + \dim(  U^\perp_A\cap V_A) \\& = \dim(\Pi_{U_A}V_J) + \dim(V_A) - \dim(U_A).
\end{align*}
As $\dim\left(\Pi_{U_A} V_J\right)\ge 0$ is trivial, this gives $\dim\left(\Pi_{U_A} V_J\right) \ge \left[\dim\left(\Pi_{V_A}V_J\right) -\dim(V_A) +\dim(U_A)\right]_+$ and by 
combining this with \eqref{eq:TooManyInequalities}, we conclude \eqref{eq:ConverseOrthogonality}.

For \eqref{eq:MaxdkConverse} we choose $J = \{l\},$ and $A=\emptyset$ where $l\in[m]$ is such that $d_l = \max_{k\in[m]} d_k,$ which gives by using \eqref{eq:ConverseOrthogonality} that
\begin{align*}
 \max_{k\in[m]} d_k \le \operatorname{rank}(E_{i_l}) \le \max_{i\in[n]} \operatorname{rank}(E_i).  
\end{align*}
Furthermore, for \eqref{eq:Sumd_kConverse} we fix $r\ge 1$ and choose $J_r = \{k\in[m]: d_k\ge r\}$ and $A =\emptyset,$ which gives by using 
\eqref{eq:ConverseOrthogonality} that
\begin{align*}
\sum_{\substack{k\in[m]\\d_k\ge r}} d_k &\le \operatorname{dim}\!\Big(\operatorname{span}_{k\in[m]:d_k\ge r}\left(\operatorname{im}(E_{i_k})\right)\Big)
    \le \operatorname{dim}\!\Big(\operatorname{span}_{i\in[n]:\operatorname{rank}(E_i)\ge r}\left(\operatorname{im}(E_i)\right)\Big) ,
\end{align*}
where the second inequality follows from \eqref{eq:SomeInequality}.

\end{proof}

\begin{lemma}[Capacity converse]\label{lem:transmittable-capacity-converse}
Let $\Phi:\B{\Hil}\to \B{\Kil}$ be a quantum channel. If $\scA\cong \bigoplus_{k} \M{d_k}$ be transmittable via $\Phi$, then
\begin{align}
\log \bigl( \max_{k} d_k \bigr) \leq Q_0(\Phi), \quad
\log \left( \sum_{k} d_k \right) \leq C_0(\Phi), \quad
\log \left( \sum_{k} d_k^2 \right) \leq C^{\rm ea}_{0}(\Phi).
\end{align}
\end{lemma}

\begin{proof}
Represent $\scA$ in the standard way and let $\cP_{\scA}$ be its
trace-preserving conditional expectation. By
\eqref{eq:transmission-graph-preorder},
$S_{\cP_{\scA}}\longrightarrow S_\Phi$.
The capacity formulas for algebraic operator systems
(Lemma~\ref{lemma:alpha-algebra}) and monotonicity under graph
homomorphisms (Lemma~\ref{lemma:op-bottleneck}) give
\begin{align*}
    \max_kd_k
    &=\alpha_q(S_{\cP_{\scA}})\leq\alpha_q(S_\Phi),\\
    \sum_kd_k
    &=\alpha(S_{\cP_{\scA}})\leq\alpha(S_\Phi),\\
    \sum_kd_k^2
    &=\alpha_{\rm ea}(S_{\cP_{\scA}})\leq\alpha_{\rm ea}(S_\Phi).
\end{align*}
Taking logarithms and applying Theorem~\ref{thm:DSW} proves the
three bounds.
\end{proof}

\begin{lemma}[Full-support spectral converse]\label{lemma:spectral-converse}
Let $\Phi:\B{\Hil}\to \B{\Kil}$ be a quantum channel, and let $\scA\cong \bigoplus_{k} \M{d_k}$ with $\sum_{k} d_k=\dim \Hil$ be transmittable via $\Phi$. Then, $S_{\Phi}$ is commutative. Moreover, for every self-adjoint $X\in S_\Phi$, the eigenvalue multiplicities of $X$ form a coarsening of the shape $(d_1,\ldots,d_K)$.
\end{lemma}

\begin{proof}
By Lemma~\ref{lemma:pairwise-S-orthogonal}, there exist pairwise $S_\Phi$-orthogonal quantum codes $C_1,\ldots,C_K\subseteq \Hil$ with $\dim C_k=d_k$. Let $P_k$ be the projection onto $C_k$. Since $\sum_k d_k=\dim\Hil$, the projection $P:=\sum_k P_k =1_{\Hil}$. The orthogonality relations 
\begin{equation*}
P_kS_\Phi P_k=\C P_k
\quad \text{and} \quad
P_kS_\Phi P_l=0 \text{ for } k\neq l
\end{equation*}
therefore imply $S_\Phi\subseteq \bigoplus_k \C P_k$. Hence, every self-adjoint $X\in S_\Phi$ has the form
\begin{equation*}
    X = \sum_k \lambda_k P_k, \quad \lambda_k \in \R,
\end{equation*}
which implies that its eigenvalue multiplicities form a coarsening of the shape $(d_1,\ldots,d_K)$. 
\end{proof}

    \begin{lemma}[Kernel converse]
\label{lem:KernelConverse}
        Let $\Phi:\B{\Hil}\to \B{\Kil}$ be a quantum channel with Kraus representation $\Phi(\cdot) = \sum_{i} E_i (\cdot) E^*_i$. For every quantum code $C\subseteq \Hil$ for $\Phi$, we have
        \begin{equation}
    \label{eq:IntersectionCode}
            C \subseteq \bigcap_{ \substack{X\in S_{\Phi} \\ X\geq \,0, \,\operatorname{rank}X < \dim C }} \ker X.
        \end{equation}
        In particular if $\scA\cong \bigoplus_k\M{d_k}$ is transmittable via $\Phi$ we have for all $r\ge 1$
        \begin{align}    \label{eq:SumdkKernelBound}
        \sum_{k: d_k\geq r} d_k \le \dim\left(\bigcap_{ \substack{X\in S_{\Phi} \\ X\geq \,0, \,\operatorname{rank}X < r }} \ker X\right) \le \dim\left( \bigcap_{\substack{i\\\operatorname{rank}(E_i)< r}} \operatorname{ker}(E_i)\right)
        \end{align}
    \end{lemma}
    \begin{proof}
        Let $P_C$ be the projection onto a code space $C\subseteq \Hil$. Then, for every positive $X\in S_{\Phi}$, Knill-Laflamme conditions (Theorem~\ref{theorem:KL-alg2}) show that $P_C X P_C = \lambda P_C$ for some  $\lambda\geq 0$. If $\lambda>0$, then $\operatorname{rank}X \geq \operatorname{rank} (P_C XP_C) = \operatorname{rank} (\lambda P_C) = \dim C$. Hence, for every positive $X\in S_{\Phi}$ with $\operatorname{rank}X < \dim C$, we must have $\lambda=0$, i.e. $C\subseteq \ker X$. This proves \eqref{eq:IntersectionCode}. 
        
        For \eqref{eq:SumdkKernelBound}, fix $r\ge 1$ and consider $\scA\cong \oplus_k \M{d_k}$ that is transmittable via $\Phi$. According to Lemma~\ref{lemma:pairwise-S-orthogonal}, there exist pairwise $S_{\Phi}$-orthogonal quantum codes $C_1,\ldots,C_K$ with $\dim C_k=d_k$. 
        By \eqref{eq:IntersectionCode}, for all $k$ such that $\dim(C_k)\geq r$, we have
        \begin{align*}
     C_k \subseteq \bigcap_{ \substack{X\in S_{\Phi} \\ X\geq \,0, \,\operatorname{rank}X < r }} \ker X.
        \end{align*}
    Hence, 
    \begin{align*}
    \bigoplus_{k: \dim(C_k)\geq r} C_k \subseteq \bigcap_{ \substack{X\in S_{\Phi} \\ X\geq \,0, \,\operatorname{rank}X < r }} \ker X,
        \end{align*}
        which implies the first inequality in \eqref{eq:SumdkKernelBound}. For the second, note that for all $i$, $E^*_iE_i\in S_\Phi$, $\operatorname{rank}(E^*_iE_i) = \operatorname{rank}(E_i)$, and $\operatorname{ker}(E^*_iE_i) = \operatorname{ker}(E_i),$ which finishes the proof.
    \end{proof}

\section{Dominating transmittable algebras} \label{sec:dom}
With the definition of transmittable algebras in place, it is natural to ask if the collection of $C^*$-algebras that are transmittable via a given channel $\Phi$ has a ``largest'' element. Since different encodings/decodings may be used to transmit algebras $\scA_{\lambda}$ with different shapes $\lambda = (\lambda_1, \ldots ,\lambda_K)$ (see Eq.~\eqref{eq:A-shape}), these algebras should be compared
only up to their abstract information content. The embedding preorder (Definition~\ref{def:embedding-preorder}) provides
such a comparison: if $\scA\leq \scB$, then any information stored in $\scA$ can
also be stored inside $\scB$. This leads to two natural notions. A maximal
transmittable algebra is one which cannot be enlarged within the class of all
transmittable algebras, while a dominating transmittable algebra, when it exists,
is stronger: it is a unique transmittable algebra into which every other
transmittable algebra embeds. This motivates the following definition. 

\begin{definition}[Dominating and maximal transmittable algebras] \label{def:max-dom} \hspace{2pt} \\ Let $\Phi:\B{\Hil}\to \B{\Kil}$ be a quantum channel with the transmission set $\scT (\Phi)$ equipped with the preorder $\leq$ induced by the embedding preorder (c.f. Definitions~\ref{def:alg-transmit} and \ref{def:embedding-preorder}).

\begin{itemize}
    \item A \emph{maximal transmittable algebra} $\scA_{\max}$ for $\Phi$ is a maximal element of $\scT(\Phi)$ with respect to $\leq$, i.e. $\scA_{\max}\in \scT(\Phi)$ such that 
    \begin{equation}
        \scB \in \scT (\Phi) \quad \text{with} \quad \scA_{\max} \leq \scB \implies \scA_{\max} \cong \scB.
    \end{equation}
    We collect all maximal transmittable algebras of $\Phi$ in the set $\max (\scT(\Phi))$.

    \item A \emph{dominating transmittable algebra} $\scA_{\dom}(\Phi)$ for $\Phi$, when it exists, is a top element of $\scT(\Phi)$ with respect to $\leq$, i.e. $\scA_{\dom}(\Phi)\in \scT(\Phi)$ such that
    \begin{equation}
            \scB\in \scT(\Phi) \implies    \scB \leq \scA_{\dom}(\Phi).
    \end{equation}
\end{itemize}
\end{definition}

Similar to $\scT(\Phi)$, the elements of $\max(\scT(\Phi))$ are also understood
up to $*$-isomorphism. Since $\scT(\Phi)$ is finite up to isomorphism
(Lemma~\ref{lemma:finite-algebra-types}), every element lies below a
maximal one. Thus, $\Phi$ admits a dominating transmittable algebra
precisely when it has a unique maximal transmittable algebra up to
isomorphism, i.e., when $\abs{\max(\scT(\Phi))}=1$.
It is natural to ask whether this property is satisfied by all quantum channels. Interestingly, this is not the case, as illustrated in the following examples, where we construct channels with multiple maximal transmittable algebras.

\begin{example}[Quantum channel without dominating transmittable algebra]
\label{ex:PedagogicalExampleWithoutDomAlgebra}
\hspace{1pt} \\
\normalfont
Let $\Hil=\C^4$ with orthonormal basis $\big\{ \ket{e_0},\ket{e_1},\ket{e_2}, \ket{e_3} \big\}$, and let $\Kil=\C^3$ with orthonormal basis $\{\ket{f_0},\ket{f_1},\ket{f_2}\}$. Define $\Phi:\B{\C^4}\to\B{\C^3}$ as
$\Phi(\cdot)=\sum_{i=1}^3E_i (\cdot) E_i^*$, where
\begin{equation*}
E_1=\ket{f_0}\!\bra{e_0}+\ket{f_1}\!\bra{e_1},
    \qquad E_2=\ket{\phi_+}\!\bra{e_2},
    \qquad E_3=\ket{\phi_-}\!\bra{e_3},
\end{equation*}
and $\ket{\phi_\pm}=(\ket{f_1}\pm\ket{f_2})/\sqrt2$.
Notice that the Kraus operators $E_i$ are partial isometries with mutually orthogonal
initial spaces. 

The channel acts isometrically on
$C=\operatorname{span}\{\ket{e_0},\ket{e_1}\}$, so $\M{2}$ is
transmittable via $\Phi$. The inputs $\ket{e_0},\ket{e_2},\ket{e_3}$ give the
mutually orthogonal pure outputs $\ket{f_0},\ket{\phi_+},\ket{\phi_-}$,
so $\C^3$ is also transmittable via $\Phi$.

We now show that $\M{2}\oplus\C$ cannot be transmitted. First, we prove that $C$
is the only two-dimensional quantum code for $\Phi$. Indeed, let $C'$ be any
quantum code with $\dim C'\geq2$. Then, the Kernal converse (Lemma~\ref{lem:KernelConverse}) shows
\begin{equation*}
    C'\subseteq\ker E_2^*E_2 \cap\ker E_3^* E_3 = \ker E_2 \cap \ker E_3 = \operatorname{span}\{\ket{e_0},\ket{e_1} \} = C.
\end{equation*}
It follows that $C'=C$. Now, let $Q=\ket{f_0}\bra{f_0}+\ket{f_1}\bra{f_1}$ project onto the
output of this qubit code. For any state $\rho$ supported
on $C^\perp=\operatorname{span}\{\ket{e_2},\ket{e_3}\}$,
\begin{equation*}
    \Phi(\rho)
    =\langle e_2|\rho|e_2\rangle\ket{\phi_+}\bra{\phi_+}
     +\langle e_3|\rho|e_3\rangle\ket{\phi_-}\bra{\phi_-}, \quad\text{with}
    \quad \Tr\bigl(Q\Phi(\rho)\bigr)=\tfrac12.
\end{equation*}
Hence, no such output is orthogonal to the entire qubit output space.
Lemma~\ref{lemma:pairwise-S-orthogonal} then rules out
$\M{2}\oplus\C$, since its qubit block would have to use $C$, leaving
no perfectly distinguishable additional classical block. Finally, the Kraus-rank converse (Lemma~\ref{lemma:KrausRankConverse}) gives
$\max_kd_k\leq2$ and $\sum_kd_k\leq3$ for any transmittable algebra
$\bigoplus_k\M{d_k}$. Consequently,
\begin{equation}
\label{eq:MaxTransEx1}
    \max(\scT(\Phi))=\{\M{2},\C^3\},
\end{equation}
and $\Phi$ has no
dominating transmittable algebra.
\end{example}

The next example has (input, output) dimensions $(3,4)$,
reversing the dimensions in
Example~\ref{ex:PedagogicalExampleWithoutDomAlgebra}.
Proposition~\ref{prop:minimal-dimensions} below shows that these two (input, output)
dimension pairs, $(4,3), (3,4)$, are minimal for channels without a dominating
transmittable algebra.

\begin{example}[Quantum channel without dominating transmittable algebra]
\label{ex:NoDomSecondExample}
\hspace{1pt} \\
\normalfont Let $\Hil=\C^3$ with orthonormal basis $\big\{\ket{e_0},\ket{e_1},\ket{e_2} \big\}$, and let $\Kil=\C^4$ with orthonormal basis $\big\{\ket{f_0},\ket{f_1},\ket{f_2},\ket{f_3}\big\}$. Define $\Phi:\B{\C^3}\to \B{\C^4}$ as $\Phi(\cdot)=E_0(\cdot) E_0^*+E_1 (\cdot) E_1^*$, where
\begin{equation*}
E_0=\ket{f_0}\!\bra{e_0}+\frac{1}{\sqrt{2}}\ket{f_1}\!\bra{e_1},
\qquad
E_1=\frac{1}{\sqrt{2}}\ket{f_2}\!\bra{e_1}+\ket{f_3}\!\bra{e_2}.
\end{equation*}
Then, $E_0^*E_0=\ket{e_0}\!\bra{e_0}+\frac12\ket{e_1}\!\bra{e_1},$ $
E_1^*E_1=\frac12\ket{e_1}\!\bra{e_1}+\ket{e_2}\!\bra{e_2},$ and
$E_0^*E_1=0$, so that 
\begin{equation*}
S_\Phi=\operatorname{span} \{E_0^*E_0,E_1^*E_1 \}
=
\operatorname{span} \{ \iden_{\Hil}, D:=E_0^*E_0-E_1^*E_1 \}.
\end{equation*}

We first show that $\M{2}$ is transmittable via $\Phi$. Let $\ket{\eta_0}=\ket{e_1}$ and $\ket{\eta_1}=(\ket{e_0}+\ket{e_2})/\sqrt{2}$, and set $C=\operatorname{span}\{\ket{\eta_0},\ket{\eta_1} \}$. If $P_C$ denotes the projection onto $C$, then $P_CDP_C=0$, which implies $P_CS_\Phi P_C=\C P_C$. Thus, $C$ is a two-dimensional quantum code, so $\M{2}$ is transmittable. Next, $\C^3$ is clearly transmittable via $\Phi$, since the pure inputs ${\ket{e_0},\ket{e_1},\ket{e_2}}$ have pairwise orthogonal output supports. 

To identify all maximal transmittable algebras, let
$\scA\cong\bigoplus_k\M{d_k}$ be transmittable via $\Phi$. By definition, $\sum_k d_k\leq \dim \Hil =  3$ (see Lemma~\ref{lemma:finite-algebra-types}). At equality, the spectral converse
(Lemma~\ref{lemma:spectral-converse}) applies. In particular, since $D = \ket{e_0}\!\bra{e_0} - \ket{e_2}\!\bra{e_2}\in S_{\Phi}$ is self-adjoint and has eigenvalues $1,0,-1$, each with multiplicity one,  it implies that both $\M{2}\oplus\C$ and $\M{3}$ cannot be transmitted via $\Phi$. Thus, the only
transmittable algebra with $\sum_kd_k=3$ is $\C^3$. Every remaining transmittable algebra
has $\sum_kd_k\leq2$ and therefore embeds into $\M{2}$. Consequently,
\begin{equation}
\label{eq:MaxTransEx2}
    \max(\scT(\Phi))=\{\M{2},\C^3\},
\end{equation}
and $\Phi$ has no dominating transmittable algebra. 
\end{example}

In fact, both examples have the same full transmission set, up to
$*$-isomorphism:
\begin{equation*}
    \scT(\Phi)=\{\C,\C^2,\C^3,\M{2}\}.
\end{equation*}
This is the transmission set depicted in the right panel of
Figure~\ref{fig:max-vs-dom}.

\begin{proposition}[Dominating algebra exists in small dimensions]
\label{prop:minimal-dimensions}
Every channel $\Phi:\B{\Hil}\to \B{\Kil}$ with $\min \{\dim \Hil, \dim \Kil \} \leq 2$ or $\dim \Hil = \dim \Kil =3$ admits a dominating transmittable algebra.
\end{proposition}

\begin{proof}
If $\scA\cong\bigoplus_k\M{d_k}$ is transmittable via a channel $\Phi:\B{\Hil}\to \B{\Kil}$, then
Lemma~\ref{lemma:finite-algebra-types} and the $r=1$ case of
Lemma~\ref{lemma:KrausRankConverse} give
\begin{equation*}
    \sum_kd_k\leq\min\{\dim\Hil,\dim\Kil\}.
\end{equation*}
If either dimension is at most two, the possible algebra types
belong to the chain
\begin{equation*}
    \C\leq\C^2\leq\M{2}.
\end{equation*}
The transmission set therefore has a unique maximal element,
which is dominating.

Now, consider $\Phi:\M{3}\to\M{3}$. The possible algebra types are
\begin{equation*}
\C,\quad\C^2,\quad\M{2},\quad\C^3,\quad
    \M{2}\oplus\C,\quad\M{3}.
\end{equation*}
The only incomparable pair in this list is $\M{2},\C^3$.
Since $\M{2}\oplus\C$ dominates both, it suffices to prove that
transmittability of $\M{2}$ and $\C^3$ implies transmittability
of $\M{2}\oplus\C$.

By Lemma~\ref{lemma:pairwise-S-orthogonal}, transmittability of
$\C^3$ supplies an orthonormal input basis
$\{\ket{u_a}\}_{a=1}^3$ with pairwise orthogonal output supports.
The output dimension is three, so these outputs are pure and
form an orthonormal basis $\{\ket{v_a}\}_{a=1}^3$.
Every Kraus operator therefore has the form
\begin{equation}\label{eq:Ei-diagonal} E_i=\sum_{a=1}^3z_{ia}\ket{v_a}\!\bra{u_a}.
\end{equation}
Indeed, the rank-one output
$\Phi(\ket{u_a}\bra{u_a})=\ket{v_a}\bra{v_a}$ forces
each vector $E_i\ket{u_a}$ to be proportional to $\ket{v_a}$. Let $C$ be a two-dimensional quantum code, with projection $P_C$.
The Knill--Laflamme conditions give
$P_CE_i^*E_jP_C=\lambda_{ij}P_C$, where
$\Lambda=(\lambda_{ij})$ is positive semidefinite with trace one.
After a unitary mixing of the Kraus operators that diagonalizes
$\Lambda$, the resulting family $\{F_i\}$ satisfies
\begin{equation*}
    P_CF_i^*F_jP_C=\delta_{ij}p_iP_C,
    \qquad p_i\geq0,\qquad\sum_ip_i=1.
\end{equation*}
For each $p_i>0$, the map $p_i^{-1/2}F_i|_C$ is an isometry,
and the corresponding two-dimensional output ranges are mutually
orthogonal. Since the output dimension is three, exactly one
$p_i$ is positive. After relabelling,
\begin{equation*}
    p_1=1,\qquad F_iP_C=0\quad \text{for } i\geq2.
\end{equation*}

If every $F_i$ with $i\geq2$ vanishes, trace preservation gives
$F_1^*F_1=1_{\C^3}$, so $\Phi$ is unitary and transmits $\M{3}$,
and hence $\M{2}\oplus\C$.
Otherwise, choose a nonzero $F_i$ with $i\geq2$.
It annihilates the two-dimensional space $C$, so it has rank one.
Unitary mixing preserves the displayed diagonal form \eqref{eq:Ei-diagonal} of the
Kraus operators, and thus $F_i$ has exactly one nonzero
coefficient, at some index $a$. Consequently,
\begin{equation*}
    C=\ker F_i
     =\operatorname{span}\{\ket{u_b}:b\neq a\}.
\end{equation*}
The output of $C$ lies in
$\operatorname{span}\{\ket{v_b}:b\neq a\}$, whereas
$\ket{u_a}$ produces $\ket{v_a}$.
Thus $C$ and $\C\ket{u_a}$ have orthogonal output supports and
are $S_\Phi$-orthogonal. Lemma~\ref{lemma:pairwise-S-orthogonal}
shows that $\M{2}\oplus\C$ is transmittable, and the proof is complete.
\end{proof}

\subsection{From hybrid capacities to dominating algebras}
\label{sec:hybrid-capacities}

In this section, we establish necessary conditions for a quantum channel to admit a dominating transmittable algebra in terms of hybrid classical-quantum capacities. More precisely, Corollory~\ref{cor:CapacitiesToAlgebra} below shows that if a dominating transmittable algebra exists, it is uniquely and explicitly determined by these hybrid capacities.

For a quantum channel $\Phi : \cL(\cH) \to \cL(\cK)$, let $d_{\max}(\Phi)$ be the maximal dimension of a quantum system that can be transmitted via $\Phi$:
\begin{align}
\label{eq:dmax}
 d_{\max}(\Phi) &:=\max\left\{d\in\N \, \Big|\ \M{d}\  \text{trasmittable via $\Phi$}\ \right\} = 2^{Q_0(\Phi)},
\end{align}
where $Q_0(\Phi)$ is the one-shot zero error quantum capacity defined in \eqref{eq:ZeroErrorQuantumCapacity}, and the equality follows from Lemma~\ref{lemma:EncDecTransmittableAlgebra}. Similarly, for $d\in\{1,\cdots,d_{\max}(\Phi)\}$, we define the \emph{hybrid classical-quantum capacities} $c_d(\Phi)\in \mathbb{N}$ as the maximal number of classical messages which can be sent simultaneously along with a $d$-dimensional quantum state\footnote{These are the one-shot zero-error analogue of the simultaneous asymptotic rate region studied in \cite{DevetakShor2005}.}:
\begin{align}
 \label{eq:c_dDef}
c_d(\Phi):= \max\left\{c\in\N \, \Big|\ \M{d}^{\oplus c}\  \text{transmittable via $\Phi$} \ \right\},
 \end{align}
More explicitly, by Lemma~\ref{lemma:EncDecTransmittableAlgebra},
$\M{d}^{\oplus c}$ is transmittable if and only if there exist
encoding and decoding channels
$\cE:\B{\C^{cd}}\to\B{\Hil}$ and
$\cD:\B{\Kil}\to\B{\C^{cd}}$ such that
\begin{equation}
    \cD\circ\Phi\circ\cE
    =\cP_{\M{d}^{\oplus c}},
    \label{eq:QCCapacitiesAlgebraTransmission}
\end{equation}
where $\M{d}^{\oplus c}$ is represented in the standard way on
$\bigoplus_{a=1}^c\C^d$, and $\cP_{\M{d}^{\oplus c}}$ is the unique trace-preserving conditional expection onto $\M{d}^{\oplus c}$.

Notice that $c_1(\Phi)=2^{C_0(\Phi)}$. The maximum in
\eqref{eq:c_dDef} is over a nonempty finite set: at least one
$d$-dimensional quantum code exists, and
Lemmas~\ref{lemma:finite-algebra-types} and
\ref{lemma:KrausRankConverse} give
\begin{equation}
    d\,c_d(\Phi)
    \leq\min\{\dim\Hil,\operatorname{rank}\Phi(1_{\Hil})\}.
    \label{eq:hybrid-dimension-bound}
\end{equation}
The function $d\mapsto c_d(\Phi)$ is non-increasing: the algebra
$\M{d_2}^{\oplus c}$ embeds into $\M{d_1}^{\oplus c}$ whenever
$d_2\leq d_1$. More generally, for $n,d\in\N$ with
$nd\leq d_{\max}(\Phi)$, we have $c_d(\Phi)\geq n\,c_{nd}(\Phi),$
since $\M{d}^{\oplus n}\leq\M{nd}$. An $nd$-dimensional
quantum system can encode one of $n$ classical labels together
with a $d$-dimensional quantum state. Applying this embedding in
each of the $c_{nd}(\Phi)$ blocks proves the inequality.

In the following, we define an algebra, which we call the \emph{capacity algebra} $\scA_{\operatorname{cap}}(\Phi)$, whose block structure is compatible with all capacities $c_d(\Phi)$, in the sense that it is the unique algebra into which exactly $c_d(\Phi)$ copies of $\M{d}$ can be packed for every $d$. In particular, Theorem~\ref{thm:CapacitiesToAlgebra} below shows that if there exists a transmittable algebra into which all $\M{d}^{\oplus c_d(\Phi)}$ can be embedded, then this algebra is necessarily the capacity algebra. From this we conclude in Corollary~\ref{cor:CapacitiesToAlgebra} that whenever $\Phi$ admits a dominating algebra, it must likewise coincide with the capacity algebra. 

For its construction, we first record how many equal matrix blocks can embed into a given algebra. For a nonzero $C^*$-algebra $\scB$ and $d\in\N$, define
\begin{equation}
    b_d(\scB)
    :=\max\left(\{0\}\cup
       \{L\in\N:\M{d}^{\oplus L}\leq\scB\}\right).
    \label{eq:MaxEmbeddingMany}
\end{equation}
The value zero means that no such block embeds. The following lemma is the equal-block
case of the finite-dimensional embedding problem
\cite{Kuperberg2003hybrid, Korte2018}.

\begin{lemma}[Embedding equal matrix blocks]
\label{lem:MaximalNumberOfBlocksToEmbedd}
Let $\scB\cong\bigoplus_{j=1}^J\M{n_j}$ be a nonzero
$C^*$-algebra. For every $d\in\N$,
\begin{equation}
    b_d(\scB)=\sum_{j=1}^J\floor{\frac{n_j}{d}}.
    \label{eq:LmaxMany}
\end{equation}
Equivalently, if
$\scB\cong\bigoplus_{m=1}^R\M{m}^{\oplus\widetilde N_m}$
with $\widetilde N_m\in\N_0$ and at least one nonzero
multiplicity,
\begin{equation}
    b_d(\scB)
    =\sum_{m=1}^R\widetilde N_m\floor{\frac{m}{d}}.
\end{equation}
\end{lemma}

\begin{proof}
For achievability, place $\floor{n_j/d}$ mutually orthogonal
$d$-dimensional subspaces inside the $j$th target block, and use
one source block on each subspace. This gives an embedding with
$\sum_j\floor{n_j/d}$ source blocks whenever this sum is positive.
When it is zero, the lower bound follows from the definition of
$b_d$.

Conversely, suppose $\M{d}^{\oplus L}\leq\scB$ with $L\in\N$.
By \eqref{eq:embedding-block-form}, its embedding has a matrix of
nonnegative integer multiplicities $\Lambda=(\Lambda_{\ell j})$
satisfying
\begin{equation*}
    \sum_{j=1}^J\Lambda_{\ell j}\geq1, \quad\text{for } \ell\in[L],
    \qquad
    d\sum_{\ell=1}^L\Lambda_{\ell j}\leq n_j
  \quad\text{for }j\in[J].
\end{equation*}
Consequently,
\begin{equation*}
    L\leq\sum_{j=1}^J\sum_{\ell=1}^L\Lambda_{\ell j}
     \leq\sum_{j=1}^J\floor{\frac{n_j}{d}}.
\end{equation*}
Grouping target blocks of equal size gives the formula in terms of $\widetilde N_m$.
\end{proof}

This lemma explains the recursive reconstruction from hybrid
capacities given below. We define integers
$N_d(\Phi)$ in descending order by
\begin{align}
    N_{d_{\max}(\Phi)}(\Phi)&:=c_{d_{\max}(\Phi)}(\Phi),\nonumber\\
    N_d(\Phi)&:=c_d(\Phi)
      -\sum_{m=d+1}^{d_{\max}(\Phi)}N_m(\Phi)\floor{\frac{m}{d}},
      \qquad 1\leq d<d_{\max}(\Phi).
    \label{eq:DefN(d)}
\end{align}
These integers need not be nonnegative.\footnote{For example,
consider a channel $\Phi:\M{7}\to\M{7}$ in the standard basis
$\{\ket{j}\}_{j=0}^6$, with Kraus operators
$E_1=\kb{0}+\kb{1}+\kb{2}$,
$E_2=\ket{0}\bra{3}+\ket{1}\bra{4}$, and
$E_3=\ket{2}\bra{5}+\kb{6}$.
The first Kraus operator transmits a qutrit on
$\operatorname{span}\{\ket{0},\ket{1},\ket{2}\}$.
The other two transmit qubits on
$\operatorname{span}\{\ket{3},\ket{4}\}$ and
$\operatorname{span}\{\ket{5},\ket{6}\}$ into the orthogonal output spaces $\operatorname{span}\{\ket{0},\ket{1}\}$ and $\operatorname{span}\{\ket{2},\ket{6}\}$.
Thus $\M{3}, \M{2}^{\oplus2} \in \scT(\Phi)$. The maximal Kraus rank is three and
$\operatorname{rank}\Phi(1_{\C^7})=4$, so
Lemma~\ref{lemma:KrausRankConverse} and
\eqref{eq:hybrid-dimension-bound} give
$d_{\max}=3$ and $(c_1,c_2,c_3)=(4,2,1)$.
The recursion yields $N_3=1$, $N_2=1$, and $N_1=-1$.}

If $N_d(\Phi)\geq0$ for every $d\in[d_{\max}(\Phi)]$, define the
\emph{capacity algebra} of $\Phi$ by
\begin{equation}
    \scA_{\text{cap}}(\Phi)
    :=\bigoplus_{d=1}^{d_{\max}(\Phi)}\M{d}^{\oplus N_d(\Phi)}, \label{eq:DefCapacityAlgebra}
\end{equation}
again omitting all summands of zero multiplicity. It is nonzero
because $N_{d_{\max}(\Phi)}(\Phi)=c_{d_{\max}(\Phi)}(\Phi)\geq1$. By construction and
Lemma~\ref{lem:MaximalNumberOfBlocksToEmbedd},
\begin{equation}
b_d\bigl(\scA_{\text{cap}}(\Phi)\bigr)=c_d(\Phi),
    \qquad d\in[d_{\max}(\Phi)].
    \label{eq:capacity-algebra-profile}
\end{equation}

\begin{theorem}
\label{thm:CapacitiesToAlgebra}
Let $\Phi$ be a quantum channel.
Suppose that a transmittable $C^*$-algebra $\scA$ satisfies
\begin{equation}
    \M{d}^{\oplus c_d(\Phi)}\leq\scA
    \qquad\text{for every }d\in[d_{\max}(\Phi)].
    \label{eq:DefAInCapProof}
\end{equation}
Then $N_d(\Phi)\geq0$ for every $d\in[d_{\max}(\Phi)]$, and
\begin{equation}
    \scA\cong\scA_{\text{cap}}(\Phi).
    \label{eq:A_max=A_capinProof}
\end{equation}
\end{theorem}

\begin{proof}
Since $\scA$ is transmittable, every simple block of $\scA$
has size at most $d_{\max} (\Phi)$. Hence
\begin{equation*}
    \scA\cong\bigoplus_{m=1}^{d_{\max}(\Phi)}\M{m}^{\oplus\widetilde N_m},
    \qquad\widetilde N_m\in\N_0.
\end{equation*}
Hypothesis~\eqref{eq:DefAInCapProof} at $d=d_{\max}(\Phi)$ also ensures that
a block of size $d_{\max}(\Phi)$ occurs. For each $d\in[d_{\max}(\Phi)]$, the same hypothesis
gives $c_d(\Phi)\leq b_d(\scA)$. Conversely, every algebra
embedding into $\scA$ is transmittable via $\Phi$, so
$b_d(\scA)\leq c_d(\Phi)$. Thus, $c_d(\Phi)= b_d(\scA)$, and
Lemma~\ref{lem:MaximalNumberOfBlocksToEmbedd} shows
\begin{equation}
    c_d(\Phi)
    =\sum_{m=d}^{d_{\max}(\Phi)}\widetilde N_m\floor{\frac{m}{d}},
    \qquad d\in[d_{\max}(\Phi)].
    \label{eq:MaximalL}
\end{equation}
The equation at $d=d_{\max}(\Phi)$ gives $\widetilde N_{d_{\max}(\Phi)}=c_{d_{\max}(\Phi)}(\Phi)=N_{d_{\max}(\Phi)}(\Phi)$.
For smaller $d$, solving for $\widetilde N_d$ gives exactly the
recursion~\eqref{eq:DefN(d)}. Descending induction therefore
shows $\widetilde N_d=N_d(\Phi)$ for every $d$.
This proves nonnegativity and the claimed isomorphism.
\end{proof}

\begin{corollary}
\label{cor:CapacitiesToAlgebra}
If a quantum channel $\Phi$ admits a dominating transmittable
algebra, then $N_d(\Phi)\geq0$ for every
$d\in[d_{\max}(\Phi)]$, and
\begin{equation*}
    \scA_{\dom}(\Phi)\cong\scA_{\text{cap}}(\Phi).
\end{equation*}
More generally, whenever $\scA_{\text{cap}}(\Phi)$ is defined and
transmittable, it is maximal.
\end{corollary}

\begin{proof}
A dominating algebra contains every transmittable algebra
$\M{d}^{\oplus c_d(\Phi)}$, so the first statement follows from
Theorem~\ref{thm:CapacitiesToAlgebra}.

For the second statement, suppose that the capacity algebra is
defined and transmittable, and let $\scB\in\scT(\Phi)$ satisfy
$\scA_{\text{cap}}(\Phi)\leq\scB$.
Then \eqref{eq:capacity-algebra-profile} implies
$\M{d}^{\oplus c_d(\Phi)}\leq\scA_{\text{cap}}(\Phi)\leq\scB$
for every $d\in[d_{\max}(\Phi)]$.
Theorem~\ref{thm:CapacitiesToAlgebra} therefore gives
$\scB\cong\scA_{\text{cap}}(\Phi)$, proving maximality.
\end{proof}

The logical implications are consequently
\begin{align}
    &\Phi\text{ admits a dominating transmittable algebra}\\
    &\qquad\implies
      \scA_{\text{cap}}(\Phi)\text{ is defined and transmittable} \label{eq:dom->cap} \\
    &\qquad\implies
      N_d(\Phi)\geq0\text{ for every }d\in[d_{\max}(\Phi)].
\end{align}
Neither implication is reversible. For the second, the channels in
Examples~\ref{ex:PedagogicalExampleWithoutDomAlgebra}
and~\ref{ex:NoDomSecondExample} have
$(c_1,c_2)=(3,1)$ and $(N_1,N_2)=(1,1)$, but their capacity
algebra $\M{2}\oplus\C$ is not transmittable.
The first implication fails in reverse for the following channel.

\begin{remark}
\label{rem:capacity-without-domination}
Consider $\Phi:\M{6}\to\M{6}$ in the standard orthonormal
basis $\{\ket{j}\}_{j=0}^5$, with
$\Phi(X)=E_0XE_0^*+E_1XE_1^*$ and
\begin{equation*}
    E_0=\sum_{j=0}^3\ket{j}\bra{j},
    \qquad
    E_1=\ket{3}\bra{4}+\ket{4}\bra{5}.
\end{equation*}

The channel acts as
identity on $C_4=\operatorname{span}\{\ket{0},\ket{1},\ket{2},\ket{3}\},$ while $\ket{5}$ produces $\ket{4}$, which is orthogonal to 
$C_4$. These two code spaces satisfy
Lemma~\ref{lemma:pairwise-S-orthogonal}. Hence,  $\M{4}\oplus\C$ is transmittable. Moreover, $\M{3}\oplus\M{2}$ is also transmittable. Use
$C_3=\operatorname{span}\{\ket{0},\ket{1},\ket{2}\}$ for the
qutrit and $C_2=\operatorname{span}\{\ket{4},\ket{5}\}$ for the
qubit. The channel acts isometrically on each, with orthogonal
output spaces $C_3$ and
$\operatorname{span}\{\ket{3},\ket{4}\}$, respectively.

Since the Kraus ranks are four and two, and
$\operatorname{rank}\Phi(1_{\C^6})
=5$, Lemma~\ref{lemma:KrausRankConverse} and
\eqref{eq:hybrid-dimension-bound} show
$d_{\max}(\Phi)\leq4$ and $d\,c_d(\Phi)\leq5$.
Moreover, $\M{4}\oplus\C$ is transmittable, so we obtain
\begin{equation*}
    d_{\max}(\Phi)=4,
    \qquad(c_1,c_2,c_3,c_4)=(5,2,1,1).
\end{equation*}
Thus, $N_4=N_1=1$, $N_2=N_3=0$, and
\begin{equation*}
    \scA_{\text{cap}}(\Phi)=\M{4}\oplus\C.
\end{equation*}

However, $\M{3}\oplus\M{2}$ does not embed into this capacity
algebra. Hence,
Corollary~\ref{cor:CapacitiesToAlgebra} shows that
$\Phi$ has no dominating transmittable algebra.
\end{remark}

\subsubsection{A finite characterization by forbidden algebra types}
\label{sec:forbidden-algebra-types}

We have seen in the previous section that the existence and transmittablity of the capacity algebra $\scA_{\rm cap}(\Phi)$, while being necessary (see \eqref{eq:dom->cap}), is not sufficient to ensure the existence of dominating transmittable algebra.
In this subsection, we identify the precise extra conditions beyond \eqref{eq:dom->cap} for a dominating transmittable algebra to exist.

For a nonzero $C^*$-algebra
$\scA\cong\bigoplus_{k=1}^K\M{d_k}$, we define
\begin{equation}
    \ell(\scA):=\sum_{k=1}^K d_k=b_1(\scA).
    \label{eq:algebra-faithful-size}
\end{equation}

\begin{definition}[Forbidden algebra types]
\label{def:forbidden-algebra-types}
For a nonzero $C^*$-algebra $\scA$, let
\begin{equation}
    \mathfrak F(\scA)
    :=\min_{\leq}\{\scB:\scB\not\leq\scA\},
    \label{eq:forbidden-algebra-types}
\end{equation}
where the minimum denotes the set of minimal elements, and all
algebras are nonzero and considered up to
$*$-isomorphism. In other words, $\scB\in\mathfrak F(\scA)$ if
$\scB\not\leq\scA$ but every $\scC<\scB$ embeds into $\scA$.
\end{definition}

\begin{lemma}[Finite obstruction bound]
\label{lem:finite-forbidden-types}
Let $\scA$ be a nonzero $C^*$-algebra. Then, every
$\scB\in\mathfrak F(\scA)$ satisfies $\ell(\scB)\leq \ell(\scA)+1$, where equality holds iff $\scB \cong \C^{\ell (\scA) +1}$. Hence, $\mathfrak F(\scA)$ is finite. Moreover, for every nonzero $C^*$-algebra $\scC$,
\begin{equation}
    \scC\not\leq\scA
    \quad\Longleftrightarrow\quad
    \exists\scB\in\mathfrak F(\scA)\text{ with }\scB\leq\scC.
    \label{eq:forbidden-type-witness}
\end{equation}
\end{lemma}

\begin{proof}
Notice that $\C^{\ell(\scA)+1}$ clearly cannot be embedded into $\scA$, whereas every
proper algebra type below it is $\C^m$ for some $m\leq \ell(\scA)$, and hence embeds into $\scA$. Consequently, $\C^{\ell (\scA)+1}\in \mathfrak F (\scA)$. Now let $\scB\in\mathfrak F(\scA)$. If $\ell(\scB)\geq \ell (\scA)+1$,
then
\begin{equation*}
    \C^{\ell (\scA)+1}\leq\C^{\ell(\scB)}\leq\scB.
\end{equation*}
Thus, minimality forces $\scB\cong\C^{\ell (\scA)+1}$. Consequently, every $\scB \in \mathfrak F(\scA)$ other than $\C^{\ell (\scA)+1}$ satisfies $\ell (\scB) \leq \ell (\scA)$. This proves the claimed bound and uniqueness.

Finally, suppose $\scC\not\leq\scA$. Among the algebra types
$\scB\leq\scC$ that do not embed into $\scA$, choose a minimal
one. Such a choice exists: the collection contains $\scC$
and is finite because $\ell(\scB)\leq\ell(\scC)$.
Every proper algebra type below this $\scB$ also lies below
$\scC$, so minimality ensures that it embeds into $\scA$.
Hence $\scB\in\mathfrak F(\scA)$. Conversely, if $\scB\in\mathfrak F(\scA)$ and $\scB\leq\scC$,
then $\scC\leq\scA$ would imply $\scB\leq\scA$, a contradiction.
This proves \eqref{eq:forbidden-type-witness}.
\end{proof}

This lemma shows that $\mathfrak F(\scA)$ can be computed from the shape of $\scA$
alone. We first enumerate the integer partitions of $1,\ldots,\ell (\scA)+1$, then use the packing criterion \eqref{eq:embedding-block-form} to discard the types embedding into $\scA$, and retain the minimal remaining types. Once we have a description of $\mathfrak F(\scA)$, we can prove that a candidate transmittable algebra $\scA\in \scT(\Phi)$ is dominating if and only if none of the forbidden algebra types in $\mathfrak F(\scA)$ are transmittable.

\begin{theorem}[Finite characterization of dominating algebras]
\label{thm:dominating-forbidden-types} \hspace{3pt} \\
Let $\Phi:\B{\Hil}\to\B{\Kil}$ be a quantum channel and let
$\scA$ be a nonzero $C^*$-algebra. Then,
\begin{equation}
    \scA_{\dom}(\Phi)\cong\scA
    \quad\Longleftrightarrow\quad
    \scA\in\scT(\Phi)
    \ \text{and}\
    \mathfrak F(\scA)\cap\scT(\Phi)=\varnothing.
    \label{eq:dominating-forbidden-types}
\end{equation}
Here, the left-hand side includes the assertion that a dominating
algebra exists. 
\end{theorem}

\begin{proof}
If $\scA$ is dominating, it is transmittable and every transmittable
algebra embeds into it. No member of $\mathfrak F(\scA)$ can
therefore be transmittable.

Conversely, suppose the right-hand side holds. If some
$\scC\in\scT(\Phi)$ did not embed into $\scA$,
Lemma~\ref{lem:finite-forbidden-types} would give
$\scB\in\mathfrak F(\scA)$ with $\scB\leq\scC$.
Downward closure of $\scT(\Phi)$ would make $\scB$ transmittable,
a contradiction. Thus $\scA$ dominates $\scT(\Phi)$.
\end{proof}

The above characterization is a finite collection of Knill--Laflamme
feasibility conditions. For
$\scB=\bigoplus_{j=1}^J\M{b_j}$ and a basis
$\{s_\mu\}_\mu$ of $S_\Phi$, Lemma~\ref{lemma:pairwise-S-orthogonal}
gives $\scB\in\scT(\Phi)$ precisely when there exist isometries
$V_j:\C^{b_j}\to\Hil$ and scalars $\lambda_{\mu,j}$ such that
\begin{align}
   \nn  V_j^*V_j&=1_{\C^{b_j}},&
    V_i^*V_j&=0\quad(i\neq j),\\
    V_j^*s_\mu V_j&=\lambda_{\mu,j}1_{\C^{b_j}},&
    V_i^*s_\mu V_j&=0\quad(i\neq j)
    \label{eq:forbidden-type-KL}
\end{align}
for every $\mu$. Theorem~\ref{thm:dominating-forbidden-types}
requires feasibility for a candidate algebra $\scA$ and infeasibility for each member of
$\mathfrak F(\scA)$. For a fixed input/output dimension, there are also
only finitely many possible dominating shapes. This gives an exact finite criterion, but it does not imply an
efficient algorithm for finding or excluding the required codes. 

Some forbidden lists, obtained from
the embedding criterion \eqref{eq:algebra-packing-assignment}, are shown in
Table~\ref{tab:forbidden-algebra-types}.
\begin{table}[htbp]
    \centering
    \renewcommand{\arraystretch}{1.25}
    \begin{tabular}{c|c}
        Proposed dominating algebra $\scA$
        & Minimal forbidden types $\mathfrak F(\scA)$\\
        \hline
        $\M{d}$, $d\geq1$
        & $\{\C^{d+1}\}$\\
        $\C^c$, $c\geq2$
        & $\{\M{2},\C^{c+1}\}$\\
        $\M{2}\oplus\C$
        & $\{\M{3},\C^4\}$\\
        $\M{3}\oplus\C$
        & $\{\M{2}\oplus\M{2},\C^5\}$\\
        $\M{3}\oplus\M{2}$
        & $\{\M{4},\C^6\}$\\
        $\M{4}\oplus\C$
        & $\{\M{3}\oplus\M{2},\C^6\}$
    \end{tabular}
    \caption{A proposed algebra is dominating exactly when it is
    transmittable and none of its minimal forbidden types is
    transmittable. For $\scA=\C$, the forbidden list is $\{\C^2\}$.}
    \label{tab:forbidden-algebra-types}
\end{table}

Combining the finite obstruction criterion with the reconstruction
from hybrid capacities leaves only one candidate algebra to test, giving us the following corollary.

\begin{corollary}[Complete criterion from hybrid capacities]
\label{cor:complete-capacity-criterion}
A quantum channel $\Phi$ admits a dominating transmittable algebra
if and only if the following conditions hold:
\begin{enumerate}
    \item $N_d(\Phi)\geq0$ for every $1\leq d\leq d_{\max}(\Phi)$,
    so $\scA_{\text{cap}}(\Phi)$ is defined,
    \item $\scA_{\text{cap}}(\Phi)\in\scT(\Phi)$ is transmittable,
    \item no algebra in
    $\mathfrak F(\scA_{\text{cap}}(\Phi))$ is transmittable via
    $\Phi$.
\end{enumerate}
When these conditions hold,
$\scA_{\dom}(\Phi)\cong\scA_{\text{cap}}(\Phi)$.
\end{corollary}

\begin{proof}
Necessity follows from Corollary~\ref{cor:CapacitiesToAlgebra}
and Theorem~\ref{thm:dominating-forbidden-types}. The latter
theorem also proves sufficiency.
\end{proof}

\subsubsection{Small quantum blocks}
\label{sec:small-quantum-blocks}

For blocks of dimension at most three, the $b_d(\cdot)$ functions from \eqref{eq:MaxEmbeddingMany} already characterize algebra
embeddability. This substantially simplifies
Corollary~\ref{cor:complete-capacity-criterion}.

\begin{lemma}[Embedding blocks of size at most three]
\label{lem:packing-three}
Let $\scA$ and $\scB$ be non-zero $C^*$-algebras with matrix blocks of size at most three.
Then
\begin{equation}
    \scB\leq\scA
    \quad\Longleftrightarrow\quad
    b_d(\scB)\leq b_d(\scA)\quad(d=1,2,3).
    \label{eq:packing-three}
\end{equation}
\end{lemma}

\begin{proof}
Necessity follows from transitivity of embeddings. Write
$\scA=\bigoplus_{d=1}^3\M{d}^{\oplus n_d}$ and
$\scB=\bigoplus_{d=1}^3\M{d}^{\oplus m_d}$, allowing zero
multiplicities. By \eqref{eq:LmaxMany}, the three inequalities are
\begin{align*}
    m_3&\leq n_3,\\
    m_2+m_3&\leq n_2+n_3,\\
    m_1+2m_2+3m_3&\leq n_1+2n_2+3n_3.
\end{align*}
Place each three-dimensional source block in a distinct
three-dimensional target block. There are then $n_2+n_3-m_3$
target blocks that can receive a two-dimensional block, so the
second inequality allows all such blocks to be placed. Finally,
the third inequality guarantees enough remaining dimensions for
the scalar blocks. This constructs the assignment in
\eqref{eq:algebra-packing-assignment}.
\end{proof}

\begin{corollary}[Complete criterion for $d_{\max}\leq3$]
\label{cor:dominating-three}
Suppose $d_{\max}(\Phi)\leq3$. Then $\Phi$ admits a dominating
transmittable algebra if and only if
$\scA_{\text{cap}}(\Phi)$ is defined and transmittable.
\end{corollary}

\begin{proof}
Necessity is Corollary~\ref{cor:CapacitiesToAlgebra}. Conversely,
every $\scB\in\scT(\Phi)$ has blocks of size at most three and
satisfies
\begin{equation*}
    b_d(\scB)\leq c_d(\Phi)
    =b_d(\scA_{\text{cap}}(\Phi)),
    \qquad d=1,2,3.
\end{equation*}
Lemma~\ref{lem:packing-three} then shows
$\scB\leq\scA_{\text{cap}}(\Phi)$.
\end{proof}

With $c_d=c_d(\Phi)$ and the convention $c_d=0$ above
$d_{\max}(\Phi)$, the candidate in this case is
\begin{equation}
    \scA_{\text{cap}}(\Phi)
    =\M{3}^{\oplus c_3}
      \oplus\M{2}^{\oplus(c_2-c_3)}
      \oplus\C^{\,c_1-2c_2-c_3},
    \label{eq:capacity-algebra-three}
\end{equation}
provided its multiplicities are nonnegative; zero summands are
omitted. When $d_{\max}(\Phi)=2$, nonnegativity is automatic
because $c_1\geq2c_2$, and the complete question is whether
$\M{2}^{\oplus c_2}\oplus\C^{c_1-2c_2}$ can be transmitted
in one coding scheme. When $d_{\max}(\Phi)=1$, the candidate
$\C^{c_1}$ is always transmittable.

\subsection{Algebraic operator systems}
\label{sec:algebraic-operator-systems}

In this section, we prove that if a channel's operator system is closed under multiplication (upto graph isomorphism), then the channel admits a dominating transmittable algebra.

\begin{theorem}\label{thm:alg-max}
    Let $\Phi:\B{\Hil}\to \B{\Kil}$ be a quantum channel such that $S_{\Phi}\subseteq \B{\Hil}$ is a unital $*-$subalgebra. Then, $\Phi$ has a dominating transmittable algebra $\scA_{\dom}(\Phi) \cong S_{\Phi}'$.
\end{theorem}

\begin{proof} Knill-Laflamme conditions (Theorem~\ref{theorem:KL-alg2}) show that $S_{\Phi}'$ is transmittable via $\Phi$.

Conversely, suppose an algebra $\scA$ can be transmitted via $\Phi$ with a faithful encoding representation $\pi_E:\scA\to \B{\Hil}$ and unit projection $P_C:= \pi_E(1_{\scA})$ projecting onto the code space $C:= \im P_C$. The Knill--Laflamme conditions give
\begin{equation}\label{eq:KL-algebra}
    \forall X\in \pi_E(\scA), \, \forall Y\in S_{\Phi}: \quad [P_CY P_C, X] = 0.
\end{equation}
Define the orbit of the code space $C$ under $S_{\Phi}$ as 
\begin{equation*}
    \widehat{C}:= \operatorname{span}\{Y\ket{\psi} : Y\in S_{\Phi}, \ket{\psi}\in C \} \subseteq \Hil.
\end{equation*}
Since $S_{\Phi}$ is closed under products and adjoints, both $\widehat{C}$ and $\widehat{C}^{\perp}$ are invariant under $S_{\Phi}$, implying $P_{\widehat{C}}\in S_{\Phi}'$. Next, for $X\in \pi_E(\scA)$, we define an extension $\widehat{X}\in \B{\widehat{C}}$ as follows:
\begin{equation}\label{eq:hat-def}
\forall Y\in S_{\Phi}, \ket{\psi}\in C : \quad   \widehat{X}Y\ket{\psi} := Y X\ket{\psi}.
\end{equation}

The commutation relation in Eq.~\eqref{eq:KL-algebra} ensures that this is well-defined. Indeed, let $X\in \pi_E(\scA)$ and $\ket{\eta}=\sum_i Y_i \ket{\psi_i} = 0$ for some $Y_i\in S_{\Phi}$ and $\ket{\psi_i}\in C$. We must show that $\sum_i Y_i X\ket{\psi_i}=0$ in $\widehat{C}$. Let $Y\in S_{\Phi}, \ket{\phi}\in C$. Then,
\begin{align*}
    \bra{\phi}Y^{*} \sum_i Y_iX \ket{\psi_i} &= \bra{\phi} \sum_i P_CY^{*} Y_i P_CX\ket{\psi_i} \\
    &= \bra{\phi} \sum_i X P_CY^{*} Y_i P_C\ket{\psi_i} \\
    &= \bra{\phi} XP_CY^{*} \sum_i Y_i \ket{\psi_i} =0,
\end{align*}
where we used the commutation relation in Eq.~\eqref{eq:KL-algebra} and the fact that $Y^{*}Y_i\in S_{\Phi}$ (since $S_{\Phi}$ is closed under products). 
Since the above holds for all $Y\in S_{\Phi}, \ket{\phi}\in C$, we must have
\begin{equation*}
  \hat{X}\ket{\eta} = \sum_i Y_i X\ket{\psi_i}=0.
\end{equation*}

For $Y,Z\in S_{\Phi}$ and $\psi\in C$, the definition gives
\begin{equation*}
    \widehat{X}Z(Y\psi)=ZYX\psi=Z\widehat{X}(Y\psi).
\end{equation*}
Thus $\widehat{X}$ commutes with the restriction of $S_{\Phi}$ to $\widehat{C}$.
Moreover, we claim that $X\mapsto \widehat{X}$ is an injective $*$-homomorphism. Linearity and multiplicativity follow from the definition~\eqref{eq:hat-def}. Since $1_{\Hil}\in S_{\Phi}$, we have $\widehat{X}|_C=X|_C$, which gives injectivity because $X=P_CXP_C$. To show $\widehat{X^{*}}=\widehat{X}^{*}$, consider $Y, Z\in S_{\Phi}$ and $\psi,\phi \in C$ and note\footnote{We avoid the bra-ket notation here for clarity of exposition.}
\begin{align*}
   \langle \widehat{X}Y\psi, Z\phi \rangle = \langle YX\psi, Z\phi \rangle &= \langle X\psi, Y^{*}Z\phi \rangle \\
   &= \langle X\psi, P_CY^{*}ZP_C \phi \rangle \\
   &= \langle \psi , P_CY^{*}ZP_C X^{*} \phi \rangle \\
   &= \langle Y\psi , Z X^{*}\phi \rangle \\
   &= \langle Y\psi , \widehat{X^{*}}Z \phi \rangle.
\end{align*}
    
Hence, $X\mapsto\widehat{X}$ embeds $\pi_E(\scA)$ into
$(P_{\widehat{C}}S_{\Phi}P_{\widehat{C}})'$, where this commutant is
taken inside $\B{\widehat{C}}$. Since $\widehat{C}$ reduces $S_{\Phi}$,
extension by zero on $\widehat{C}^{\perp}$ gives the required embedding 
\begin{align*}
    \pi : \scA &\to S_{\Phi}' \\
    a&\longmapsto\widehat{\pi_E(a)}\oplus0_{\widehat{C}^{\perp}},
\end{align*}
thus proving $\scA \leq S_{\Phi}'$.
\end{proof}

For an alternative proof of Theorem~\ref{thm:alg-max}, see Appendix~\ref{appen:alg-dom}. Using Lemma~\ref{lemma:alg-DPI}, we can easily extend this result as follows.

\begin{corollary}\label{corollary:alg-max}
    Let $\Phi:\B{\Hil}\to \B{\Kil}$ be a quantum channel such that $S_{\Phi}\subseteq \B{\Hil}$ is graph-isomorphic to a unital $*$-subalgebra $S\subseteq \B{\Hil'}$ on a finite-dimensional Hilbert space $\Hil'$: $S_{\Phi}\longrightarrow S$ and $S\longrightarrow S_{\Phi}$. Then, $\Phi$ has a dominating algebra $\scA_{\dom}(\Phi) \cong S'$.
\end{corollary}
 
\subsubsection{Highly divisible channels}

In this subsection, we show that for any channel $\Psi$, after sufficiently many self-iterations $\Psi \circ \Psi \circ \ldots \circ \Psi$, the resulting channel always admits a dominating transmittable algebra.

\begin{definition}
    A quantum channel $\Phi:\B{\Hil}\to \B{\Hil}$ is said to be $l$-Markovian divisible if there exists another quantum channel $\Psi:\B{\Hil}\to \B{\Hil}$ such that 
    \begin{equation}
        \Phi = \underbrace{\Psi\circ \Psi \circ \ldots \circ\Psi}_{l \operatorname{times}} =: \Psi^l.
    \end{equation}
\end{definition}

The following lemma notes an important stabilization property of operator systems under identical sequential compositions of quantum channels $\Psi\circ \Psi \circ \ldots \Psi$. 

\begin{lemma}\label{lemma:op-chain} \cite{Singh2025zero-markovian} 
    Let $\Psi: \B{\Hil}\to \B{\Hil}$ be a quantum channel. Then, there exists $L\leq (\dim \Hil)^2 - \dim S_{\Psi}$ such that the following chain of (strict) inclusions/equalities hold:
    \begin{equation*}
        S_{\Psi} \subset S_{\Psi^2} \subset \ldots \subset S_{\Psi^L} = S_{\Psi^{L+1}} = \ldots
    \end{equation*}
\end{lemma}

Using the above result, the notion of ``stabilized operator system'' of a quantum channel was introduced and studied in \cite{Singh2026markovian}, \cite{Singh2025thesis}.

\begin{definition} \cite{Singh2026markovian} 
\label{def:stab-opsys}
    Let $\Psi:\B{\Hil}\to \B{\Hil}$ be a quantum channel with $d=\dim \Hil$. We define the \emph{stabilized operator system} of $\Psi$ as follows:
    \begin{equation*}
        S_{\Psi^{\infty}} := \bigcup_{l\in \mathbb{N}} S_{\Psi^l} = S_{\Psi^{d^2}} = S_{\Psi^{d^2+1}} =\ldots ,
    \end{equation*}
    where the latter equalities follow from Lemma~\ref{lemma:op-chain}.
\end{definition}

\begin{theorem}\label{theorem:dom-highly-divisible}
    Let $\Phi:\B{\Hil}\to \B{\Hil}$ be an $l$-Markovian divisible channel with $l\geq(\dim\Hil)^2$. Then, $\Phi$ admits a dominating transmittable algebra that is $*$-isomorphic to its peripheral algebra $\mathscr{X}^*(\Phi)$ (see Definition~\ref{def:peripheral-alg}):
    \begin{equation}
         \scA_{\dom}(\Phi) \cong \mathscr{X}^*(\Phi).
    \end{equation}
\end{theorem}

\begin{proof}
    The proof essentially follows from \cite[Theorem 9]{Singh2026markovian}.
    Since $\Phi$ is $l-$divisible, there exists a channel $\Psi$ such that $\Phi=\Psi^l$. Thus, using Definition~\ref{def:stab-opsys} and \cite[Theorem 9]{Singh2026markovian}, 
    \begin{equation*}
        S_{\Phi} = S_{\Psi^l} = S_{\Psi^{\infty}} = S_{\Phi^{\infty}} = S_{\cP_{\Phi}}.
    \end{equation*}
    Hence, it suffices to look at the peripheral projection channel $\cP_{\Phi}$. The operator system relations (Lemma~\ref{lemma:op-homo}) along with Eq.~\eqref{eq:PPbar-relations} imply that 
    \begin{equation*}
    S_{\mathcal{P}_{\Phi}} \subseteq S_{R_V \circ \mathcal{P}_{\Phi}} \subseteq S_{\mathcal{V}\circ R_V \circ \mathcal{P}_{\Phi}} = S_{\mathcal{P}_{\Phi}}.
\end{equation*}
Hence, using Lemma~\ref{lemma:op-homo} again with Eq.~\eqref{eq:PPbar-relations2}, we obtain the homomorphism
\begin{equation*}
     S_{R_V \circ \mathcal{P}_{\Phi}} = S_{\mathcal{P}_{\Phi}} \longrightarrow S_{\overbar{\mathcal{P}}_{\Phi}}.
\end{equation*}
Moreover, using Lemma~\ref{lemma:op-homo} again with Eq.~\eqref{eq:PPbar-relations3}, we obtain the reverse homomorphism 
\begin{equation*}
S_{\overbar{\mathcal{P}}_{\Phi}} = S_{R_V\circ \mathcal{P}_{\Phi}\circ \mathcal{V}} \longrightarrow S_{R_V\circ \mathcal{P}_{\Phi}} = S_{\mathcal{P}_{\Phi}}.    
\end{equation*}
The adjoint $\overbar{\cP}_\Phi^*$ is a conditional expectation
onto $\mathscr X^*(\Phi)$ with $S_{\overbar{\cP}_\Phi}=\bigl(\mathscr X^*(\Phi)\bigr)'$, see \eqref{eq:SPbar=X*'},
where the commutant is taken in $\B{\Hil_0^\perp}$. This is a
unital $*$-algebra. Corollary~\ref{corollary:alg-max} and the
finite-dimensional bicommutant identity now give the assertion.
\end{proof}

\subsection{Channels with zero one-shot zero-error quantum capacity}\label{sec:zero-quantum-capacity}

\begin{theorem}\label{theorem:dom-classical}
    Let $\Phi:\B{\Hil}\to \B{\Kil}$ be a channel such that its one-shot zero-error quantum capacity $Q_0(\Phi)=0$. Then, $\Phi$ has a dominating transmittable algebra 
    \begin{equation}
\scA_{\operatorname{dom}}(\Phi) \cong \C^{2^{C_0(\Phi)}} := \underbrace{\C \oplus \C \oplus \ldots \oplus \C}_{2^{C_0(\Phi)} \text{ times}}.
    \end{equation}
\end{theorem}
\begin{proof}
    Since $Q_0(\Phi)=0$, any transmittable algebra for $\Phi$ is abelian. Since the set of abelian algebras is totally ordered:
    \begin{equation*}
        \C^{m}\leq \C^n \iff m\leq n,
    \end{equation*}
    it is clear that $\Phi$ has a dominating transmittable algebra, with shape given by the maximum number of classical messages that can be sent through $\Phi$.
\end{proof}

All PPT channels, entanglement-breaking channels, and classical channels of the form in Eq.~\eqref{eq:classicalchannel} are examples of this kind. In this class, we can construct various examples $\Phi$ for which the dominating transmittable algebra is strictly larger than the peripheral algebra: $\mathscr{X}^*(\Phi) < \scA_{\operatorname{dom}}(\Phi)$. We note an interesting classical example below.

\begin{example}
\emph{
We borrow an example from \cite{Akelbek2009ergodic, Singh2025zero-markovian}. Consider the $d\times d$ stochastic matrix  
\begin{equation}\label{eq:Ad}
    N := \left( \begin{array}{cccccc}
        0 & 1/2 & 0 & 0 & 0 & 0 \\
        0 & 0 & 1 & 0 & 0 & 0 \\
        0 & 0 & 0 & 1 & 0 & 0 \\
        \vdots & \vdots & \vdots & \vdots & \ddots & \vdots \\
        0 & 0 & 0 & 0 & 0 & 1 \\
        1 & 1/2 & 0 & 0 & 0 & 0 \\
    \end{array} \right),
\end{equation}
and define the corresponding quantum channel $\Phi_N:\M{d}\to \M{d}$ as (see Remark~\eqref{remark:c-q-embed}):
\begin{equation}\label{eq:classicalchannel}
    \forall X\in \M{d}: \quad \Phi_N(X) = \sum_{i,j} N_{ij}X_{jj} \ket{i}\bra{i}.
\end{equation}
One can check that $\Phi_N$ is mixing \cite{Akelbek2009ergodic}, i.e. there exists a unique invariant state 
\begin{equation}
    \rho = \operatorname{diag}(\pi), \quad \pi = \left(\frac{1}{2d-1}, \frac{2}{2d-1}, \frac{2}{2d-1},\ldots ,\frac{2}{2d-1}\right)^\top
\end{equation}
such that $\lim_{n\to \infty}\Phi^n_N(\cdot) =\Tr (\cdot)\rho$. Hence, the peripheral algebra is trivial $\mathscr{X}^*(\Phi)\cong\C$. However, one can send $d-1$ classical messages via $\Phi_N$ by using the encoding states 
\begin{equation}
    \ket{0}\bra{0},     \ket{2}\bra{2}, \ldots ,     \ket{d-1}\bra{d-1},
\end{equation}
since columns $1,3,4,\ldots d$ of $N$ have pairwise disjoint supports. Moreover, since columns $1,2$ have overlapping supports, one cannot send more than $d-1$ messages. Hence, $C_0(\Phi_N)=\log (d-1)$, and we get a dominating transmittable algebra $\scA_{\operatorname{dom}}(\Phi_N)\cong\C^{d-1}$.
}
\end{example}

\begin{remark}
The families of channels with no zero-error quantum capacity and channels whose operator systems are (graph) isomorphic to $*-$algebras are not comparable with each other. To see this, consider a classical stochastic matrix $N\in \Mr{5}$ defined as

\begin{equation}
    N(j|i) := \begin{cases}
        1/2 \quad \text{if } j=i \,\,\text{or} \,\, j=i+1 \, ( \operatorname{mod} 5) \\
        0 \quad \quad\text{otherwise},
    \end{cases}
\end{equation}
and let $\Phi_N:\B{\C^5}\to \B{\C^5}$ be the associated quantum channel (see Eq.~\eqref{eq:PhiN}). 
Then, we claim that $S_{\Phi_N}$ cannot be (graph) isomorphic to a $*-$algebra. To see this, note that the confusability graph of $N$ is the $5-$cycle $G_{N} = C_5$. One can easily calculate the independence numbers $\alpha(S_{\Phi_N})=2$ and $\alpha(S_{\Phi_N}\otimes S_{\Phi_N})=5$, which are not multiplicative. However, if $S_{\Phi_N}$ was graph isomorphic to a $*-$algebra, its independence numbers must be multiplicative (see Lemmas~\ref{lemma:op-bottleneck}, \ref{lemma:alpha-algebra}), and the claim follows. 

On the other hand, the identity channel $\id : \B{\C^d}\to \B{\C^d}$ has $S_{\id}= \C 1_d$ which is the trivial unital $*$-algebra, but clearly $Q_0(\id)=\log d >0$.
\end{remark}

The two classes of Channels from Sections~\ref{sec:algebraic-operator-systems} and \ref{sec:zero-quantum-capacity} do not exhaust the channels covered by
Theorem~\ref{thm:dominating-forbidden-types}, as illustrated by the example given in Appendix~\ref{appen:SnotAlgebra}. See also Figure~\ref{fig:channel-classes} for an illustration of all classes of channels discussed in this paper.

\section{Tensor products and many-use transmission}
\label{sec:tensor-products}

The theory developed so far describes the hybrid algebra types that can
be transmitted through a \emph{single} use of a noisy channel. We now
consider several independent channel uses and ask whether product coding
captures every transmittable algebra type, or whether joint encoding can
produce new ones. We show that joint encoding can indeed yield new types
and, more strongly, that a tensor product of channels admitting dominating
algebras need not itself admit a dominating transmittable algebra.
For repeated uses of a fixed channel $\Phi$, we then study the number
of maximal transmittable algebra types of $\Phi^{\otimes n}$ and how
rapidly this number can grow with the blocklength $n$.

\subsection{Product coding and the loss of domination}
\label{sec:tensor-domination}

For two algebras $\scA\cong\bigoplus_i\M{d_i}$ and
$\scB\cong\bigoplus_j\M{e_j}$, we have
\begin{equation}
    \scA\otimes\scB\cong\bigoplus_{i,j}\M{d_i e_j}.
\end{equation}
If $\Phi$ and $\Psi$ are two channels, and $\scA\in\scT(\Phi)$ and $\scB\in\scT(\Psi)$, then taking tensor product of the corresponding encoders and decoders (Lemma~\ref{lemma:EncDecTransmittableAlgebra}) shows that $\scA\otimes\scB\in\scT(\Phi\otimes\Psi)$. To include all algebras
that these product codes can support, we define, for lower sets $I,J$ of algebra types (see Remark~\ref{remark:lower-set} and Appendix~\ref{appen:order-theory}),
\begin{equation}
    I\star J:=\left\{\scC:
    \exists\,\scA\in I,\ \scB\in J
    \text{ with }\scC\leq\scA\otimes\scB\right\}.
    \label{eq:star-definition}
\end{equation}
This is again a lower set, and product coding gives
\begin{equation}
    \scT(\Phi)\star\scT(\Psi)
    \subseteq\scT(\Phi\otimes\Psi).
    \label{eq:product-coding-lower-set}
\end{equation}
Strict inclusion means that joint coding transmits an algebra which does
not embed into any product of separately transmittable algebras. This
refines the usual scalar question of superadditivity of classical and quantum capacities \cite{Shannon1956zero,Hastings2009super,Yard2008super,
Chen2010zerosuper,Duan2009zerosuper}.

If $\Phi$ and $\Psi$ have dominating algebras $\scA$ and $\scB$,
respectively, then
\begin{equation}
    \scT(\Phi)\star\scT(\Psi)
    =\{\scC:\scC\leq\scA\otimes\scB\}
    =:\downarrow(\scA\otimes\scB).
\end{equation}
Here $\downarrow\scD$ denotes the \emph{principal lower set}
generated by an algebra $\scD$, consisting of all algebra types
that embed into $\scD$
(see Appendix~\ref{appen:order-theory}). Strict inclusion in \eqref{eq:product-coding-lower-set} alone
does not rule out a larger dominating algebra for the product
channel. The following lemma lets us prove the stronger statement
that domination itself can fail.

\begin{lemma}
\label{lemma:Nxid}
Let $N\in\Mr{d_B\times d_A}$ be a stochastic matrix, let $\Phi_N$ be
its associated quantum channel (Remark~\ref{remark:c-q-embed}), and write
$\alpha=\alpha(G_N)$. For every $d\geq1$, the algebra
$\M{d}^{\oplus\alpha}$ is maximal transmittable through
$\Phi_N\otimes\id_d$.
\end{lemma}

\begin{proof}
The algebra $\M{d}^{\oplus \alpha}\cong \C^{\alpha} \otimes \M{d}$ is transmittable by product coding via $\Phi_N \otimes \id_d$. Now, the noncommutative graph of the product channel is
\begin{equation*}
    S_{\Phi_N\otimes\id_d}
    =\spa\{\ket{i}\!\bra{j}\otimes1_d:
      i=j\text{ or }i\sim_N j\}.
\end{equation*}
Suppose $\scB=\bigoplus_k\M{d_k}\in \scT (\Phi_N \otimes \id_d)$ and
$\M{d}^{\oplus\alpha}\leq\scB$.
By Lemma~\ref{lemma:pairwise-S-orthogonal}, choose isometries
$V_k:\C^{d_k}\to\C^{d_A}\otimes\C^d$ defining mutually
$S_{\Phi_N\otimes\id_d}$-orthogonal quantum codes. Set
$V_{i,k}=(\bra i\otimes1_d)V_k$. The Knill--Laflamme conditions give,
whenever $i=j$ or $i\sim_N j$,
\begin{equation}
    V_{i,k}^*V_{j,l}=0\quad(k\neq l),
    \qquad
    V_{i,k}^*V_{j,k}\in\C1_{d_k}.
    \label{eq:KL-Nxid}
\end{equation}
For each $k$, at least one $V_{i,k}$ is nonzero, since
$\sum_iV_{i,k}^*V_{i,k}=1_{d_k}$. Such a nonzero $V_{i,k}$ forces each $k$-block to be of size $      d_k = \operatorname{rank} (V_{i,k}^*V_{i,k})= \operatorname{rank} (V_{i,k}) \leq d.$ The embedding assumption $\M{d}^{\oplus\alpha}\leq\scB$ further forces at least
$\alpha$ distinct blocks of size $d$ in $\scB$.

For each size-$d$ block $k$, choose $i_k$ with $V_{i_k,k}\neq0$.
This is an invertible $d\times d$ matrix. If two chosen vertices $i_k, i_l$ were
equal or adjacent, \eqref{eq:KL-Nxid} would give
$V_{i_k,k}^*V_{i_l,l}=0$, a contradiction. Thus, the chosen vertices
form an independent set in $G_N$, which implies that there are exactly $\alpha$ size-$d$
blocks in $\scB$. If an additional block existed in $\scB$, choose a vertex $v$ where its
vertex map is nonzero. Again by \eqref{eq:KL-Nxid}, $v$ is distinct
from and nonadjacent to every $i_k$, contradicting the definition of
$\alpha$. Hence, $\scB\cong\M{d}^{\oplus\alpha}$.
\end{proof}

Using Lemma~\ref{lemma:Nxid} and the recent result of \cite{ambuj2026super}, we obtain the main result of this section.

\begin{theorem}[Domination is not stable under tensor products]
\label{thm:tensor-destroys-domination}
There exist channels $\Phi$ and $\Psi$, each with a dominating
transmittable algebra, for which $\Phi\otimes\Psi$ has no dominating
transmittable algebra.
\end{theorem}

\begin{proof}
Take the classical channel $N$ with $18$ inputs and $9$ outputs used in
\cite[Theorem~1]{ambuj2026super}, based on the $18$-vector construction
of \cite{Cabello1996bell}. Its confusability graph has independence
number $4$. Hence $\Phi_N$ has dominating algebra $\C^4$, while
$\id_4$ has dominating algebra $\M{4}$.

The joint protocol of \cite{ambuj2026super} transmits $18$ classical
messages through $\Phi_N\otimes\id_4$. Briefly, each message specifies
a classical input and a vector in $\C^4$; the observed classical output
restricts the possible vectors to an orthonormal basis, allowing perfect
decoding of the quantum output. Thus $\C^{18}$ is transmittable.
On the other hand, Lemma~\ref{lemma:Nxid} shows that
$\M{4}^{\oplus4}$ is maximal transmittable. Since
$\C^{18}\not\leq\M{4}^{\oplus4}$, this maximal algebra is not
dominating. Finiteness of the transmission set gives another maximal
algebra above $\C^{18}$, necessarily incomparable with
$\M{4}^{\oplus4}$.
\end{proof}

\subsection{The number of maximal algebras at block-length $n$}
\label{sec:tensor-width}

To quantify the possible failure of domination at larger block-lengths, we introduce and study the notion of \emph{width} for a given channel in this section, which counts the maximal transmittable algberas at each block-length. More precisely, for a fixed channel $\Phi$, set
\begin{equation}
    \scT_n(\Phi):=\scT(\Phi^{\otimes n}),
    \qquad
    w_n(\Phi):=\abs{\max(\scT_n(\Phi))}.
    \label{eq:width-definition}
\end{equation}
In particular, $w_n(\Phi)=1$ precisely when $\Phi^{\otimes n}$ admits a
dominating algebra.

For a split into $m$ and $n$ uses, product coding \eqref{eq:product-coding-lower-set} yields
\begin{equation}
    \scT_m(\Phi)\star\scT_n(\Phi)
    \subseteq\scT_{m+n}(\Phi).
    \label{eq:tensor-power-product-coding}
\end{equation}
This inclusion alone does not imply
$w_{m+n}(\Phi)\geq w_m(\Phi)w_n(\Phi)$, since tensor products of different maximal types can become isomorphic or comparable, and
product codes need not remain maximal when joint coding is allowed. Nevertheless, if all maximal types can be obtained from a fixed finite
collection of codes by taking tensor products and passing to embedded
algebras, then their number grows at most polynomially. We prove this in the following lemma.

\begin{lemma}[Finite generation implies polynomial growth]
\label{lemma:finite-generation-polynomial}
Let $\scG_1,\ldots,\scG_R$ be maximal transmittable algebras for a channel $\Phi$ at
positive block-lengths $\ell_1,\ldots,\ell_R$, respectively.
Suppose every maximal algebra $\scC$ at every block-length $n\geq1$
satisfies
\begin{equation}
    \scC\leq\bigotimes_{j=1}^R\scG_j^{\otimes a_j},
    \qquad a_j\in\N_0,
    \qquad\sum_{j=1}^R a_j\ell_j=n.
\end{equation}
Then $w_n(\Phi)\leq\binom{n+R}{R}=O(n^R)$.
\end{lemma}

\begin{proof}
Each product on the right is itself transmittable at block-length $n$,
so maximality forces $\scC$ to be isomorphic to that product.
The product type is determined by its exponent vector. Since each
$\ell_j\geq1$, we have $\sum_j a_j\leq n$, so the number of such
vectors is at most
\begin{equation*}
    \sum_{s=0}^n\binom{s+R-1}{R-1}=\binom{n+R}{R}.
\end{equation*}
\end{proof}

We call a maximal algebra at block-length $n\geq2$
\emph{product-indecomposable at that block-length} if it cannot be
transmitted by product coding across any split into two smaller
block-lengths, i.e.\ it belongs to no
$\scT_m(\Phi)\star\scT_{n-m}(\Phi)$ with $1\leq m<n$.
If such types occurred only at bounded block-lengths, repeated
factorization would express every maximal algebra using the finitely
many maximal types at those smaller block-lengths. Lemma~\ref{lemma:finite-generation-polynomial} therefore implies that superpolynomial growth of $w_n$ forces
product-indecomposable types at arbitrarily large block-lengths.

The input dimension also gives a universal upper bound on $w_n$.

\begin{proposition}[Universal upper bound]
\label{prop:width-upper-bound}
Let $\Phi:\B{\Hil}\to\B{\Kil}$ be a channel with
$d=\dim\Hil$, and let $p(r)$ denote the number of integer
partitions of $r$. Then, for every $n\geq1$,
\begin{equation}
    w_n(\Phi)
    \leq\sum_{r=1}^{d^n}p(r)
    \leq2^{d^n}-1.
    \label{eq:width-upper-bound}
\end{equation}
The Hardy-Ramanujan formula \cite{HardyRamanujan1918} gives the
sharper asymptotic bound
\begin{equation}
    w_n(\Phi)
    \leq
    \frac{1+o(1)}{4\sqrt{3}}
    \exp\!\left(\pi\sqrt{\frac{2d^n}{3}}\right),
    \qquad n\to\infty.
    \label{eq:width-upper-bound-HR}
\end{equation}
In particular, the growth
is at most double exponential in the block-length $n$.
\end{proposition}

\begin{proof}
If an algebra
$\scA\cong\bigoplus_{k=1}^K\M{d_k}$
is transmittable via $\Phi^{\otimes n}$, then $\sum_{k=1}^K d_k\leq d^n$ according to Lemma~\ref{lemma:finite-algebra-types}. The isomorphism type of $\scA$ is uniquely determined by its block dimensions, arranged in nonincreasing order. Thus, for each fixed value
$r=\sum_kd_k$, the possible algebra types are in bijection
with the integer partitions of $r$. It follows that
\begin{equation*}
    w_n(\Phi)
    \leq \abs{\scT(\Phi^{\otimes n})}
    \leq \sum_{r=1}^{d^n}p(r),
\end{equation*}
where the transmission set is counted up to
$*$-isomorphism. 

To obtain the elementary estimate, recall that an
ordered composition of $r$ is an ordered sequence of
positive integers summing to $r$. Such a composition
is specified by placing separators in a row of $r$ ones:
at each of the $r-1$ gaps, we either insert a separator
or leave the adjacent ones in the same group.
There are therefore exactly $2^{r-1}$ ordered
compositions. Listing the parts of each partition
in nonincreasing order gives an injection into these
compositions, so $p(r)\leq2^{r-1}$. Consequently,
\begin{equation*}
    w_n(\Phi)
    \leq\sum_{r=1}^{d^n}2^{r-1}
    =2^{d^n}-1.
\end{equation*}

For the sharper asymptotic estimate, the
Hardy-Ramanujan formula \cite{HardyRamanujan1918} gives
\begin{equation*}
    p(r)
    =\frac{1+o(1)}{4\sqrt{3}\,r}
      \exp\!\left(\pi\sqrt{\frac{2r}{3}}\right),
    \qquad r\to\infty.
\end{equation*}
The partition function is nondecreasing, since adjoining a part of size one injects the partitions of $r$ into
those of $r+1$. Hence, for a fixed $d$,
\begin{equation*}
    \sum_{r=1}^{d^n}p(r)
    \leq d^np(d^n)
    =\frac{1+o(1)}{4\sqrt{3}}
      \exp\!\left(\pi\sqrt{\frac{2d^n}{3}}\right), \qquad n\to \infty
\end{equation*}
\end{proof}

The construction in the following section give a double-exponential
lower bound on $w_n$ for a specifically constructed channel.

\subsection{Double-exponential growth of width}

In this subsection, we construct a channel for which the number of maximal transmittable algebras grows doubly exponentially in block-length. The construction has two main ingredients. We first show that the tensor powers of a fixed channel $\Phi$ admit $2^n$ pairwise non-isomorphic maximal transmittable algebras. We then tensor $\Phi$ with a completely dephasing qubit channel, which provides $2^n$ perfectly distinguishable classical sectors. Choosing one of the previously constructed algebras in each of these sectors gives the desired double-exponential growth.

Consider the channel $\Phi:\B{\C^3}\to\B{\C^4}$ from Example~\ref{ex:NoDomSecondExample}, defined via the Kraus operators
\begin{equation}
E_0
=
\ket{f_0}\!\bra{e_0}
+
\frac{1}{\sqrt{2}}\ket{f_1}\!\bra{e_1},
\qquad
E_1
=
\frac{1}{\sqrt{2}}\ket{f_2}\!\bra{e_1}
+
\ket{f_3}\!\bra{e_2},
\label{eq:kraus}
\end{equation}
where $\{\ket{e_0},\ket{e_1},\ket{e_2}\}$ and
$\{\ket{f_0},\ket{f_1},\ket{f_2},\ket{f_3}\}$ are the standard bases of $\C^3$
and $\C^4$, respectively. From this we see
\begin{equation}
E_0^*E_0
=
\operatorname{diag}\left(1,\frac{1}{2},0\right),
\qquad
E_1^*E_1
=
\operatorname{diag}\left(0,\frac{1}{2},1\right),
\qquad
E_0^*E_1=0.
\label{eq:effects}
\end{equation}
Let $\Delta_2$ denote the completely dephasing qubit channel, i.e. $\Delta_2(\,\cdot\,) = \kb{0}(\,\cdot\,)\kb{0} + \kb{1}(\,\cdot\,)\kb{1},$ and define
\begin{equation}
\Psi\eqdef\Delta_2\otimes\Phi.
\label{eq:Psi}
\end{equation}
For proving converses on the transmittability of certain algebras via $\Phi^{\otimes n}$ and $\Psi^{\otimes n}$ and by that maximality within $\mathcal{T}(\Phi^{\otimes n})$ and $\mathcal{T}(\Psi^{\otimes n})$ , we define for $1\le q \le 2^n$ the number
\begin{align}
\label{eq:Defsn(q)}
    s_n(q) = \sum_{i=1}^q 2^{|x^{(i)}|}.
 \end{align}
Here, $x^{(1)}, x^{(2)},x^{(3)},\cdots ,x^{(2^n)}\in\{0,1\}^n$ denotes an enumeration of all bit strings such that their Hamming weights are ordered as
\begin{align}
    \label{eq:Orderingfors_n}
    0=|x^{(1)}|\le |x^{(2)}|\le |x^{(3)}|\le \cdots\le |x^{(2^n)}|=n.
\end{align}
Note that as the definition of $s_n(q)$ only depends on the Hamming weights of the first $q$ bit strings, the value of $s_n(q)$ is independent of the particular choice of the enumeration satisfying \eqref{eq:Orderingfors_n}. 
Furthermore, we immediately see from definition that $q\mapsto s_n(q)$ is strictly monotonicially increasing with endpoints $s_n(1) = 1$ and $s_n(2^n) = 3^n.$
The $s_n(q)$ numbers are useful for proving converse bounds due to the following monotonicity relation\footnote{To see this note that for $a,b\geq 1$ integers such $a+b\leq 2^n$, $
s_n(a+b)-s_n(a)
=
\sum_{i=1}^b 2^{|x^{(a+i)}|}\\
>
\sum_{i=1}^b 2^{|x^{(i)}|}
=
s_n(b).$
Here the strict inequality follows because $|x^{(a+1)}|>|x^{(1)}| =0$. Now the embedding order $\bigoplus_j\M{a_j}\le \bigoplus_j\M{b_j}$ implies that there exists a function from the $i$ labels to the $j$ labels such that for all $j$ we have $
\sum_{i : f(i) = j}a_i\leq b_j$ (see \eqref{eq:algebra-packing-assignment}). This gives
$s_n(b_j)
\geq
s_n\left(\sum_{i\, \text{s.t.}\, f(i) = j}a_i\right)
\geq
\sum_{i\, \text{s.t.}\, f(i) = j}s_n(a_i).$
Summing over all $j$ shows \eqref{eq:cost-monotone}. Morover, from this we can also easily see that if additionally $\sum_i s_n(a_i)
=
\sum_j s_n(b_j)$ implies that the $f$ above is a bijection and hence the two algebras are isomorphic, which proves the second line in \eqref{eq:cost-monotone}.}
 \begin{align}
\nn \bigoplus_i\M{a_i}
&\leq
\bigoplus_j\M{b_j} \quad\implies\quad
\sum_i s_n(a_i)
\leq
\sum_j s_n(b_j), \\
\bigoplus_i\M{a_i}
&<
\bigoplus_j\M{b_j} \quad \implies \quad \sum_i s_n(a_i)
<
\sum_j s_n(b_j),
\label{eq:cost-monotone}
\end{align}
for all positive integers $
a_i,b_j\leq 2^n.$

With that we can state the main result of this section.

\begin{proposition}\label{prop:DoubleExp}
For $n\in\N$ and $1\leq q\leq 2^n$, the algebra
\begin{equation}
\scB_{n,q}
\eqdef
\M{q}\oplus\C^{\,3^n-s_n(q)}
\label{eq:Bnq}
\end{equation}
is transmittable via $\Phi^{\otimes n}$ and a maximal element of $\Tcal(\Phi^{\otimes n})$. In particular this gives
\begin{equation}
w_n(\Phi)\geq 2^n.
\label{eq:Phi-lb}
\end{equation}
Moreover for $N=(N_1,\ldots,N_{2^n})\in\N_0^{2^n}$ satisfying
\begin{equation}
\sum_{q=1}^{2^n} N_q=2^n,
\label{eq:Nconstraint}
\end{equation}
the algebra
\begin{equation}
\scC_N
:=
\bigoplus_{q=1}^{2^n}\scB_{n,q}^{\oplus N_q}
\label{eq:CN}
\end{equation}
is transmittale via $\Psi^{\otimes n}$ and a maximal element of $\Tcal(\Psi^{\otimes n})$. In particular from this we see 
\begin{equation}
w_n(\Psi)
\ge 2^{2^n -1}
\label{eq:doubleexp}
\end{equation}

\end{proposition}
Before proving of Proposition~\ref{prop:DoubleExp} we start with summarising its main ideas: The transmittability of the algebras $\scB_{n,q}$ via $\Phi^{\otimes n}$ follows by constructing a $q$-dimensional quantum code whose support contains exactly $s_n(q)$ computational basis vectors. Since the effects of $\Phi^{\otimes n}$ are diagonal in the computational basis, see \eqref{eq:effects}, we can transmit additional $3^n-s_n(q)$ many classical blocks using those vectors. 
Transmittability of $\scC_N$ via $\Psi^{\otimes n}\equiv \Delta^{\otimes n}_2\otimes \Phi^{\otimes n}$ then follows by transmitting through $\Delta_2^{\otimes n}$ a classical label specifying which of the $2^n$ direct summands $\scB_{n,q}$ in \eqref{eq:CN} is being transmitted, while transmitting the corresponding $\scB_{n,q}$ itself through $\Phi^{\otimes n}$ as described above.

The actual work to prove Proposition~\ref{prop:DoubleExp} is then to show that these algebras are maximal in the set of transmittable algebras of $\Phi^{\otimes n}$ and $\Psi^{\otimes n}$ respectively. For that we provide Lemma~\ref{lem:shadow}, building on Lemmas~\ref{lem:face-shadow} and~\ref{prop:oneblock}, a converse bound on transmittable algebras of these channels in terms of the $s_n(q)$ numbers defined in \eqref{eq:Defsn(q)}. We further show that these converse bounds in 
Lemma~\ref{lem:shadow} are saturated by $\scB_{n,q}$ and $\scC_N$ respectively which by \eqref{eq:cost-monotone} implies the desired maximality.

In the following we start with some useful definitions and observations, then continue to state and proof the mentioned lemmas and provide the proof of Proposition~\ref{prop:DoubleExp}.

For $y=(y_1,\ldots,y_n)\in\{0,1\}^n$, define
\begin{equation}
E_y
\eqdef
E_{y_1}\otimes\cdots\otimes E_{y_n}\quad\text{and}\quad 
M_y
\eqdef
E_y^*E_y.
\end{equation}
The Kraus operators have mutually orthogonal output ranges:
$E_y^*E_z=0$ for $y\neq z$. Hence,
\begin{align}
\label{eq:OperatorSystemPhin}
    S_{\Phi^{\otimes n}} = \operatorname{span}\Big\{ M_y \,:\, y\in\{0,1\}^n\Big\}.
\end{align}

The converse bound of Proposition~\ref{prop:DoubleExp} is based on a description of the computational basis of
$(\C^3)^{\otimes n}$ in terms of faces of the Boolean hypercube. A
\emph{face} of $\{0,1\}^n$ is a set of the form
\begin{equation*}
F=F_1\times\cdots\times F_n\subseteq \{0,1\}^n,
\qquad
F_i\in\bigl\{\{0\},\{1\},\{0,1\}\bigr\}.
\end{equation*}
Note that there are $3^n$ distinct face sets. Those form a particularly well-suited parametrisation of the computational basis of $\C^{3^n}$ as for 
\begin{equation}
\label{eq:FaceVectors}
\ket{F}
\eqdef
\ket{e_{\tau(F_1)}}\otimes\cdots\otimes\ket{e_{\tau(F_n)}},
\qquad
\tau(\{0\})=0,
\quad
\tau(\{0,1\})=1,
\quad
\tau(\{1\})=2,
\end{equation}
we see directly from \eqref{eq:effects} that
\begin{equation}
M_y\ket{F}
=
\begin{cases}
\abs{F}^{-1}\ket{F},&y\in F,\\
0,&y\notin F.
\end{cases}
\label{eq:face-action}
\end{equation}

The following lemma provides an useful combinatorial fact relating these face sets with the $s_n(q)$ numbers. Its
proof is deferred to Appendix~\ref{app:face-shadow}.

\begin{lemma}\label{lem:face-shadow}
Let $\emptyset\neq Y\subseteq\{0,1\}^n$ and $q\in\N$. Suppose that every
$y\in Y$  satisfies $y\in F\subseteq Y$ for at least $q$ faces $F\subseteq \{0,1\}^n$. Then
$q\leq\abs{Y}\leq 2^n$. Moreover, for $y\in Y$ and a family of faces $
\cF_y
\subseteq
\{F\subseteq Y : F \text{ face, } y\in F\}$ with $|\cF_y|=q,$ we have
\begin{equation}
\sum_{y\in Y}\sum_{F\in\cF_y}\frac{1}{\abs{F}}
\geq
s_n(q).
\label{eq:face-shadow}
\end{equation}
\end{lemma}
We use this combinatorial statement into a converse bound 
for the transmittable of a single quantum via $\Phi^{\otimes n}$. For that we define
\begin{align}
\label{eq:DefVandPi}
V\eqdef\sum_{y\in\{0,1\}^n}E_y
\quad\text{and}\quad
\Pi\eqdef VV^*.
\end{align}
Note that since $E^*_yE_{y'} =0$ for $y\neq y'$ and further $\sum_{y} E^*_yE_y =\1,$ we have that $V$ is an isometry and that $\Pi$ is a projection in the output space of the channel $\Phi^{\otimes n}$ satsfying 
\begin{align}
\label{eq:PiRank}
    \operatorname{rank}(\Pi) = 3^n.
\end{align}
Furthermore, for $C\subseteq(\C^3)^{\otimes n}$ being a $q$-dimensional correctable subspace for
$\Phi^{\otimes n}$ and $P_C$ being the projection onto $C,$ which by the Knill-Laflamme condition satisfies 
$P_CM_yP_C=\lambda_yP_C$ for all $y\in\{0,1\}^n$ and some $\lambda_y\ge 0$, we can define the projection\footnote{To see that $Q_C$ is a projection note that by orthogonality of $E_y$, we have $Q^2_C = \sum_{y: \lambda_y>0}\lambda_y^{-2} E_yP_CM_yP_CE_y^* =  \sum_{y: \lambda_y>0}\lambda_y^{-1} E_yP_CE_y^* = Q_C.$}
\begin{equation}
\label{eq:QCprojection}
Q_C
\eqdef
\sum_{y: \lambda_y>0}\lambda_y^{-1}E_yP_CE_y^*,
\end{equation}
We can think of $Q_C$ projecting onto the output space of $C$ under the channel $\Phi^{\otimes n},$ i.e. more precisely onto the space $\operatorname{im}(\Phi^{\otimes n}(P_C)).$ The following lemma establishes a lower bound on the trace of the overlap of the projections $\Pi$ and $Q_C.$

\begin{lemma}\label{prop:oneblock}
Let $C\subseteq(\C^3)^{\otimes n}$ be a $q$-dimensional correctable subspace for
$\Phi^{\otimes n}$ and $P_C$ the projection onto $C,$ i.e. $P_CM_yP_C = \lambda_y P_C$ for all $y\in\{0,1\}^n$ and some $\lambda_y\ge 0.$ Furthermore, let $Q_C$ be the corresponding projection defined in \eqref{eq:QCprojection}. Then $q\leq 2^n$ 
and
\begin{equation}
\Tr(Q_C\Pi) =\sum_{y\, \text{s.t.}\, \lambda_y>0}\lambda_y^{-1}\Tr(P_CM_y^2) \ge s_n(q).
\label{eq:oneblock}
\end{equation}
\end{lemma}

\begin{proof}
By the
Knill-Laflamme condition we have
\begin{equation*}
P_CM_yP_C=\lambda_yP_C,
\qquad
\lambda_y\geq 0,
\qquad
\sum_y\lambda_y=1.
\end{equation*}
In particular this gives that the set $Y\eqdef\{y\in\{0,1\}^n:\lambda_y>0\}$ is non-empty. For $y\notin Y$, we have by definition $P_CM_yP_C=0$ and hence $M_y^{1/2}P_C=0$. Using
\eqref{eq:face-action}, we conclude that
\begin{equation}
C
\subseteq
\spa\{\ket{F}:F\subseteq Y\}.
\label{eq:support}
\end{equation}
For $y\in Y,$ we see by Knill-Laflamme that the operator
\begin{equation*}
Q_y
\eqdef
\lambda_y^{-1}M_y^{1/2}P_CM_y^{1/2}
\end{equation*}
is a projection with \begin{align}
\label{eq:Q_yProperties}
    \operatorname{rank}(Q_y) = \operatorname{rank}(P_C) =q \ \qquad\text{and}\qquad \ \operatorname{im}(Q_y) \subseteq
\spa\{\ket{F}:y\in F\subseteq Y\},
\end{align}
where the second statement follows by \eqref{eq:support} and
\eqref{eq:face-action}.  In particular this gives that each $y\in Y$ lies in at least $q$ distinct faces $F$ such that $F\subseteq Y$. The set $Y$ thus satisfies the assumptions of
Lemma~\ref{lem:face-shadow}, which in particular gives $q\le 2^n.$
By \eqref{eq:face-action} and \eqref{eq:Q_yProperties}, we see
\begin{equation}
\begin{split}
\lambda_y^{-1}\Tr(P_CM_y^2)
&=
\Tr(Q_yM_y)\\
&\geq
\min_{\substack{
\cF_y\subseteq\{F:y\in F\subseteq Y\}\\
\abs{\cF_y}=q
}}
\sum_{F\in\cF_y}\frac{1}{\abs{F}}.
\end{split}
\label{eq:kyfan}
\end{equation}
which by choosing for every $y\in Y$ a minimising 
family $\mathcal{F}_y$ in \eqref{eq:kyfan}  gives by Lemma~\ref{lem:face-shadow}
\begin{equation}
\sum_{y\in Y}\lambda_y^{-1}\Tr(P_CM_y^2)
\geq
s_n(q).
\label{eq:blocksum}
\end{equation}
Lastly, by definition of $\Pi$ and $Q_C$ in \eqref{eq:DefVandPi} and \eqref{eq:QCprojection}, we find 
\begin{equation*}
\Tr(Q_C\Pi)
=
\sum_{y\in Y}\lambda_y^{-1}\Tr(P_CM_y^2)
\end{equation*}
which finishes the proof.
\end{proof}

This implies the following converse bound on transmittable algebras via $\Phi^{\otimes n}$ and $\Psi^{\otimes n}.$ 

\begin{lemma}\label{lem:shadow}
Let $\scA\cong \bigoplus_j\M{q_j}$. If $\scA$ is transmittable through
$\Phi^{\otimes n}$, then
\begin{equation}
q_j\leq 2^n
\quad\text{for every }j,
\qquad
\sum_j s_n(q_j)\leq 3^n.
\label{eq:shadow-Phi}
\end{equation}
If $\scA$ is transmittable through $\Psi^{\otimes n}$, then
\begin{equation}
q_j\leq 2^n
\quad\text{for every }j,
\qquad
\sum_j s_n(q_j)\leq 6^n.
\label{eq:shadow-Psi}
\end{equation}
\end{lemma}

\begin{proof}
We first consider $\Phi^{\otimes n}$ and focus on the proof of \eqref{eq:shadow-Phi}. Let $P_j$ be the code projection
corresponding to the block $\M{q_j}$ and let $Q_j$ be the associated output
projection defined in \eqref{eq:QCprojection}. The hybrid Knill-Laflamme
conditions, see e.g.~Lemma~\ref{lemma:pairwise-S-orthogonal}, imply
\begin{equation*}
P_jM_yP_k=0
\qquad
\text{for }j\neq k.
\end{equation*}
Together with the orthogonality of the Kraus operators, i.e. $E^*_yE_{y'}=0$ for $y\neq y'$,  this shows that the
projections $Q_j$ are mutually orthogonal. From this we see $\sum_jQ_j\leq\1$ and, hence, \eqref{eq:PiRank}
Lemma~\ref{prop:oneblock} give
\begin{equation}
\label{eq:Qjargument}
\sum_js_n(q_j)
\leq
\sum_j\Tr(Q_j\Pi)
=
\Tr\left(\left(\sum_jQ_j\right)\Pi\right)
\leq
\Tr(\Pi)
=
3^n.
\end{equation}
The bound $q_j\leq 2^n$ also follows from
Lemma~\ref{prop:oneblock}, which finishes the proof of \eqref{eq:shadow-Phi}.

We now continue with the proof of \eqref{eq:shadow-Psi}. Using the Kraus operators
$\kb{0},\kb{1}$ for $\Delta_2$, the Kraus operators and effects of
$\Psi^{\otimes n}$ can be written as
\begin{equation}
\label{eq:DefE_ayandM_ay}
\widetilde E_{a,y}
\eqdef
\kb{a}\otimes E_y,
\qquad
\widetilde M_{a,y}
\eqdef
\widetilde E_{a,y}^*\widetilde E_{a,y}
=
\kb{a}\otimes M_y,
\end{equation}
where $a,y\in\{0,1\}^n$.

Assume first that $\mathscr{A}\cong \M{q}$ is transmittable through $\Psi^{\otimes n},$ which gives by Knill-Laflamme that  
\begin{equation}
\label{eq:KnillLafltilde}
P\widetilde M_{a,y}P
=
\lambda_{a,y}P
\end{equation}
for all $a,y\in\{0,1\}^n$ and some $\lambda_{a,y} \ge 0$ and where $P$ denotes the projection onto the block $\M{q}$ within the input space $\C^{2^n}\otimes \C^{3^n}.$
Note that since $\sum_y\widetilde M_{a,y}=\kb{a}\otimes \1$ for all $a\in\{0,1\}^n$, we have by \eqref{eq:KnillLafltilde} that  \begin{align}
\label{eq:pa}
P(\kb{a}\otimes \1) P = \sum_y \lambda_{a,y} P =: p_a P.\end{align}  
Fix now $a\in\{0,1\}^n$ such that $p_a>0.$ With that we can define operators $\widetilde P_a$ and $P_a$ by 
\begin{equation*}
\widetilde P_a
\eqdef
p_a^{-1}\left(\kb{a}\otimes \1\right) P\left(\kb{a}\otimes \1\right) = \kb{a}\otimes P_a
\end{equation*}
and note that by \eqref{eq:pa} both $\widetilde P_a$ and $P_a$ are a rank-$q$ projections.
Moreover, we see by \eqref{eq:KnillLafltilde} and the particular form of $\widetilde M_{a,y}$ in \eqref{eq:DefE_ayandM_ay} that 
\begin{equation*}
\widetilde P_a\widetilde M_{a,y}\widetilde P_a=
\frac{\lambda_{a,y}}{p_a} \widetilde P_a \quad \text{and hence } \quad P_aM_y P_a = \frac{\lambda_{a,y}}{p_a} P_a
\end{equation*}
Therefore, we can apply Lemma~\ref{prop:oneblock} to the code space projection $P_a$ which gives
\begin{align}
\nn \sum_{y \ \text{s.t.}\ \lambda_{a,y}>0}
\lambda_{a,y}^{-1}
\Tr(P\widetilde M_{a,y}^2)
&= \sum_{y \ \text{s.t.}\ \lambda_{a,y}>0}
\lambda_{a,y}^{-1}
\Tr\left(\left(\kb{a}\otimes \1\right) P\left(\kb{a}\otimes \1\right) \widetilde M_{a,y}^2\right)\\& =
\sum_{y \ \text{s.t.}\ \lambda_{a,y}>0}
\left(\frac{\lambda_{a,y}}{p_a}\right)^{-1}
\Tr( P_a M_{y}^2)  
\ge
s_n(q).
\label{eq:sector-charge}
\end{align}
and furthermore $q\leq 2^n.$ 

Define
\begin{equation*}
\widetilde V
\eqdef
\sum_{a,y\in\{0,1\}^n}\widetilde E_{a,y}
=
\1\otimes V,
\qquad
\widetilde\Pi
\eqdef
\widetilde V\widetilde V^*
=
\1\otimes\Pi.
\end{equation*}
Thus $\widetilde\Pi$ is a projection of rank $6^n$. Furthermore, define
analogously to \eqref{eq:QCprojection}
\begin{equation}
\label{eq:TildeQ_C}
\widetilde Q_C
\eqdef
\sum_{\substack{a,y\\\lambda_{a,y}>0}}
\lambda_{a,y}^{-1}\widetilde E_{a,y}P\widetilde E_{a,y}^*.
\end{equation}
From definition of $\widetilde Q_C$ and $\widetilde \Pi$ and furthermore using
\eqref{eq:sector-charge}, we obtain
\begin{equation}
\label{eq:OnTildeBlock}
\Tr(\widetilde Q_C\widetilde\Pi)
=
\sum_{\substack{a\\p_a>0}}
\sum_{\substack{y\\ \lambda_{a,y}>0}}
\lambda_{a,y}^{-1}
\Tr(P\widetilde M_{a,y}^2)
\geq
s_n(q).
\end{equation}
Hence coherent support across several selector sectors can only increase the
charge of a block.

Now for $\mathscr{A} \cong \bigoplus_j \M{q_j}$ being transmittable through $\Psi^{\otimes n}$, we let $\widetilde Q_j$ be the projection defined in  \eqref{eq:TildeQ_C} corresponding to the codespace of the $j^{th}$ block. By the same argument as for $Q_j$ around \eqref{eq:Qjargument}, we see that the $\widetilde Q_j$ projections are mutually orthogonal and hence $\sum_j \widetilde Q_j\le \1$ which gives by \eqref{eq:OnTildeBlock} 
\begin{equation}
\sum_js_n(q_j)
\leq
\sum_j\Tr(\widetilde Q_j\widetilde \Pi)
\le
\Tr(\widetilde \Pi)
=
6^n
\end{equation}
and finishes the proof.
\end{proof}

\begin{proof}[Proof of Proposition~\ref{prop:DoubleExp}]
We start by showing that for $1\le q\le 2^n$ the algebras $\scB_{n,q}$ defined in \eqref{eq:Bnq} are transmittable via $\Phi^{\otimes n}.$ For that define
\begin{equation}
\ket{u_0}\eqdef \ket{e_1},
\qquad
\ket{u_1}\eqdef\frac{\ket{e_0}+\ket{e_2}}{\sqrt{2}},
\label{eq:u01}
\end{equation}
which are orthonormal vectors in $\C^3.$
Denoting $C=\spa\{u_0,u_1\}$ and $P_C$ the orthogonal projection onto $C,$ we see from \eqref{eq:effects} that $P_CE^*_0E_0P_C =P_CE^*_1E_1P_C =\frac{1}{2} P_C$. Denoting for $x\in\{0,1\}^n$ the vector $\ket{u_x}:=\ket{u_{x_1}}\otimes\cdots\otimes \ket{u_{x_n}}$ and using the enumeration defined in \eqref{eq:Orderingfors_n}, this gives
\begin{align}
 C_q = \operatorname{span}\left\{ \ket{u_{x^{(i)}}}: i \in[q]\right\} \subseteq \C^{3^n}  
\end{align}
is a $q$-dimensional correctable subspace for $\Phi^{\otimes n}.$ Note that for each $x\in\{0,1\}^n$ the vector $\ket{u_x}$ can be expanded in exactly $2^{|x|}$ many computational basis vectors of $\C^{3^n}$ and furthermore for $x\neq x'$ the corresponding computational basis vectors in the expansions of $\ket{u_x}$ and $\ket{u_{x'}}$ are disjoint. From this and the definition of $s_n(q)$ in \eqref{eq:Defsn(q)} we see that $C_q$ is contained in the linear hull of $s_n(q) = \sum_{i=1}^q2^{|x^{(i)}|}$  computational basis vectors, i.e. using the face parametrisation \eqref{eq:FaceVectors} of the computational basis, we have for a suitable enumeration $\{F_j\}_{j=1}^{3^n}$ of all face sets that
\begin{align}
    C_q \subseteq \operatorname{span}\left\{\ket{F_j} : \, j \in[s_n(q)]\right\}.
\end{align}
Using \eqref{eq:OperatorSystemPhin} and  \eqref{eq:face-action} together with the Knill-Laflamme condition, e.g. more precisely Lemma~\ref{lemma:pairwise-S-orthogonal}, we see that $\mathscr{B}_{n,q} = \M{q}\oplus \C^{3^n-s_n(q)}$ can be transmitted via $\Phi^{\otimes n}$ by encoding the $q$-dimensional quantum block through the code space $C_q$ and the $3^n-s_n(q)$ classical blocks via the remaining orthogonal computational basis elements $\ket{F_j}$ for $j \in [3^n]\setminus [s_n(q)].$

Moreover, note that $\mathscr{B}_{n,q}$ saturates the converse bound \eqref{eq:shadow-Phi} in Lemma~\ref{lem:shadow} as $s_n(1)=1$ and therefore
\begin{equation}
\label{eq:CostB_nq}
s_n(q)
+
\left(3^n-s_n(q)\right)s_n(1)
=
3^n,
\end{equation}
which by \eqref{eq:cost-monotone} shows that $\mathscr{B}_{n,q}$ is in fact a maximal transmittable algebra of $\Phi^{\otimes n}$ according to Definition~\ref{def:max-dom}. As clearly $\mathscr{B_{n,q}} \not\cong \mathscr{B_{n,q'}}$ for $q\neq q',$ this shows $w_n(\Phi) \ge |\{\mathscr{B}_{n,q}: q\in[2^n]\}| = 2^n,$ which finishes the proof of \eqref{eq:Phi-lb}.

We continue to show $N = (N_1,\cdots, N_{2^n}) \in\N_0^{2^n}$ śuch that $\sum_{q=1}^{2^n} N_q =2^n$, the algebra $\mathscr{C}_N = \bigoplus_{q=1}^{2^n} \mathscr{B}^{\oplus N_q}_{n,q}$ is transmittable via the channel $\Psi^{\otimes n} \equiv \Delta^{\otimes n}_2\otimes\Phi^{\otimes n}.$  To see this we first note that $\mathscr{C}_N$ consists out of $2^n$ orthogonal copies of $\mathscr{B}_{n,q}$ with different values of $q$ and we can hence choose for all $i \in[2^n]$  a $q_i\in[2^n]$ such that
\begin{align}
    \mathscr{C}_N = \bigoplus_{i=1}^{2^n} \mathscr{B}_{n,q_i}.
\end{align}
This algebra can then be transmitted via $\Psi^{\otimes n}$ by transmitting the classical label $i\in [2^n] \equiv \{0,1\}^n$ through the ideal classical channel $\Delta^{\otimes n}_2$ and then, depending on $i,$ transmitting $\mathscr{B}_{n,q_i}$ via $\Phi^{\otimes n}$ using the encoding outlined above.

Moreover, using the corresponding calculation for $\mathscr{B}_{n,q}$ in \eqref{eq:CostB_nq}, we see that $\mathscr{C}_{N}$ saturates the converse bound \eqref{eq:shadow-Psi} in Lemma~\ref{lem:shadow} as by assumption
\begin{equation*}
\sum_{q=1}^{2^n} N_q 3^{n} = 6^n.
\end{equation*}
Combining this with \eqref{eq:cost-monotone} shows that $\mathscr{C}_{N}$ is in fact a maximal transmittable algebra of $\Psi^{\otimes n}$ according to Definition~\ref{def:max-dom}.  Furthermore, for $q\geq 2$, the multiplicity of the simple block $\M{q}$ in $\scC_N$
is exactly $N_q$ and hence $\scC_N\cong \scC_{N'},$ implies\footnote{We use that for natural numbers $d_j$ satisfying $d_{j}\neq d_{j'}$ for $j\neq j'$ we have that $\bigoplus_{j} \M{d_j}^{\oplus m_j} \cong \bigoplus_{j} \M{d_j}^{\oplus m'_j}$ that is equivalent to $m_j =m'_j$ for all $j.$}  $N_q=N'_q$ and by that also 
\begin{equation*}
N_1
=
2^n-\sum_{q=2}^{2^n} N_q = 2^n-\sum_{q=2}^{2^n} N'_q  = N'_1,
\end{equation*}
which in total gives $N=N'.$ From this we see
\begin{align}
    w_n(\Psi) \ge \left|\left\{\scC_N\,:\quad N \in\N^{2^n}_0,\,\text{ s.t.}\,\, \sum_{q=1}^{2^n} N_q =2^n\right\}\right| = \binom{2^{n+1}-1}{2^n-1} \ge 2^{2^n-1},
\end{align}
which finishes the proof of \eqref{eq:doubleexp}.
\end{proof}

\section{Discussion and outlook}
\label{sec:discussion}

We have developed a theory of exact hybrid classical-quantum communication via noisy channels in which the
transmitted information is described by a finite-dimensional $C^*$-algebra. In such an algebra $\scA\cong \bigoplus_k\M{d_k}$, the block label carries
classical information, while each block carries a quantum state whose
dimension may depend on that label. Embeddings compare these
communication tasks (Definition~\ref{def:embedding-preorder}), and $\scT(\Phi)$ collects all algebra types that a noisy
channel $\Phi$ can transmit with perfect recovery (Definition~\ref{def:alg-transmit}). Classical, quantum, and
hybrid zero-error capacities extract particular numerical features of this set.
The connection with Knill--Laflamme error correction provides concrete
conditions for realizing its elements as codes (Theorem~\ref{theorem:KL-alg2} and Lemma~\ref{lemma:pairwise-S-orthogonal}).

One of our central findings is a distinction between maximal and dominating transmittable algebras (Section~\ref{sec:dom}). We show that a channel can admit several incomparable transmittable algebras that are maximal with respect to the embedding order, so an optimal use of the channel depends on the operational task and the type of information one wishes to preserve. On the other hand, a dominating algebra, if it exists, presents a unique optimal way to use the channel for transmitting hybrid information. Hybrid capacities determine the
only possible dominating algebra type, called the capacity algebra, but a transmittable capacity algebra does not suffice to ensure that it is dominating (Section~\ref{sec:hybrid-capacities}). The finite
forbidden-type criterion supplies the missing condition by testing
whether any minimal algebra type outside the candidate's embedding lower set
is transmittable (Theorem~\ref{thm:dominating-forbidden-types} and
    Corollary~\ref{cor:complete-capacity-criterion}). Algebraic operator systems, up to graph equivalence,
and channels with $Q_0(\Phi)=0$ give explicit sufficient classes of channels with dominating algebras, while
the general criterion also applies beyond their union.

Tensor products reveal that this structure can change under joint
coding (Section~\ref{sec:tensor-products}). Channels that separately admit dominating algebras can lose
this property when used together (Theorem~\ref{thm:tensor-destroys-domination}). For a fixed channel $\Phi$, the number
$w_n(\Phi)$ of maximal transmittable algebra types can grow double exponentially with the
block-length (Proposition~\ref{prop:DoubleExp}), which is optimal in scale (Proposition~\ref{prop:width-upper-bound}). Such growth also rules out generating all transmittable types by tensor products and embeddings from any finite collection of codes at bounded block-lengths (Lemma~\ref{lemma:finite-generation-polynomial}). 

\subsection{Future directions} 

A first direction is to study the computational complexity of the finite domination criterion (Theorem~\ref{thm:dominating-forbidden-types} and
    Corollary~\ref{cor:complete-capacity-criterion}). There are two distinct tasks here: enumerating
the minimal forbidden types for any candidate algebra, and then deciding their transmittability.
Even when the list of forbidden types is available, excluding a type requires ruling out
every encoding satisfying its Knill-Laflamme conditions \eqref{eq:forbidden-type-KL}. Sharper
converses, certified relaxations, and channel symmetries could make
these tests tractable for further channel families. A complementary
structural question is which finite lower sets of algebra types can
occur as transmission sets for quantum channels.

For tensor products, it remains to identify channel classes for
which product coding describes every transmittable algebra type, or for which
a dominating algebra persists at every block-length. The converse
possibility is also of interest: can a channel with $w_1(\Phi)>1$
have $w_n(\Phi)=1$ at some larger block-length? More generally, it would be interesting to relate bounded, polynomial, exponential, or faster growth of $w_n$
to specific structural properties of the underlying noisy channel $\Phi$.

Sequential composition connects the theory to information storage \cite{Singh2025zero-markovian, Fawzi2026markovian, Singh2026markovian}.
For a channel acting on a fixed space, the sets $\scT(\Phi^l)$ shrink
with $l$ (Lemma~\ref{lemma:op-chain}), and sufficiently high powers have a dominating algebra
determined by the peripheral algebra (Theorem~\ref{theorem:dom-highly-divisible}). Describing the intermediate
transmission sets and finding sharper stabilization bounds would clarify how the information recoverable after a finite number of steps approaches the information that survives indefinitely.

Finally, it would be desirable to develop an approximate version of the
algebra transmission theory presented here, allowing a small recovery
error measured, for instance, in the diamond norm.

\section*{Acknowledgements}
The authors thank Omar Fawzi, Andreas Winter and Michael Wolf for insightful discussions.
RS acknowledges support by the European Research Council (ERC Grant %
Agreement No.~948139 and ERC Grant AlgoQIP, Agreement No. 851716), %
from the Excellence Cluster Matter and Light for Quantum Computing %
(ML4Q-2), from the QuantERA II Programme of the
European Union’s Horizon 2020 research and innovation programme %
under Grant Agreement No
101017733 (VERIqTAS) as well as the government grant managed by the %
Agence Nationale de la
Recherche under the Plan France 2030 with the reference %
ANR-22-PETQ-0007. SS acknowledges support from the Deutsche Forschungsgemeinschaft (DFG, German Research Foundation) via TRR 352 – Project-ID 470903074. 

\section*{AI declaration}

The core ideas, research questions, and central definitions for this project were conceived by the human authors. Some of the main results (Lemmas~\ref{lemma:EncDecTransmittableAlgebra}, \ref{lemma:pairwise-S-orthogonal},  Examples~\ref{ex:PedagogicalExampleWithoutDomAlgebra}, \ref{ex:NoDomSecondExample}, Theorem~\ref{thm:CapacitiesToAlgebra}, Corollary~\ref{cor:CapacitiesToAlgebra}, Theorem~\ref{thm:alg-max}, Corollary~\ref{corollary:alg-max}, Theorem~\ref{theorem:dom-highly-divisible}, Theorem~\ref{theorem:dom-classical}) were also first proved by the authors. Subsequently, AI (ChatGPT 5.6 Sol and ChatGPT 6 Astra) was used to explore alternative proof strategies and pursue new investigations. In particular, the converses in Section~\ref{subsec:converse} and the finite-domination criterion (Theorem~\ref{thm:dominating-forbidden-types} and
Corollary~\ref{cor:complete-capacity-criterion}) were developed with AI-assistance. Moreover, the double-exponential width growth example (Proposition~\ref{prop:DoubleExp}) is purely AI-generated, with the current exposition revised by the authors. Throughout, AI was used for proof-reading, to help with the writing process, and to improve the overall quality of exposition. The authors have verified all AI generated material and are responsible for the mathematical claims, the final presentation, and any remaining errors.

\section*{Appendices}
\addappheadtotoc
\begin{subappendices}
\renewcommand{\setthesubsection}{\Alph{subsection}}
\counterwithin{theorem}{subsection}
\renewcommand{\thetheorem}{\thesubsection.\arabic{theorem}}

\subsection{Technical proofs} \label{appen:tech-proofs}
\subsubsection{Second proof of Theorem~\ref{thm:alg-max}} \label{appen:alg-dom}

Knill-Laflamme conditions (Theorem~\ref{theorem:KL-alg2}) show that $S_{\Phi}'$ is transmittable via $\Phi$. 

Conversely, since $S_{\Phi}\subseteq \B{\Hil}$ is a unital $*-$algebra, there exists a decomposition $\Hil=\oplus A_k \otimes B_k$ such that $S_{\Phi}= \oplus_k 1_{A_k} \otimes \B{B_k}$. Hence, without loss of generality, we can take $\Phi=\cP_{S_{\Phi}'}$ to be the unique trace-preserving conditional expectation onto the commutant $S_{\Phi}'=\oplus_k \B{A_k}\otimes 1_{B_k}$. Now, consider a $*-$algebra $\scA=\oplus_k \M{d_k}$ that is transmittable via $\Phi$ with encoder $\cE$ and decoder $\cD$, so that $\cD \circ \Phi \circ \cE = \cD\circ \Phi \circ \Phi\circ \cE= \cP_{\scA}$ (see Lemma~\ref{lemma:EncDecTransmittableAlgebra}). We now mimic the proof strategy of \cite{delsol2025emulation}.
Let $P$ be the orthogonal projection onto $\supp\Phi\cE(1)\subseteq\Hil$.
Since $\Phi\cE(1)\in S_{\Phi}'$ and $S_{\Phi}'$ is a $*$-algebra,
its support projection $P\in S_{\Phi}'$.
For every $X\in\B{\Hil}$,
$\cE^*\Phi^*(X)=\cE^*\Phi^*(PXP)$. Moreover, 
\begin{equation}\label{eq:X>=0impliesPXP=0}
    X\geq0,\quad\cE^*\Phi^*(X)=0
    \quad\Longrightarrow\quad PXP=0.
\end{equation}
Indeed, the hypothesis implies $\Tr(\Phi\cE(1)X)=0$, and
$\Phi\cE(1)$ is strictly positive on $P\Hil$.
Define 
\begin{equation}
    \Theta (\cdot) := P (\Phi^* \circ\cD^*(\cdot)) P.
\end{equation}
Then, $\cE^* \Phi^* \Theta(A)=A$ for all $A\in \scA$. We claim that $\Theta|_{\scA} : \scA\to P \im (\Phi^*) P=PS_{\Phi}'P$ is an injective $*$-homomorphism. Injectivity follows from the fact that $\Theta|_{\scA}$ admits a left-inverse. Moreover, two applications of the Schwarz inequality give, for all $A\in\scA$,
\begin{align}
   A^{*}A=\cE^*\Phi^* (\Theta(A^{*}A))
   &\geq \cE^*\Phi^*(\Theta(A^{*})\Theta(A)) \nonumber\\
   &\geq \cE^*\Phi^*(\Theta(A^{*}))\cE^*\Phi^*(\Theta(A))=A^{*}A.
\end{align}
The Schwarz defect
$\Delta_A:=\Theta(A^*A)-\Theta(A)^*\Theta(A)$ is positive and satisfies
$\cE^*\Phi^*(\Delta_A)=0$. The implication \eqref{eq:X>=0impliesPXP=0} for positive operators
therefore gives $P\Delta_AP=0$. Since $\Delta_A=P\Delta_AP$, we obtain
$\Delta_A=0$. Applying the same argument to $A^*$ shows that every
$A\in\scA$ lies in the multiplicative domain of $\Theta$
\cite{Choi1974schwarz}, which proves the claim. Consequently,
\begin{equation}
    \scA \leq P S_{\Phi}'P \leq S_{\Phi}',
\end{equation}
where the final inequality holds because $P\in S_{\Phi}'$.

\subsubsection{Proof of Lemma~\ref{lem:face-shadow}}
\label{app:face-shadow}

Let $Y\subseteq\{0,1\}^n$ be non-empty and denote by $m_Y
\eqdef
\abs{
\left\{
(y,F):
y\in Y,\;
F\subseteq Y \text{ a face},\;
y\in F
\right\}
}.$
Choose an enumeration $
(y^{(1)},F^{(1)}),\ldots,(y^{(m_Y)},F^{(m_Y)})
$
of these pairs such that
\begin{equation*}
\frac{1}{\abs{F^{(1)}}}
\leq
\frac{1}{\abs{F^{(2)}}}
\leq
\cdots
\leq
\frac{1}{\abs{F^{(m_Y)}}}.
\end{equation*}
For $1\leq r\leq m_Y$, define
\begin{equation}
R_Y(r)
\eqdef
\sum_{i=1}^r\frac{1}{\abs{F^{(i)}}}.
\label{eq:RY}
\end{equation}
As for the definition of $s_n(q)$, the value of $R_Y(r)$ is independent of
the particular enumeration satisfying the above ordering. Furthermore, for $t\geq0$ define
\begin{equation}
G_Y(t)
\eqdef
\sum_{i=1}^{m_Y}
\left(
t-\frac{1}{\abs{F^{(i)}}}
\right)_+.
\label{eq:GY}
\end{equation}

\begin{proof}[Proof of Lemma~\ref{lem:face-shadow}]
We first note that $q\leq\abs{Y}$. Indeed, fix $y\in Y$. For every face
$y\in F\subseteq Y$, let $y^F\in\{0,1\}^n$ denote the vertex of $F$ opposite to $y$, i.e. for $F = F_1\times \cdots \times F_n$ and $F_i\in\{\{0\},\{1\},\{0,1\}\}$ this bit string $y^F$ can be defined component wise as
\begin{align}
y^F_i =\begin{cases}
y_i,\quad \text{for } F_i =\{0\},\{1\},\\
1-y_i, \quad \text{for } F_i =\{0,1\}.
\end{cases}
\end{align}
The map $F\mapsto y^F$ for $F$ face such that $y\in F$ is injective and satisfies $y^F\in Y$. Hence there
are at most $\abs{Y}$ faces $F\subseteq Y$ containing $y$. By assumption
there are at least $q$ such faces, and therefore $
q\leq\abs{Y}\leq2^n.$

We next show by induction on $n$ that
\begin{equation}
G_Y(t)
\leq
G_{\{0,1\}^n}
\left(
\frac{\abs{Y}}{2^n}t
\right)
\label{eq:GYbound}
\end{equation}
for every $t\geq0$. The statement is immediate for $n=0$. For the induction
step, define
\begin{equation*}
A
\eqdef
\{x\in\{0,1\}^{n-1}:(x,0)\in Y\},
\qquad
B
\eqdef
\{x\in\{0,1\}^{n-1}:(x,1)\in Y\}.
\end{equation*}
Every face $F\subseteq Y$ either has its last coordinate fixed to $0$, fixed
to $1$, or has its last coordinate free. In the last case,
$F=F'\times\{0,1\}$ for a face $F'\subseteq A\cap B$. Since then
$\abs{F}=2\abs{F'}$, it follows directly from \eqref{eq:GY} that
\begin{equation*}
G_Y(t)
=
G_A(t)+G_B(t)+G_{A\cap B}(2t).
\end{equation*}

Put
\begin{equation*}
\alpha\eqdef\frac{\abs{A}}{2^{n-1}},
\qquad
\beta\eqdef\frac{\abs{B}}{2^{n-1}},
\qquad
g\eqdef G_{\{0,1\}^{n-1}},
\end{equation*}
and assume without loss of generality that $\alpha\geq\beta$. Since
$\abs{A\cap B}\leq\abs{B}$, the induction hypothesis and monotonicity of $g$
give
\begin{equation*}
G_Y(t)
\leq
g(\alpha t)+g(\beta t)+g(2\beta t).
\end{equation*}
Since $g$ is convex and $\alpha\geq\beta$, we have
\begin{equation*}
g(\alpha t)+g(2\beta t)
\leq
g((\alpha+\beta)t)+g(\beta t).
\end{equation*}
Moreover, $\beta\leq(\alpha+\beta)/2$, and hence
\begin{align}
G_Y(t)
\leq
g((\alpha+\beta)t)+2g(\beta t)
\leq
g((\alpha+\beta)t)
+
2g\left(\frac{\alpha+\beta}{2}t\right)
=
G_{\{0,1\}^n}
\left(
\frac{\alpha+\beta}{2}t
\right).
\end{align}
As $(\alpha+\beta)/2=\abs{Y}/2^n$, this proves \eqref{eq:GYbound}.

We now relate $G_Y$ to $R_Y$. From the ordering in \eqref{eq:RY}, for every
integer $1\leq r\leq m_Y$ we have
\begin{equation}
R_Y(r)
=
\sup_{t\geq0}
\left\{
rt-G_Y(t)
\right\}.
\label{eq:RYGY}
\end{equation}
Indeed, choosing $t$ between
$1/\abs{F^{(r)}}$ and $1/\abs{F^{(r+1)}}$ gives
\begin{equation*}
rt-G_Y(t)
=
\sum_{i=1}^r\frac{1}{\abs{F^{(i)}}}
=
R_Y(r).
\end{equation*}

Applying \eqref{eq:RYGY} with $r=q\abs{Y}$ and using
\eqref{eq:GYbound}, we obtain
\begin{align}
R_Y(q\abs{Y})
&=
\sup_{t\geq0}
\left\{
q\abs{Y}t-G_Y(t)
\right\}\geq
\sup_{t\geq0}
\left\{
q\abs{Y}t
-
G_{\{0,1\}^n}
\left(
\frac{\abs{Y}}{2^n}t
\right)
\right\}=
R_{\{0,1\}^n}(q2^n),
\end{align}
where in the last equality we substituted
$u=\abs{Y}t/2^n$.

It remains to calculate $R_{\{0,1\}^n}(q2^n)$. For every
$d\in\{0,\ldots,n\}$ there are
$
2^{n-d}\binom{n}{d}$
faces $F\subseteq\{0,1\}^n$ with $\abs{F}=2^d$. Each such face contains
$2^d$ vertices, and hence the value $2^{-d}$ occurs
$2^n\binom{n}{d}$ times in the ordering defining
$R_{\{0,1\}^n}$. Consequently,
\begin{align}
R_{\{0,1\}^n}(q2^n)
=
2^n
\sum_{i=1}^q
2^{-n}2^{\abs{x^{(i)}}}=
\sum_{i=1}^q2^{\abs{x^{(i)}}}=
s_n(q),
\end{align}
where we used the enumeration $x^{(1)},\ldots,x^{(2^n)}$ from
\eqref{eq:Orderingfors_n}.

Finally, the families $\cF_y$ appearing in Lemma~\ref{lem:face-shadow}
contain exactly $q\abs{Y}$ pairs $(y,F)$. By the definition of $R_Y$,
their sum is therefore bounded from below by
\begin{equation*}
\sum_{y\in Y}\sum_{F\in\cF_y}\frac{1}{\abs{F}}
\geq
R_Y(q\abs{Y})
\geq
s_n(q),
\end{equation*}
which finishes the proof.
\end{proof}

\subsection{Further examples}
\subsubsection{Channel with dominating algebra
whose operator system is not graph-isomorphic to an operator algebra} \label{appen:SnotAlgebra}

Let $\Phi:\M{4} \to \M{4}$ be defined as $\Phi(\cdot) = E_0 (\cdot) E^*_0 + E_1 (\cdot) E^*_1$ with Kraus operators
\begin{align}
    E_0 = \kb{0} +\kb{1},\qquad E_1 = \ket{1}\!\bra{2} + \kb{3}.
\end{align}
The operator system of $\Phi$ is explicitly given by
\begin{align}
S_{\Phi} &=
\operatorname{span}\Big\{
|0\rangle\!\langle 0|+|1\rangle\!\langle 1|,
 |2\rangle\!\langle 2|+|3\rangle\!\langle 3|,
 |1\rangle\!\langle 2|,
 |2\rangle\!\langle 1|
\Big\} \\ 
&= \left\{\begin{pmatrix} a & 0 & 0& 0\\ 0& a &\gamma_1& 0\\0&\gamma_2&b&0\\0& 0 &0&b\end{pmatrix}\,:  a,b,\gamma_1,\gamma_2\in\C\right\} \label{eq:SPhi-not-algebra}
\end{align}
We first prove that $\Phi$ has a dominating transmittable algebra $\scA_{\operatorname{dom}}(\Phi) \cong\M{2} \oplus \C $.
To see this, we first note that $\M{2} \oplus \C$ is transmittable via $\Phi$. Indeed, we can use the following encoding map and decoding maps and apply Lemma~\ref{lemma:EncDecTransmittableAlgebra}: \begin{align}
    \cE :\M{3} &\to \M{4}, \quad X = (x_{ij})_{i,j=0}^2 \mapsto  \begin{pmatrix}
        x_{00} & x_{01} & 0 & 0\\
        x_{10} & x_{11} & 0 & 0\\
        0& 0& 0  &0  \\
        0 &  0 & 0 & x_{22} 
    \end{pmatrix},\\ 
    \cD : \M{4} &\to \M{3},\quad
         Y = (y_{ij})_{i,j=0}^3 \mapsto \begin{pmatrix}
        y_{00}+y_{22} & y_{01} & 0\\
        y_{10} & y_{11} & 0 \\
    0 & 0 & y_{33} 
    \end{pmatrix}.
\end{align}
To show that $\M{2}\oplus \C$ is also dominating, we first note that any transmittable algebra $\scA\cong \oplus_k \M{d_k}$ satisfies $\max_k d_k \leq \max_{i} \operatorname{rank}(E_i)=2$ according to Lemma~\ref{lemma:KrausRankConverse}. Moreover, $\sum_k d_k \leq 3$, since otherwise if $\sum_k d_k=4$, Lemma~\ref{lemma:spectral-converse} implies that $S_{\Phi}$ must be commutative, which is not the case.\footnote{Alternatively, this can be seen that by using Lemma~\ref{lemma:KrausRankConverse} which gives $\sum_k d_k \le \dim(\operatorname{im}(E_0)+\operatorname{im}(E_0)) = \dim(\operatorname{span}(\ket{0},\ket{1},\ket{3}))=3.$} This leaves $\scA \in \{\C, \C \oplus \C, \C\oplus\C \oplus C, \M{2}, \M{2} \oplus \C \}$ as the only allowed possibilities. Hence, every transmittable algebra $\scA \leq \M{2}\oplus \C$.

We claim that $S_{\Phi}$ is not graph isomorphic to any finite-dimensional $*$-algebra.

\begin{lemma}
    The operator system $S_{\Phi}$ in \eqref{eq:SPhi-not-algebra} is \emph{not} graph isomorphic to any finite-dimensional $*$-algebra.
\end{lemma}

\begin{proof}
Note that our previous calculations show that $\alpha(S_{\Phi})=3$ and $\alpha_q(S_{\Phi})=2.$ Assume, for contradiction, that $S_{\Phi}$ is graph isomorphic to a finite-dimensional $*$-algebra $T$. Since graph isomorphism preserves $\alpha$ and $\alpha_q$ (Lemma~\ref{lemma:op-bottleneck}), it must be the case that $T \cong (I_2\otimes \cL(G)) \oplus \cL(H)$
for some finite-dimensional Hilbert spaces $G,H$ (Lemma~\ref{lemma:alpha-algebra}).

Since $S_{\Phi} \longrightarrow T$, there exists an isometry $V:\mathbb C^4 \to (\mathbb C^2\otimes G)\oplus H$ such that
\begin{equation}
    V^{*}\big((I_2\otimes \cL(G))\oplus \cL(H)\big)V \subseteq S_{\Phi},
\end{equation}
where we have absorbed the ancilla $E$ in Definition~\ref{def:op-homo} into $G,H$. Write
\begin{equation}
V|i\rangle = \ket{u_i} \oplus \ket{h_i},
\qquad
\ket{u_i}\in \mathbb C^2\otimes G,  \ket{h_i}\in H.    
\end{equation}

Using the zero pattern of $S_{\Phi}$, the $(0,1)$-entry of every compressed operator vanishes, while the $(0,0)$- and $(1,1)$-entries agree. Hence for all $X\in \cL(G)$ and $Y\in \cL(H)$,
\begin{align}
\langle u_0 | I_2\otimes X | u_1\rangle + \langle h_0 | Y |h_1\rangle &= 0, \\
\langle u_0 | I_2\otimes X | u_0\rangle - \langle u_1 |I_2\otimes X | u_1\rangle
+
\langle h_0 | Y |h_0\rangle - \langle h_1 | Y |h_1\rangle &= 0.
\end{align}
Setting $X=0$ shows $|h_0\rangle\langle h_1|=0,
|h_0\rangle\langle h_0|=|h_1\rangle\langle h_1|,$
hence $h_0=h_1=0$. The same argument with the pair $(2,3)$ gives $h_2=h_3=0$. Therefore, we can restrict the isometry as 
\begin{equation}
    V:\mathbb C^4 \to \mathbb C^2\otimes G \quad \text{such that} \quad V^{*}(I_2\otimes \cL(G))V \subseteq S_{\Phi}.
\end{equation} 
Now, let $\rho_{ij}:=\operatorname{Tr}_{\mathbb C^2}(|u_i\rangle\langle u_j|)\in \cL(G)$, so that for all $X\in \cL(G)$, $\langle i|V^{*}(I_2\otimes X)V|j\rangle = \operatorname{Tr}(X\rho_{ji})$. Again, using the zero pattern of $S_{\Phi}$, we obtain $\rho_{00}=\rho_{11},\rho_{22}=\rho_{33}$ and $0=\rho_{01}=\rho_{23}=\rho_{02}=\rho_{03}$. We now use the following lemma.

\begin{lemma}\label{lemma:uv}
    If $\ket{u},\ket{v}\in \mathbb C^2\otimes G$ are nonzero and satisfy
\begin{equation}
\operatorname{Tr}_{\mathbb C^2}(|u\rangle\langle u|)
=
\operatorname{Tr}_{\mathbb C^2}(|v\rangle\langle v|),
\qquad
\operatorname{Tr}_{\mathbb C^2}(|u\rangle\langle v|)=0,
\end{equation}
then there exist nonzero $\ket{g}\in G$ and orthogonal unit vectors $\ket{x},\ket{x'}\in \mathbb C^2$ such that
\begin{equation}
\ket{u}= \ket{x}\otimes \ket{g},
\qquad
\ket{v}= \ket{x'}\otimes \ket{g}.
\end{equation}
\end{lemma}
\begin{proof}
Write
$u=\sum_{i=0}^1\ket i\otimes \ket{g_i}$ and
$v=\sum_{i=0}^1\ket i\otimes \ket{h_i}$, and define
$U,V:\C^2\to G$ by $U\ket i= \ket{g_i}$ and $V\ket i= \ket{h_i}$.
The assumptions become
\[
UU^*=VV^*=:R\neq0,\qquad UV^*=0.
\]
If $r=\operatorname{rank}R$, then
$\operatorname{rank}U=\operatorname{rank}V=r$.
The equation $UV^*=0$ implies that
$\operatorname{im}U^*$ and $\operatorname{im}V^*$ are
orthogonal subspaces of $\C^2$. Hence $2r\leq2$, and
therefore $r=1$.
Write $R=\ket g\bra g$ with $g\neq0$.
Since the two reduced states equal $R$, there are unit
vectors $x,x'\in\C^2$ such that
$u=x\otimes g$ and $v=x'\otimes g$.
Finally,
\[
0=\operatorname{Tr}_{\C^2}(\ket u\bra v)
=\langle x'|x\rangle\ket g\bra g,
\]
so $x\perp x'$.
\end{proof}

Applying the lemma to $\ket{u_0}, \ket{u_1}$ gives
\begin{equation}
\ket{u_0}= \ket{x_0}\otimes \ket{g},
\qquad
\ket{u_1}= \ket{x_1}\otimes \ket{g},
\end{equation}
with $x_0\perp x_1$ and $g\neq 0$. Applying it to $\ket{u_2}, \ket{u_3}$ gives
\begin{equation}
\ket{u_2} = \ket{y_0} \otimes \ket{\tilde g},
\qquad
\ket{u_3} = \ket{y_1}\otimes \ket{\tilde g},
\end{equation}
with $y_0\perp y_1$ and $\tilde g\neq 0$. Finally, $\rho_{02}=\rho_{03}=0$ shows
\begin{equation}
0=\operatorname{Tr}_{\mathbb C^2}(|u_0\rangle\langle u_2|)
=\langle x_0 |y_0\rangle\,|g\rangle\langle \tilde g|,
\end{equation}
\begin{equation}
0=\operatorname{Tr}_{\mathbb C^2}(|u_0\rangle\langle u_3|)
=\langle x_0 |y_1\rangle\,|g\rangle\langle \tilde g|.
\end{equation}
Since $g,\tilde g\neq 0$, it follows that $\langle x_0|y_0\rangle=\langle x_0|y_1\rangle=0.$ But $y_0,y_1$ are orthonormal in $\mathbb C^2$. Therefore, $x_0=0$, contradiction. 

\end{proof}

\subsubsection{Dominating algebras under composition of channels}

In the following, we construct a channel $\Phi$ admitting a
dominating transmittable algebra, whereas $\Phi^2$ does not. Thus, admitting a dominating transmittable algebra is not
preserved under composition. Together with
Theorem~\ref{theorem:dom-highly-divisible}, which guarantees a
dominating transmittable algebra for $\Phi^l$ whenever $l\geq16,$
this also shows that $|\max\scT(\Phi^l)|$ need not be monotonic in $l.$

Let $\Hil=\C^4$ with orthonormal basis $\{\ket{0},\ldots,\ket{3}\}$
and set $\ket{\pm}:=(\ket{0}\pm\ket{2})/\sqrt2.$ We define
$\Phi:\M{4}\to\M{4}$ by
$\Phi(\rho)=E_0\rho E_0^*+E_1\rho E_1^*,$ where
\begin{align}
    E_0&=\ket{0}\!\bra{+}+\kb{1}+\ket{3}\!\bra{-},\nonumber\\
    E_1&=\ket{-}\!\bra{3}.
    \label{eq:composition-partial-isometries-kraus}
\end{align}

\begin{proposition}
\label{prop:composition-partial-isometries}
The channel $\Phi$ defined in
\eqref{eq:composition-partial-isometries-kraus} satisfies
\begin{align}
    \scA_{\dom}(\Phi)&\cong\M{3},\nonumber\\
    \max\bigl(\scT(\Phi^2)\bigr)&=\{\M{2},\C^3\}.
    \label{eq:composition-partial-isometries-maximal}
\end{align}
Hence, $\Phi$ admits a dominating transmittable algebra, whereas
$\Phi^2$ does not.
\end{proposition}

\begin{proof}
Let $C=\operatorname{span}\{\ket{0},\ket{1},\ket{2}\}$ and denote
its orthogonal projection by $P_C.$ We have
\begin{align*}
    P_CE_i^*E_jP_C=\delta_{i0}\delta_{j0}P_C,
\end{align*}
for $i,j=0,1.$
Thus, the Knill-Laflamme conditions, Theorem~\ref{theorem:KL-alg2}, show that $C$ is a three-dimensional
code space for $\Phi,$ so $\M{3}$ is transmittable.
To prove domination, suppose that
$\scA\cong\bigoplus_{k=1}^m\M{d_k}\in\scT(\Phi)$ with $\sum_kd_k\geq4.$
Let $i_k\in\{0,1\}$ be the Kraus indices from
Lemma~\ref{lemma:KrausRankConverse} and set
$I_j:=\{k:i_k=j\}.$
Applying \eqref{eq:ConverseOrthogonality2} in
Remark~\ref{remark:KrausSequentialConverse} to each subfamily gives
\begin{align*}
    \sum_{k\in I_0}d_k&\leq\operatorname{rank}(E_0)=3,
    &\sum_{k\in I_1}d_k&\leq\operatorname{rank}(E_1)=1.
\end{align*}
Since $\sum_kd_k\geq4,$ we have $\sum_{k\in I_0}d_k=3$ and $\sum_{k\in I_1}d_k=1.$
We order the unique index in $I_1$ last, so that
$d_m=1$ and $V_m=\operatorname{im}(E_0).$
Using $\operatorname{rank}(E_0^*E_1)=1,$
\eqref{eq:ConverseOrthogonality2} now gives
\begin{align*}
    1=d_m\leq1-[1-3+3]_+=0,
\end{align*}
a contradiction. Thus, $\sum_kd_k\leq3.$
Hence, every transmittable algebra embeds into $\M{3}.$ This
proves the first equality in
\eqref{eq:composition-partial-isometries-maximal}.

We now consider $\Phi^2.$ Set
$\ket{\psi}:=(\ket{0}+\ket{3})/\sqrt2.$ Since $E_1^2=0,$
the Kraus operators of $\Phi^2$ are given by
\begin{align*}
    F_0:=E_0^2=\ket{\psi}\!\bra{+}+\kb{1},\qquad
    F_1:=E_1E_0=\kb{-},\qquad
    F_2:=E_0E_1=\kb{3}.
\end{align*}
 The channel
acts isometrically on the code space
$C'=\operatorname{span}\{\ket{+},\ket{1}\},$
since $F_0|_{C'}$ is an isometry and $F_1|_{C'}=F_2|_{C'}=0.$
Furthermore, the three mutually orthogonal states
$\kb{1},\kb{-},\kb{3}$ are fixed by $\Phi^2.$ Thus,
Lemma~\ref{lemma:pairwise-S-orthogonal} gives
$\M{2},\C^3\in\scT(\Phi^2).$

Since $F_0$ is the only Kraus operator of rank two,
\eqref{eq:KrausSaturatedTwoBlock} in
Remark~\ref{remark:KrausSequentialConverse} rules out
$\M{2}\oplus\C\in\scT(\Phi^2),$ as this would require
\begin{align*}
    1\leq\operatorname{rank}(F_j)
      -\operatorname{rank}(F_0^*F_j)=0
\end{align*}
for some $j\in\{0,1,2\}.$

Finally, Lemma~\ref{lemma:EncDecTransmittableAlgebra} implies
$\scT(\Phi^2)\subseteq\scT(\Phi),$ since the second use can be
included in the decoder. Thus, $\sum_kd_k\leq3$ for every
$\bigoplus_k\M{d_k}\in\scT(\Phi^2),$ while
Lemma~\ref{lemma:KrausRankConverse} gives $\max_kd_k\leq2.$
Together with $\M{2}\oplus\C\notin\scT(\Phi^2),$ these bounds
prove the second equality in
\eqref{eq:composition-partial-isometries-maximal}.
\end{proof}

\subsection{Order-theoretic terminology}
\label{appen:order-theory}

In this appendix we collect the order-theoretic terminology used in the paper.
We follow standard conventions from the theory of partially ordered sets
\cite{DaveyPriestley2002,Trotter1992,Birkhoff1967}.

\subsubsection{Preorders and partially ordered sets}

A \emph{preorder} on a set $P$ is a binary relation $\leq$ which is reflexive
and transitive. Thus,
\begin{equation}\label{eq:B1}
\forall x,y,z\in P: \qquad    x\leq x,
    \quad
    x\leq y, \, y\leq z  \implies
    x\leq z.
\end{equation}
A preorder need not be antisymmetric. We write
\begin{equation*}
    x\sim y
    \quad\Longleftrightarrow\quad
    x\leq y\text{ and }y\leq x.
\end{equation*}
Then $\sim$ is an equivalence relation, and the quotient $P/{\sim}$ is
partially ordered by the relation induced from $\leq$. A \emph{partially ordered set}, or \emph{poset}, is a set $P$ equipped with a
binary relation $\leq$ which is reflexive, transitive, and antisymmetric. Thus, in addition to \eqref{eq:B1}, we have
\begin{equation}
    x\leq y,\ y\leq x
    \quad\Longrightarrow\quad
    x=y.
\end{equation}

In this paper, the embedding relation of Definition~\ref{def:embedding-preorder}
is a preorder on finite-dimensional $C^*$-algebras. After identifying
$*$-isomorphic algebras, it becomes a partial order.

\subsubsection{Lower sets and downward closure}

Let $(P,\leq)$ be a poset. A subset $I\subseteq P$ is called a \emph{lower set},
(or \emph{down-set}, or \emph{order-ideal}) if
\begin{equation}
    x\in I,\ y\leq x
    \quad\Longrightarrow\quad
    y\in I.
\end{equation}
Equivalently, once an element belongs to $I$, every element below it also
belongs to $I$.

For a subset $S\subseteq P$, its \emph{downward closure} is
\begin{equation}
    \downarrow S
    :=
    \{y\in P:\exists\,x\in S\text{ such that }y\leq x\}.
\end{equation}
This is the smallest lower set containing $S$. For a single element $x\in P$,
we write
\begin{equation}
    \downarrow x:=\downarrow\{x\}.
\end{equation}
A lower set of the form $\downarrow x$ is called \emph{principal}.

For a channel $\Phi$, the transmission set $\scT(\Phi)$ is a finite lower set in the
algebra-type poset: if $\scA\in\scT(\Phi)$ and $\scB\leq\scA$, then
$\scB\in\scT(\Phi)$.

\subsubsection{Antichains and maximal elements}

Two elements $x,y\in P$ are called \emph{comparable} if $x\leq y$ or $y\leq x$.
Otherwise they are \emph{incomparable}. A subset $A\subseteq P$ is an
\emph{antichain} if its distinct elements are pairwise incomparable:
\begin{equation}
    x,y\in A,\ x\neq y
    \quad\Longrightarrow\quad
    x\nleq y\text{ and }y\nleq x.
\end{equation}

Let $I\subseteq P$. An element $x\in I$ is \emph{maximal in $I$} if there is no
strictly larger element of $I$ above it. Equivalently,
\begin{equation}
    y\in I,\ x\leq y
    \quad\Longrightarrow\quad
    y=x.
\end{equation}
We denote the set of maximal elements of $I$ by $\max(I)$. The set $\max(I)$ is always an antichain. If $I$ is a finite lower set, then $I$ is determined by its maximal elements:
\begin{equation}
    I=\downarrow\max(I).
\end{equation}
Conversely, every finite antichain $A$ determines a lower set
$\downarrow A$. 

In this paper, the maximal transmittable algebras of a channel $\Phi$ are
precisely the elements of $\max(\scT(\Phi))$. Moreover, since $\scT(\Phi)$ is a finite lower set, we have
\begin{equation}
    \scT(\Phi) = \downarrow \max(\scT(\Phi)).
\end{equation}

\subsubsection{Top elements and principal lower sets}

Let $I\subseteq P$. An element $t\in I$ is called a \emph{top element}, or
\emph{greatest element}, of $I$ if
\begin{equation}
    x\leq t
    \qquad
    \text{for every }x\in I.
\end{equation}
A top element, if it exists, is unique.

Every top element is maximal, but a maximal element need not be a top element.
A finite lower set $I$ has a top element if and only if it is principal:
\begin{equation}
    I=\downarrow t
\end{equation}
for some $t\in I$. Equivalently,
\begin{equation}
    I\text{ has a top element}
    \quad\Longleftrightarrow\quad
    |\max(I)|=1.
\end{equation}

In the terminology of this paper, a dominating transmittable algebra for
$\Phi$ is a top element of $\scT(\Phi)$. Thus $\Phi$ has a dominating
transmittable algebra if and only if
\begin{equation}
    |\max(\scT(\Phi))|=1.
\end{equation}

\subsubsection{Upper and lower bounds}

Let $S\subseteq P$. An element $u\in P$ is an \emph{upper bound} for $S$ if
\begin{equation}
    x\leq u
    \qquad
    \text{for every }x\in S.
\end{equation}
An upper bound $u$ is a \emph{least upper bound}, or \emph{supremum}, if $u\leq v$ for every upper bound $v$ of $S$. If it exists, it is unique and is denoted
$\sup S$ or $\bigvee S$.

Similarly, an element $\ell\in P$ is a \emph{lower bound} for $S$ if
\begin{equation}
    \ell\leq x
    \qquad
    \text{for every }x\in S.
\end{equation}
A lower bound $\ell$ is a \emph{greatest lower bound}, or \emph{infimum}, if $v\leq \ell$ for every lower bound $v$ of $S$. If it exists, it is denoted $\inf S$ or
$\bigwedge S$.

The algebra-type poset need not have suprema or infima for arbitrary finite
subsets. In such cases, it is still useful to consider minimal
upper and maximal lower bounds. An upper bound $u$ of $S$ is a \emph{minimal upper bound} if no
strictly smaller upper bound lies below it. Equivalently, if $v$ is another
upper bound for $S$ and $v\leq u$, then $v=u$. Similarly, a lower bound $\ell$ of $S$ is a \emph{maximal lower bound} if no
strictly larger lower bound lies above it. Equivalently, if $v$ is another
lower bound for $S$ and $\ell\leq v$, then $v=\ell$.

In the algebra-type poset for example, the set $\{\M{3}, \M{2}\oplus\M{2}\}$ has no supremum, since it admits two incomparable minimal upper bounds given by $\M{4}$ and $\M{3}\oplus\M{2}.$

\end{subappendices}

\bibliographystyle{plainurl}
\bibliography{references}

\end{document}

%% file: channel-classes-tikz.tex
\begin{tikzpicture}[
    x=1cm,y=1cm,
    font=\fontsize{10}{12}\selectfont,
    every node/.style={align=center,inner sep=2pt},
    region/.style={rounded corners=5pt,line width=0.7pt},
    witness/.style={font=\fontsize{9}{10.8}\selectfont,text=black!75}
]
  \definecolor{ccAll}{HTML}{F2F2F2}
  \definecolor{ccFrame}{HTML}{82758E}
  \definecolor{ccCap}{HTML}{F5F2F8}
  \definecolor{ccTrans}{HTML}{EDF3FA}
  \definecolor{ccBlue}{HTML}{527EA0}
  \definecolor{ccBlueFill}{HTML}{EAF3FB}
  \definecolor{ccGold}{HTML}{AD863B}
  \definecolor{ccGoldFill}{HTML}{FCF3DE}
  \definecolor{ccOverlap}{HTML}{EEF1E2}
  \definecolor{ccPurple}{HTML}{79629A}

  % All channels, including those with no defined capacity algebra.
  \draw[region,draw=black!60,fill=ccAll]
    (-0.3,-0.3) rectangle (16.1,11.35);
  \node[font=\fontsize{10}{12}\selectfont\bfseries] at (7.9,10.94)
    {All quantum channels};
  \node[witness] at (7.9,10.50)
    {The capacity algebra $\mathscr A_{\mathrm{cap}}(\Phi)$ need not be defined};

  % Three nested existence conditions inside the class of all channels.
  \draw[region,draw=ccFrame,fill=ccCap]
    (0,0) rectangle (15.8,10.1);
  \node at (7.9,9.65)
    {$\mathscr A_{\mathrm{cap}}(\Phi)$ is defined};

  \draw[region,draw=ccBlue!80!black,fill=ccTrans]
    (0.3,0.3) rectangle (15.5,9.15);
  \node at (7.9,8.66)
    {$\mathscr A_{\mathrm{cap}}(\Phi)$ is defined and transmittable};

  \draw[region,draw=ccBlue!55!black,fill=white,line width=0.95pt]
    (0.6,0.6) rectangle (15.2,8.15);
  \node[font=\fontsize{10}{12}\selectfont\bfseries] at (7.9,7.69)
    {A dominating transmittable algebra exists};

  % Two incomparable sufficient classes, with nonempty intersection.
  \fill[ccBlueFill] (5.65,4.15) ellipse (4.65 and 2.65);
  \fill[ccGoldFill] (10.75,4.15) ellipse (3.8 and 2.65);
  \begin{scope}
    \clip (5.65,4.15) ellipse (4.65 and 2.65);
    \fill[ccOverlap] (10.75,4.15) ellipse (3.8 and 2.65);
  \end{scope}

  % High divisibility lies in the algebraic graph class and intersects
  % both Q_0=0 and Q_0>0. Its position does not assume S_Phi itself
  % is closed under multiplication.
  \fill[white,fill opacity=0.55]
    (5.9,3.05) ellipse (3.35 and 1.25);

  \draw[draw=ccBlue,line width=0.85pt]
    (5.65,4.15) ellipse (4.65 and 2.65);
  \draw[draw=ccGold,line width=0.85pt]
    (10.75,4.15) ellipse (3.8 and 2.65);
  \draw[draw=ccPurple,line width=0.85pt,dashed]
    (5.9,3.05) ellipse (3.35 and 1.25);

  \node[text width=5.3cm] at (4.9,5.62)
    {$S_\Phi$ is graph-equivalent\\
     to a unital $*$-algebra};
  \node[text width=3.3cm] at (11.65,5.66)
    {$Q_0(\Phi)=0$\\[-1pt]
     {\fontsize{9}{10.8}\selectfont one-shot zero-error}};

  \node[text=ccPurple!80!black] at (4.9,3.63)
    {Highly divisible};
  \node[font=\fontsize{9}{10.8}\selectfont] at (4.9,3.13)
    {$\Phi=\Psi^\ell,\quad \ell\geq d^2$};

  % Witnesses make the overlap and the non-inclusions explicit.
  \fill[ccPurple!80!black] (4.9,2.72) circle (1.25pt);
  \node[witness,anchor=north] at (4.9,2.62)
    {Identity channel, $d\geq2$};

  \fill[ccPurple!80!black] (8.15,3.68) circle (1.25pt);
  \node[witness,anchor=north,text width=1.65cm] at (8.15,3.54)
    {Complete\\dephasing};

  \fill[ccGold!70!black] (11.6,3.83) circle (1.25pt);
  \node[witness,anchor=north,text width=3.0cm] at (11.6,3.68)
    {Classical channel\\with graph $C_5$};

  \node[witness] at (7.9,7.13)
    {Further dominating examples lie outside both structural classes};

  % A separate chain avoids making unsupported geometric claims
  % about the overlaps of these additional sufficient classes.
  \node[font=\fontsize{9}{10.8}\selectfont] at (7.9,-0.80)
    {$\{\text{classical}\}\subseteq
      \{\text{entanglement-breaking}\}\subseteq
      \{\text{PPT}\}\subseteq
      \{Q_0=0\}$};
\end{tikzpicture}